\documentclass[a4paper,final,11pt]{article}
\usepackage{amsfonts,amsmath,amsthm,amssymb,xcolor,colortbl}
\usepackage{mathtools}
\usepackage{enumitem}
\usepackage{calrsfs}
\DeclareMathAlphabet{\pazocal}{OMS}{zplm}{m}{n}

\usepackage[margin=1in]{geometry}
\usepackage[onehalfspacing]{setspace}
\usepackage{soul}
\usepackage{natbib}
\usepackage{clipboard}	
\usepackage{graphicx}
\newclipboard{PaperClipboard}	
\usepackage{xr-hyper} 
\usepackage[colorlinks,linkcolor=links,citecolor=cites,urlcolor=MyDarkBlue]{hyperref}
\usepackage{cleveref}
\usepackage{tikz}
\usepackage{bbm}
\usepackage[justification=centering]{caption}
\usepackage[justification=centering]{subfig}
\usepackage[section]{placeins}
\usepackage{parskip}
\usepackage{nicefrac}
\usepackage{appendix}
\usepackage[multiple]{footmisc}
\definecolor{MyDarkBlue}{rgb}{0,0.08,0.45}
\definecolor{cites}{HTML}{324b13}
\definecolor{links}{HTML}{1a663b}
\definecolor{MyLightMagenta}{cmyk}{0.1,0.8,0,0.1}
\definecolor{scyan}{HTML}{CBEAFC}
\definecolor{red}{HTML}{B5595C}
\definecolor{green}{HTML}{609B57}
\definecolor{spink}{HTML}{FFB0FF}
\definecolor{yellow}{HTML}{E5A919}
\definecolor{gray}{HTML}{c8c8c8}
\definecolor{darkgray}{HTML}{646464}
\definecolor{blue}{RGB}{0,114,178}
\usetikzlibrary{positioning}
\usetikzlibrary{snakes}
\usetikzlibrary{calc}
\usetikzlibrary{arrows}
\usetikzlibrary{decorations.markings}
\usetikzlibrary{shapes.misc}
\usetikzlibrary{matrix,shapes,fit,tikzmark}
\usetikzlibrary{patterns}
\usetikzlibrary{trees,calc,positioning,snakes}
\usetikzlibrary{decorations.pathmorphing}
\usetikzlibrary{decorations.shapes}
\usepackage{pgfplots}
\pgfplotsset{width=7cm,compat=1.16}
\pgfplotsset{ytick style={draw=none}}
\pgfplotsset{ztick style={draw=none}}
\pgfplotsset{xtick style={draw=none}}

\usepgfplotslibrary{ternary} 
\usepgfplotslibrary{fillbetween}

\newcommand{\horizfig}[2][]{%
  \begin{minipage}{3in}\subfloat[#1]{#2}\end{minipage}}
\newtheorem{theorem}{Theorem}
\newtheorem{lemma}{Lemma} 
\newtheorem{proposition}{Proposition}
 
\newtheorem{example}{Example}
\newtheorem{definition}{Definition} 
\newtheorem{corollary}{Corollary}
\newtheorem{assumption}{Assumption}
\newtheorem{remark}{Remark}
\newtheorem{observation}{Observation}

\newcommand{\aaction}{\ensuremath{a}}
\newcommand{\aactionb}{\ensuremath{\aaction^\prime}}
\newcommand{\aactionbb}{\ensuremath{\aaction^{\prime\prime}}}

\newcommand{\Actions}{\ensuremath{A}}
\newcommand{\Actionsb}{\ensuremath{B}}
\newcommand{\Actionsc}{\ensuremath{C}}
\newcommand{\payoff}{\ensuremath{u}}

\newcommand{\payoffc}{\ensuremath{\hat{\payoff}}}
\newcommand{\type}{\ensuremath{\omega}}
\newcommand{\typeb}{\ensuremath{\type^\prime}}
\newcommand{\Types}{\ensuremath{\Omega}}
\newcommand{\Posteriors}{\ensuremath{\Delta(\Types)}}

\newcommand{\belief}{\ensuremath{\mu}}
\newcommand{\beliefv}{\ensuremath{\boldsymbol{\belief}}}
\newcommand{\beliefvb}{\ensuremath{\tilde{\beliefv}}}

\newcommand{\prior}{\ensuremath{\belief_0}}
\newcommand{\priorb}{\ensuremath{\prior^\prime}}
\newcommand{\priorv}{\ensuremath{\boldsymbol{\prior}}}
\newcommand{\marginalbase}{\ensuremath{\nu}}
\newcommand{\marginal}{\ensuremath{\marginalbase_0}}
\newcommand{\marginalv}{\ensuremath{\boldsymbol{\marginal}}}

\newcommand{\marginalb}{\ensuremath{\bar{\marginalbase}_0}}
\newcommand{\joint}{\ensuremath{\pi}}

\newcommand{\reals}{\ensuremath{\mathbb{R}}}
\newcommand{\realsa}{\ensuremath{\reals^\Actions}}
\newcommand{\realst}{\ensuremath{\reals^\Types}}
\newcommand{\Reals}{\ensuremath{\mathbb{R}}}

\newcommand{\ext}{\ensuremath{\mathrm{ext}}}

\newcommand{\bsplit}{\ensuremath{\tau}}
\newcommand{\base}{\ensuremath{G}}

\newcommand{\demandbase}{\ensuremath{\delta}}
\newcommand{\demand}{\ensuremath{\demandbase_{(\marginal,\bsplit)}}}
\newcommand{\graph}{\ensuremath{\Gamma}}
\newcommand{\graphp}{\ensuremath{\graph_P}}
\newcommand{\strat}{\ensuremath{\alpha}}

\newcommand{\simplex}{\ensuremath{\Delta}}
\newcommand{\opt}{\ensuremath{\simplex_\payoff^*}}
\newcommand{\opta}{\ensuremath{\opt(\aaction)}}
\newcommand{\optab}{\ensuremath{\opt(\aactionb)}}

\newcommand{\normal}{\ensuremath{N}}
\newcommand{\direction}{\ensuremath{p}}

\newcommand{\marginalp}{\ensuremath{\overline{\marginalbase}_0}}
\newcommand{\rnpayoff}{\ensuremath{\payoffc}}

\newcommand{\bce}{\ensuremath{\mathrm{BCE}}}

\newcommand{\bceu}{\ensuremath{\bce\left(\payoff\right)}}
\newcommand{\bcebase}{\ensuremath{\bce\left(\base\right)}}
\newcommand{\scr}{\ensuremath{\sigma}}
\newcommand{\nplayers}{\ensuremath{N}}
\newcommand{\setplayers}{\ensuremath{\left[\nplayers\right]}}
\newcommand{\playerindex}{\ensuremath{n}}
\newcommand{\playeri}{\ensuremath{n}}
\newcommand{\pair}{\ensuremath{(\prior,\marginal)}}
\newcommand{\marginali}{\ensuremath{\marginalbase_{0,\playerindex}}}
\newcommand{\Actionsi}{\ensuremath{\Actions_\playerindex}}
\newcommand{\Actionsmi}{\ensuremath{\Actions_{-\playerindex}}}
\newcommand{\payoffi}{\ensuremath{\payoff_\playerindex}}

\newcommand{\actioni}{\ensuremath{\aaction_\playerindex}}
\newcommand{\actionbi}{\ensuremath{\aactionb_\playerindex}}

\newcommand{\actionmi}{\ensuremath{\aaction_{-\playerindex}}}

\newcommand{\aindex}{\ensuremath{j}}
\newcommand{\tindex}{\ensuremath{i}}
\newcommand{\anum}{\ensuremath{J}}
\newcommand{\tnum}{\ensuremath{I}}

\newcommand{\price}{\ensuremath{p}}
\newcommand{\pricev}{\ensuremath{\boldsymbol{\price}}}
\newcommand{\pricea}{\ensuremath{q}}
\newcommand{\priceav}{\ensuremath{\boldsymbol{\pricea}}}
\newcommand{\mult}{\ensuremath{\lambda}}
\newcommand{\multb}{\ensuremath{\mult^\prime}}

\newcommand{\multp}{\ensuremath{\mult^\uparrow}}
\newcommand{\multm}{\ensuremath{\mult^\downarrow}}
\newcommand{\multv}{\ensuremath{\boldsymbol{\mult}}}
\newcommand{\multvp}{\ensuremath{\multv^\uparrow}}
\newcommand{\multvm}{\ensuremath{\multv^\downarrow}}

\newcommand{\carda}{\ensuremath{|\Actions|}}
\newcommand{\cardt}{\ensuremath{|\Types|}}

\newcommand{\priceup}{\ensuremath{\price^\uparrow}}
\newcommand{\pricedown}{\ensuremath{\price^\downarrow}}
\newcommand{\pricevup}{\ensuremath{\pricev^\uparrow}}
\newcommand{\pricevdown}{\ensuremath{\pricev^\downarrow}}
\newcommand{\priceaup}{\ensuremath{\pricea^\uparrow}}
\newcommand{\priceadown}{\ensuremath{\pricea^\downarrow}}
\newcommand{\priceavup}{\ensuremath{\priceav^\uparrow}}
\newcommand{\priceavdown}{\ensuremath{\priceav^\downarrow}}
\newcommand{\diff}{\ensuremath{d}}

\newcommand{\diffv}{\ensuremath{\boldsymbol{\diff}}}

\newcommand{\dm}{DM}
\newcommand{\minkowski}{\ensuremath{\mathrm{M}}}
\newcommand{\gfset}{\ensuremath{X}}
\newcommand{\gelement}{\ensuremath{v}}
\newcommand{\gelementb}{\ensuremath{\tilde{\gelement}}}
\newcommand{\gbold}{\ensuremath{\boldsymbol{\gelement}}}
\newcommand{\gbbold}{\ensuremath{\boldsymbol{\gelementb}}}
\newcommand{\gboldb}{\ensuremath{\tilde{\gbold}}}
\newcommand{\gsubset}{\ensuremath{V}}
\newcommand{\gsubsetb}{\ensuremath{\gsubset^\prime}}
\newcommand{\vol}{\ensuremath{b}}
\newcommand{\bdirection}{\ensuremath{\boldsymbol{\direction}}}
\newcommand{\bbelief}{\ensuremath{\boldsymbol{\belief}}}

\newcommand{\polytope}{\ensuremath{P}}
\newcommand{\face}{\ensuremath{F}}
\newcommand{\cone}{\ensuremath{C}}
\newcommand{\pmpair}{\ensuremath{(\payoff,\marginal)}}
\newcommand{\ppmpair}{\ensuremath{\minkowski\pmpair}}

\newcommand{\Prices}{\ensuremath{\pazocal{P}}}
\newcommand{\Pricesup}{\ensuremath{\Prices^\uparrow}}
\newcommand{\Pricesdown}{\ensuremath{\Prices^\downarrow}}
\newcommand{\slope}{\ensuremath{\gamma}}
\newcommand{\slopev}{\ensuremath{\boldsymbol{\slope}}}
\newcommand{\ctt}{\ensuremath{\kappa}}
\newcommand{\cttv}{\ensuremath{\boldsymbol{\ctt}}}
\newcommand{\candt}{\ensuremath{\type^\star}}
\newcommand{\aud}{\ensuremath{\mathrm{AUD}}}
\newcommand{\candi}{\ensuremath{\tindex^\star}}
\newcommand{\candti}{\ensuremath{\type_{\candi}}}
\newcommand{\lowd}{\ensuremath{\underline{\diff}}}
\newcommand{\highd}{\ensuremath{\overline{\diff}}}
\newcommand{\lowdj}{\ensuremath{\lowd_{j+1,j}}}
\newcommand{\highdj}{\ensuremath{\highd_{j+1,j}}}
\newcommand{\lowdk}{\ensuremath{\lowd_{k+1,k}}}
\newcommand{\highdk}{\ensuremath{\highd_{k+1,k}}}

\newcommand{\costr}{\ensuremath{c_{\mathrm{I}}}}
\newcommand{\costa}{\ensuremath{c_{\mathrm{II}}}}
\newcommand{\foc}{\ensuremath{\mathrm{FOC}}}
\newcommand{\dx}{\ensuremath{\mathrm{d}}}

\newcommand{\nodes}{\ensuremath{V}}
\newcommand{\node}{\ensuremath{v}}
\newcommand{\edges}{\ensuremath{E}}
\newcommand{\flow}{\ensuremath{f}}
\newcommand{\capacity}{\ensuremath{c}}

\newcommand{\param}{\ensuremath{\theta}}
\newcommand{\paramb}{\ensuremath{\param^\prime}}
\newcommand{\Param}{\ensuremath{\Theta}}

\newcommand{\decision}{\ensuremath{D}}

\newcommand{\dimm}{\ensuremath{d}}
\newcommand{\basis}{\ensuremath{b}}
\newcommand{\basisv}{\ensuremath{\boldsymbol{\basis}}}
\newcommand{\Basis}{\ensuremath{\mathrm{B}}}
\newcommand{\N}{\ensuremath{\pazocal{N}}}
\newcommand{\aff}{\ensuremath{\mathrm{aff}}}

\newcommand{\triple}{\ensuremath{(\pricev,\priceav,\multv)}}
\newcommand{\tripleo}{\ensuremath{(\pricev_1,\priceav_1,\multv_1)}}
\newcommand{\triplet}{\ensuremath{(\pricev_2,\priceav_2,\multv_2)}}

\usepackage{mparhack}

 \usepackage{chngcntr}
\usepackage{apptools}
\AtAppendix{\counterwithin{equation}{section}}
\AtAppendix{\counterwithin{figure}{section}}
\AtAppendix{\counterwithin{table}{section}}
\AtAppendix{\counterwithin{example}{section}}
\AtAppendix{\counterwithin{lemma}{section}}
\AtAppendix{\counterwithin{remark}{section}}
\AtAppendix{\counterwithin{proposition}{section}}
\AtAppendix{\counterwithin{theorem}{section}}
\AtAppendix{\counterwithin{corollary}{section}}
\usepackage{etoolbox}
\makeatletter
\patchcmd{\hyper@makecurrent}{%
    \ifx\Hy@param\Hy@chapterstring
        \let\Hy@param\Hy@chapapp
    \fi
}{%
    \iftoggle{inappendix}{
        \@checkappendixparam{chapter}%
        \@checkappendixparam{section}%
        \@checkappendixparam{subsection}%
        \@checkappendixparam{subsubsection}%
        \@checkappendixparam{paragraph}%
        \@checkappendixparam{subparagraph}%
    }{}%
}{}{\errmessage{failed to patch}}

\newcommand*{\@checkappendixparam}[1]{%
    \def\@checkappendixparamtmp{#1}%
    \ifx\Hy@param\@checkappendixparamtmp
        \let\Hy@param\Hy@appendixstring
    \fi
}
\makeatletter

\newtoggle{inappendix}
\togglefalse{inappendix}

\apptocmd{\appendix}{\toggletrue{inappendix}}{}{\errmessage{failed to patch}}
\apptocmd{\subappendices}{\toggletrue{inappendix}}{}{\errmessage{failed to patch}}
\title{Revealed Information\thanks{The latest version of the paper can be found \href{https://www.dropbox.com/scl/fi/0ewp06nfanzlxyhpg2pgs/revealed-information.pdf?rlkey=conoz38fz5sd1lqjep06wy5wm&dl=0}{here}. We are grateful to Yaron Azrieli, Michael Grubb, Denniz Kattwinkel, Emir Kamenica,  Jan Knoepfle, Shengwu Li, Elliot Lipnowski, Marco Mariotti, Meg Meyer, Stephen Morris, Pietro Ortoleva, Jacopo Perego, Andrea Prat, John Rehbeck, Phil Reny, Vasiliki Skreta, Ming Yang, Joel Watson, and especially Quitze Valenzuela-Stookey and our discussants, Amanda Friedenberg and Cristina Gualdani, for thought-provoking questions and insightful discussions. Seminar audiences at Stony Brook, NYU, LACEA/LAMES, UPenn, SAET McKenzie Lecture, UCLA, CalPoly, University of Wisconsin, CEPR Virtual IO seminar, WIET, UCL, QMUL, Oxford, Princeton, Georgetown, University of Toronto, UPF, and UAB gave us thoughtful feedback for which we are thankful. Ran Eilat acknowledges support from the BSF under grant~\#2022294. Finally, we thank the Sloan Foundation for financial support.}}
\author{Laura Doval\thanks{Economics Division, Columbia Business School and CEPR. E-mail: \href{mailto:laura.doval@columbia.edu}{\texttt{laura.doval@columbia.edu}}}\and Ran Eilat\thanks{Department of Economics, Ben Gurion University of the Negev. E-mail: \href{mailto:eilatr@bgu.ac.il}{\texttt{eilatr@bgu.ac.il}}}\and Tianhao Liu\thanks{Department of Economics, Columbia University. E-mail: \href{mailto:tl3014@columbia.edu}{\texttt{tl3014@columbia.edu}}}\and Yangfan Zhou\thanks{Department of Economics, Columbia University. E-mail: \href{mailto:yz3905@columbia.edu}{\texttt{yz3905@columbia.edu}}}}
\begin{document}
\pagenumbering{gobble}
\maketitle
\begin{abstract}
%
An analyst observes the frequency with which a decision maker (\dm) takes actions, but not the frequency conditional on payoff-relevant states. We ask when the analyst can rationalize the \dm's choices \emph{as if} the \dm\ first learns something about the state before acting. We provide a support-function characterization of the triples of utility functions, prior beliefs, and (marginal) distributions over actions such that the \dm's action distribution is consistent with information given the agent's prior and utility function. Assumptions on the cardinality of the state space and the utility function allow us to refine this characterization, obtaining a sharp system of finitely many inequalities the utility function, prior, and action distribution must satisfy. We apply our characterization to study comparative statics and to identify conditions under which a single information structure rationalizes choices across multiple decision problems. We characterize the set of distributions over posterior beliefs that are consistent with the \dm's choices. We extend our results to settings with a continuum of actions and states assuming the first-order approach applies, and to simple multi-agent settings.%
\end{abstract}
\newpage
\clearpage
\pagenumbering{arabic}

\section{Introduction}\label{sec:intro}

When economic agents make decisions under uncertainty, they rely on information about unknown factors. Yet, researchers rarely observe this information directly, creating a fundamental challenge: How can we infer what agents know from their observed choices? This challenge has significant implications for economic modeling, because assumptions about information can dramatically impact model predictions and parameter estimates.

In this paper, we provide a framework to determine when observed choice patterns can be rationalized by some information structure, without requiring the researcher to observe the relationship between choices and underlying states. Specifically, given a decision maker's (\dm) utility function, prior beliefs, and an observed distribution of actions, we characterize when this action distribution can be explained as the result of the \dm\ optimally responding to some information about the state. Our main contribution is a support-function characterization that translates this question into a finite system of inequalities involving the utility function, prior, and action distribution.

Consider a concrete example: an analyst studying whether a judge's bail decisions are informed by recidivism risk \citep{rambachan2022identifying}. The analyst observes only the frequency with which the judge grants bail, not the frequency with which the judge grants bail conditional on whether the defendant will recidivate. Our framework allows the analyst to determine which combinations of the judge's utility function and prior beliefs about recidivism would make the observed bail decisions consistent with the judge having some information about recidivism risk.

Many recent empirical studies seek preference estimates robust to informational assumptions, recognizing how strongly these assumptions affect outcomes. For instance, \cite{dickstein2018exporters}, \cite*{dickstein2024patient}, \cite{gualdani2019identification}, and \cite{rambachan2022identifying} develop methods to understand the role of information in firms' export decisions, physicians' treatment recommendations, voter choices, and prediction mistakes, respectively. Similar approaches appear in multi-agent settings such as auctions \citep*{syrgkanis2017inference} and entry games \citep{magnolfi2019estimation}.

Some of these approaches rely on Bayes correlated equilibrium (\bce), developed by \cite{bergemann2016bayes} for games and \cite{kamenica2011bayesian} for single-agent settings. Given a payoff structure\textemdash players' utility functions and their common prior over states\textemdash\bce\ provides conditions under which an outcome distribution can be rationalized as if players had access to information before playing. Importantly, checking whether an outcome is a \bce\ requires verifying only a finite system of linear inequalities.

However, a critical gap exists between what \bce\ requires and the data empirical researchers typically have available. \bce\ presumes the analyst observes the joint distribution over payoff-relevant states and action profiles. In our judge example, \bce\ would require observing bail decisions conditional on whether defendants would recidivate\textemdash data that are rarely available. In practice, analysts typically observe only the marginal distribution of actions (the frequency of bail grants), not the action distribution conditional on the state (the frequency of bail grants conditional on recidivism risk). For a given payoff structure, rather than checking finitely many linear inequalities, the analyst checks whether a \bce\ \emph{exists} whose marginals over the actions matches the observed choices.

Our paper bridges this gap in single-agent settings by characterizing when marginal distributions over states and actions are consistent with a (single-agent) \bce\ given the \dm's utility function. Formally,  given the triple of a utility function, prior beliefs, and observed action distribution, we study when a \bce\ exists whose marginals over the states and actions coincide with the \dm's prior and action distribution. When such a \bce\ exists, we say the marginals\textemdash the \dm's prior and the observed action distribution\textemdash are \bce-consistent given the utility function.

Our main contributions are in Sections \ref{sec:main} and \ref{sec:utility}. In \autoref{theorem:h-representation}, we provide a characterization in terms of a system of finitely many inequalities, linear in both marginal distributions, such that the marginals are \bce-consistent if and only if these inequalities hold. For a given action distribution and utility function, these inequalities characterize the support function of the set of priors that make the \dm's choices consistent with information. Although  \autoref{theorem:h-representation} precisely identifies the finitely many inequalities that must hold for the marginals to be \bce-consistent, the characterization is rather implicit in that it does not describe them in closed form. Our remaining results consider assumptions on the cardinality of the set of states or on the utility function under which we provide closed-form characterizations of the system of finitely many inequalities in \autoref{theorem:h-representation}.

 \autoref{proposition:2-3} provides a closed-form characterization when the state space has at most three elements.  As we explain in the main text, the inequalities in \autoref{proposition:2-3} are always a subset of those in \autoref{theorem:h-representation}, and they can be readily expressed in terms of primitives even when the state space has more than three elements. Thus, they can be used to rule out pairs of prior beliefs and action distributions that are not \bce-consistent given the utility function.
 
 Theorems \ref{proposition:aud} and \ref{proposition:2-step} offer parallel characterizations for utility functions with affine and two-step differences (Definitions \ref{definition:aud} and \ref{definition:2-step}), respectively. For instance, all binary decision problems and the discrete analog of quadratic loss have affine differences, whereas the discrete analog of absolute error loss has two-step differences. Both results rest on \autoref{proposition:mcv}, which narrows the search for halfspaces defining the set of \bce-consistent distributions for utility functions satisfying increasing differences and concavity in actions (the discrete counterpart to first-order approach conditions). In \autoref{sec:foa}, we further generalize this characterization to compact Polish action and state spaces under the first‑order approach (cf. \citealp*{kolotilin2023persuasion}).

In \autoref{sec:applications}, we apply our results to study comparative statics and \bce-consistency \emph{across} decision problems. \autoref{sec:cs} applies \autoref{proposition:aud} to study comparative statics in the set of \bce-consistent marginals when considering changes to the \dm's prior or utility function. Building on results in \cite*{bergemann2022counterfactuals}, we show in \autoref{sec:across} how our results can be used to study whether a single information structure exists that rationalizes a \dm's choices across different decision problems. In \autoref{sec:games}, we apply our results to study \bce-consistency in simple games.

Because our characterization results are not constructive, we study in \autoref{sec:core} which information structures make the marginals \bce-consistent. In \autoref{proposition:core-bp}, we characterize the Bayes plausible distributions over posteriors that implement a given action distribution, interpreting \bce-consistency as a market-clearing condition in a persuasion economy and building on \cite{gale1957theorem}. 
%
%

Our results have implications for both empirical work and theoretical analysis. Empirically, for a given action distribution and utility function, our characterization results non-parametrically identify the set of priors such that the prior and action distribution are \bce-consistent given the utility function, which is useful whenever the analyst has no information on what the prior should be, but may have auxiliary data on the \dm’s payoffs.\footnote{By contrast, studies of risk aversion across domains assume the \dm's beliefs about expected claim rates coincide with the frequencies in the data and estimate the curvature of the utility function \citep*{cohen2007estimating, barseghyan2011risk,barseghyan2013nature}.} Furthermore, for a given action distribution, our results characterize joint restrictions on the prior and the utility function for the triple of the utility function, prior, and action distribution to be consistent with information.  

Our framework also has applications in behavioral economics, particularly for cognitive uncertainty models where $\mathrm{DM}$s exhibit random behavior across instances of the same problem (e.g., \citealp*{khaw2021cognitive}; \citealp{enke2023cognitive}). In these models, the state often represents the correct action and the \dm\ has a noisy perception of this state. Whereas laboratory experiments may generate state-dependent choice data, outside the lab, analysts typically only observe average choices. Our results can test whether behavior is consistent with \emph{Bayesian} cognitive uncertainty\textemdash whether noisy perception of states can be rationalized via an information structure.

Theoretically, our results open up the study of marginal information design\textemdash akin to reduced-form implementation in mechanism design\textemdash where an information designer cares only about the \dm's actions, and not the state of the world. From this perspective, results such as \autoref{proposition:mcv} reveal the structure of the binding constraints in information design problems, and we expect it can be used to further the study of Bayesian persuasion.

\paragraph{Related Literature} Our analysis relates to several strands of literature. \cite{lu2016random}, \cite{rehbeck2023revealed}, \cite{de2022rationalizing}, and  \cite{azrieli2022marginal} study related rationalization problems but differ in either the available data or characterization approach. In \cite{lu2016random}, the analyst observes the \dm's stochastic choice out of every possible menu. As in our paper, the analyst in \cite{rehbeck2023revealed} and \cite{de2022rationalizing} observes the \dm's stochastic choice from a single menu in a static and dynamic decision problem, respectively. Unlike our paper, their characterization results are in terms of the non-existence of a (possibly randomized) deviation, akin to Pearce's lemma. In \cite{azrieli2022marginal}, the analyst observes the distribution of menus the \dm\ faces and the \dm's distribution of choices, but not the distribution of choices from each menu. Despite the difference in the settings, we discuss how \autoref{proposition:core-bp}, which we obtain relying on \cite{gale1957theorem}, can be obtained using their results.

A literature in decision theory and experimental economics studies when state-dependent stochastic choice data can be rationalized via costly information acquisition and whether the data identify the information acquisition costs (e.g., \citealp*{caplin2015revealed,caplin2017rationally,chambers2020costly,dewan2020estimating,denti2022posterior,caplin2023rationalizable}). Like we do, many of these papers provide results for a given utility function.\footnote{To be sure,  \cite{caplin2023rationalizable} consider recovering the utility function.}  Unlike our paper, the \dm's prior is observed in the data. Relatedly, \cite{ergin2010unique} and \cite{dillenberger2014theory,dillenberger2023subjective} study when menu choice data is consistent with costly information acquisition.

A literature in information design studies problems with (given) marginals. \cite*{arieli2021feasible} and \cite{morris2020no} characterize joint distributions over posterior beliefs that are consistent with some information structure with binary and finitely many states, respectively. Assuming the sender and receiver care only about the posterior mean of the states, \cite*{toikka2022bayesian} show the sender's problem is a linear programming problem that only depends on the marginal distribution over actions. \cite{kolotilin2023persuasion} characterize properties of optimal information structures assuming the first-order approach applies in a large class of persuasion problems with nonlinear sender preferences. \cite{strack2024privacy} show the optimization over privacy-preserving signals can be cast as an optimal transport problem. 

Methodologically, our paper relates to the econometrics literature on random sets for partial identification, where support functions are used to study the \emph{Aumann expectation} of a random set (e.g., \citealp*{galichon2011set,beresteanu2011sharp,molchanov2018random}). Our extension to compact Polish action and state spaces relies on \citet[Theorem 3]{strassen1965existence}, which also appears in \cite{galichon2011set}.

In lieu of an organizational paragraph, we collect here mathematical notation and definitions used throughout the paper:
\paragraph{Mathematical conventions and definitions} For a finite set \gfset, we denote by $\reals^\gfset$ the set of vectors of length $|\gfset|$. Depending on context, we refer to elements of $\reals^\gfset$ either as functions from $\gfset$ to the reals or as vectors in $\reals^\gfset$. We reserve $\gelement$ (serif) for the function $\gelement:\gfset\rightarrow\reals$ and $\gbold$ (boldface) for the vector in $\reals^\gfset$. When we wish to emphasize the length of \gbold\ we write $\gbold\in\reals^{|\gfset|}$. If $\gbold,\gboldb$ are two vectors in $\reals^\gfset$, we denote by $\gbold\gboldb$ their inner product, $\sum_{i=1}^{|\gfset|}\gbold_\tindex\gboldb_\tindex$. 

Given two nonempty subsets $\gsubset,\gsubsetb\subset\reals^\gfset$, their \emph{Minkowski sum} is the set $\gsubset+\gsubsetb=\{\gbold+\gboldb:\gbold\in\gsubset,\gboldb\in\gsubsetb\}$. Given a set $\gsubset\subset\reals^\gfset$, the support function of \gsubset\ is the mapping $\pricev\mapsto\sup\{\pricev\gbold:\gbold\in V\}$. A (convex) cone \cone\ is a subset of $\reals^\gfset$ that is closed under addition and non-negative scalar multiplication. A vector $\boldsymbol{w}\in\cone$ is an \emph{extreme ray} if no linearly independent $\gbold,\gboldb\in\cone$ and positive scalars $\lambda,\gamma$ exist such that $\boldsymbol{w}=\lambda\gbold+\gamma\gboldb$. Note that if $\gbold\in\cone$ is an extreme ray, so is $\lambda\gbold$ for $\lambda>0$. When we refer to an extreme ray, we refer to one representative of this equivalence class. Finally, a polytope is a bounded subset of $\reals^\gfset$ defined as the intersection of finitely many halfspaces of the form $\bdirection_l\gbold\leq\vol_l$ for some $(\bdirection_l,\vol_l)\in\reals^{\gfset}\times\reals$, $l\in\{1,\dots,L\}$. 


\section{Model}\label{sec:model}
\paragraph{A decision problem with given marginals} Our  model considers a  \dm\ taking an action under uncertainty about a state of the world. We denote by $\Types=\{\type_1,\dots,\type_\tnum\}$ the finite set of states of the world and by $\Actions=\{\aaction_1,\dots,\aaction_{\anum}\}$ the finite set of actions. The \dm's utility function $\payoff:\Actions\times\Types\rightarrow\reals$ describes the \dm's payoff as a function of the action she takes and the state of the world. The tuple $\langle\Types,\Actions,\payoff\rangle$ defines the decision problem. 

Letting \Posteriors\ denote the set of distributions over \Types, the \dm's utility defines the set of beliefs for which a given action $\aaction\in\Actions$ is optimal, \opta, as follows:
\begin{align}\label{eq:opt-a}
\opta=\{\belief\in\Posteriors:(\forall\aactionb\in\Actions)\sum_{\type\in\Types}\belief(\type)\left(\payoff(\aaction,\type)-\payoff(\aactionb,\type)\right)\geq0\}.
\end{align}

Below, we regard \Posteriors\ as a full-dimensional set in $\reals^{\tnum-1}$. Furthermore, to streamline the presentation, we implicitly assume each nonempty \opta\ is full dimensional in $\reals^{\tnum-1}$.\footnote{Because $\cup_{\aaction\in\Actions}\opta=\Posteriors$, at least one \opta\ is full dimensional in $\reals^{\tnum-1}$. Assuming all nonempty \opta\ are full-dimensional sets in $\reals^{\tnum-1}$ simplifies exposition.} \autoref{appendix:main} deals with the general case.\footnote{Our results remain the same, but properly stating them requires defining the embedding of a lower-dimensional subset  into $\reals^{\tnum-1}$.} Thus, whenever we refer to the dimension of a set, we mean its dimension in $\reals^{\tnum-1}$, even if our notation does not make it explicit.


We take the viewpoint of an analyst who knows the decision problem, but not whether the \dm\ has access to further information before taking her action. The analyst also observes the \dm's distribution over actions, $\marginal\in\Delta(\Actions)$. The analyst's goal is to determine for which prior distributions over the states,  $\prior\in\Posteriors$, \marginal\ can be rationalized as the result of the  \dm\ optimally choosing her actions after observing the outcome of an information structure. 

For a given prior distribution $\prior\in\Posteriors$, the results in \cite{myerson1982optimal}, \cite{kamenica2011bayesian}, and \cite{bergemann2016bayes} imply an information structure exists that rationalizes \marginal\ given the utility function \payoff\ if and only if the pair of distributions \pair\ satisfy the following:
\begin{definition}[\bce-consistency given \payoff]\label{definition:bce-c} The pair of distributions, $\pair\in\Posteriors\times\Delta(\Actions)$, is \bce-consistent given \payoff\ if a joint distribution $\joint\in\Delta(\Actions\times\Types)$ exists such that the following hold:
\begin{align}
(\forall\aaction\in\Actions)(\forall\aactionb\in\Actions)&\sum_{\type\in\Types}\joint(\aaction,\type)\left[\payoff(\aaction,\type)-\payoff(\aactionb,\type)\right]\geq0\tag{O}\label{eq:obedience}\\
(\forall\type\in\Types)&\sum_{\aaction\in\Actions}\joint(\aaction,\type)=\prior(\type)\tag{M$_{\prior}$}\label{eq:state-mg}\\
(\forall\aaction\in\Actions)&\sum_{\type\in\Types}\joint(\aaction,\type)=\marginal(\aaction).\tag{M$_{\marginal}$}\label{eq:action-mg}
\end{align}
We denote by \bceu\ the set of \bce-consistent pairs \pair\ given \payoff.
\end{definition}
In words, an information structure exists that rationalizes the \dm's choices \marginal\ given her utility function \payoff\ and prior belief \prior\ if and only if a joint distribution $\joint\in\Delta(\Actions\times\Types)$ exists that satisfies Equations \ref{eq:obedience}, \ref{eq:state-mg}, and \ref{eq:action-mg}. \autoref{eq:obedience} states that if the \dm\ knows \aaction\ has been drawn according to \joint\textemdash but not the state\textemdash the \dm\ finds \aaction\ optimal. Equations \ref{eq:state-mg} and \ref{eq:action-mg} state the joint distribution is consistent with the pair \pair: the generated information averages out to the prior \eqref{eq:state-mg}, and the \dm's average choices coincide with \marginal\ \eqref{eq:action-mg}. 

With \autoref{definition:bce-c} at hand, we can now state the analyst's problem formally: given \marginal\ and \payoff, the analyst seeks to characterize the set of prior distributions \prior\ such that $\pair\in\bceu$. We denote the set of all such priors by \ppmpair; that is,
\begin{align}
\ppmpair=\{\prior\in\Posteriors:\pair\in\bceu\}.
\end{align}
Taking \marginal\ as given, the question of whether \pair\ is \bce-consistent given \payoff\ is equivalent to whether \prior\ belongs in \ppmpair.

\paragraph{The action marginal as a distribution over posteriors} A joint distribution $\joint\in\Delta(\Actions\times\Types)$ with marginals $\pair$ induces a belief system, $\{\belief(\cdot|\aaction)\in\Posteriors:\aaction\in\Actions\}$, describing the \dm's beliefs conditional on action \aaction, which satisfies that for all actions $\aaction\in\Actions$,
\[\marginal(\aaction)\belief(\type|\aaction)=\joint(\aaction,\type).\]
In this case, one can view \marginal\ as a distribution over posteriors and the belief system $\left(\belief(\cdot|\aaction)\right)_{\aaction\in\Actions}$ as its support. Consequently, whether \pair\ is \bce-consistent given \payoff\ is equivalent to whether a belief system $\left(\belief(\cdot|\aaction)\right)_{\aaction\in\Actions}$ exists that satisfies the following. First, for all states $\type\in\Types$,
\begin{align}\label{eq:state-mg-b}\tag{BP$_{\prior}$}
    \sum_{\aaction\in\Actions}\marginal(\aaction)\belief(\type|\aaction)=\prior(\type),
\end{align}
and for all $\aaction,\aactionb\in\Actions$,
\begin{align}\label{eq:obedience-b}\tag{O$_\belief$}
\sum_{\type\in\Types}\marginal(\aaction)\belief(\type|\aaction)\left[\payoff(\aaction,\type)-\payoff(\aactionb,\type)\right]\geq0.
\end{align}
Then, Equations \ref{eq:state-mg-b} and \ref{eq:obedience-b} require that (i) \marginal\ induces a Bayes plausible distribution over posteriors and (ii) for all actions \aaction, the \emph{posterior belief} $\belief(\cdot|\aaction)$ is an element of $\opta$. Under this interpretation,  the action distribution \marginal\ describes the frequency with which inducing beliefs in $\opta$ is necessary. The results in \autoref{sec:main} use the representation of  \bce-consistency given the utility function \payoff\ implied by Equations \ref{eq:state-mg-b} and \ref{eq:obedience-b}.

We close this section with \autoref{remark:analyst}, which compares our approach with that in the empirical work discussed in the introduction. Readers interested in the characterization results can jump to \autoref{sec:main}, with little loss of continuity.

\begin{remark}[]\label{remark:analyst}
Given the set of states and actions, and the \dm's choices, \marginal, a natural question to ask is whether a prior distribution \prior\ and a utility function \payoff\ exist such that the \dm's choices can be rationalized as if the \dm\ had access to information before optimally choosing their actions. We do not pursue this question for two reasons. First, without further restrictions, the answer to that question is trivial: one can always choose \prior\ and \payoff\ such that the \dm\ is indifferent between all actions in the support of \marginal\ at the prior. Second, in empirical work, the analyst typically postulates a finite-dimensional parameter space, \Param, that parameterizes the prior and the utility function. The analyst then seeks to characterize the set of parameters \param\ such that the pair $(\prior^\param,\marginal)$ is \bce-consistent given $\payoff^\param$. Our approach is similar, but  our results show the analyst only needs to parameterize the utility function: for each utility function \payoff, our results non-parametrically identify the set of priors \prior\ such that \pair\ is \bceu\ (if any). The analyst can then repeat this exercise for each $\param$ in their parameter space.
\end{remark}

\section{Basic characterization}\label{sec:main}
In this section, we introduce our basic characterization result, \autoref{theorem:h-representation}. 

\paragraph{A Minkowski-sum representation} 

For a fixed utility function \payoff, Equations \ref{eq:state-mg-b} and \ref{eq:obedience-b} allow us to immediately characterize the set \ppmpair\ as the \marginal-weighted \emph{Minkowski sum} of the sets \opta. That is, we claim
\begin{align}\label{eq:minkowski-representation}\tag{MS}
\ppmpair=\sum_{\aaction\in\Actions}\marginal(\aaction)\opta,
\end{align}
Hence, $\pair\in\bceu$ if and only if \prior\ is in the \marginal-weighted Minkowski sum of the sets \opta. That any prior in \ppmpair\ is an element of $\sum_{\aaction\in\Actions}\marginal(\aaction)\opta$ follows from  Equations \ref{eq:state-mg-b} and \ref{eq:obedience-b}. Conversely, consider a collection of beliefs $\{\belief(\cdot|\aaction)\in\Posteriors:\aaction\in\Actions\}$ such that for all $\aaction\in\Actions$, $\belief(\cdot|\aaction)\in\opta$. This collection together with the prior $\priorb=\sum_{\aaction\in\Actions}\marginal(\aaction)\belief(\cdot|\aaction)$ satisfy Equations \ref{eq:state-mg-b} and \ref{eq:obedience-b}; hence, the pair $(\priorb,\marginal)$ is \bce-consistent given \payoff. We illustrate this and other results in this section with \autoref{example:match}:


\begin{example}\label{example:match}
A \dm\ wants to match their action to the state of the world. We assume $\Types=\{\type_1,\type_2,\type_3\}$ and $\Actions=\{\aaction_1,\aaction_2,\aaction_3\}$. The \dm's utility function is then $\payoff(\aaction_j,\type_i)=\mathbbm{1}[i=j]$. In \autoref{fig:ex-match}, we depict the simplex as a full-dimensional subset of $\reals^2$, with the probabilities of $\type_1$, $\belief_1$, and of $\type_2$, $\belief_2$, in the $x$- and $y$-axis, respectively. In \autoref{fig:match}, the shaded areas represent the beliefs for which a given action is optimal (cf. \autoref{eq:opt-a}): $\aaction_1$ is optimal in the blue region, $\aaction_2$ is optimal in the red region, and $\aaction_3$ is optimal in the orange one. In \autoref{fig:ms-match}, the polytope delineated in black depicts the Minkowski sum of the sets \opta\ in \autoref{fig:match} weighted by the uniform action marginal, $\marginal=(1/3,1/3,1/3)$, and hence the set of all prior distributions consistent with the uniform action marginal in this example.
\begin{figure}[t!]
\centering
\subfloat[The sets \opta.]{\scalebox{0.65}{%
\begin{tikzpicture}
\begin{axis}[xmin=0,xmax=1,ymin=0,ymax=1,xticklabels={},yticklabels={},xlabel=$\mu_1$,ylabel=$\mu_2$,x label style={at={(axis description cs:1,-0.01)}},
    y label style={at={(axis description cs:-0.01,1)},rotate=-90}
    ,width=9cm,height=9cm]
    \addplot[name path=axis,forget plot,domain=0:1]{0};
    \addplot[name path=F,forget plot,draw=red!50,domain=0:0.5]{1-x};
        \addplot[name path=F1,forget plot,draw=blue!50,domain=0.5:1]{1-x};
    \addplot[name path=A0,forget plot,draw=blue!50,domain=1/3:0.5]{x};
    \addplot[name path=B0,forget plot,draw=blue!50, domain=1/3:0.5]{1-2*x};
    \addplot[blue!50,forget plot]fill between[of= A0 and B0, soft clip={domain=1/3:0.5}];
     \addplot[blue!50]fill between[of= F1 and axis, soft clip={domain=0.5:1}];
            \addlegendentry{$\Delta_u^*(a_1)$}
       \addplot[name path=C0,forget plot,draw=red!30, domain=0:1/3]{(1-x)/2};
    \addplot[red!30,forget plot]fill between[of= F and C0, soft clip={domain=0:1/3}];
        \addplot[red!30]fill between[of= F and A0, soft clip={domain=1/3:0.5}];
                    \addlegendentry{$\Delta_u^*(a_2)$}
        \addplot[orange!30,forget plot]fill between[of=C0 and axis, soft clip={domain=0:1/3}];
                \addplot[orange!30]fill between[of=B0 and axis, soft clip={domain=1/3:0.5}];
                            \addlegendentry{$\Delta_u^*(a_3)$}
\end{axis}
\end{tikzpicture}}\label{fig:match}
}\hfill
\subfloat[The \marginal-weighted Minkowski sum of the sets \opta.]{
\scalebox{0.65}{%
\begin{tikzpicture}
\begin{axis}[xmin=0,xmax=1,ymin=0,ymax=1,xticklabels={},yticklabels={},xlabel=$\mu_1$,ylabel=$\mu_2$,x label style={at={(axis description cs:1,-0.01)}},
    y label style={at={(axis description cs:-0.01,1)},rotate=-90}
    ,width=9cm,height=9cm]
    \addplot[name path=axis,forget plot,domain=0:1]{0};
    \addplot[name path=F,forget plot,draw=red!50,domain=0:0.5]{1-x};
        \addplot[name path=F1,forget plot,draw=blue!50,domain=0.5:1]{1-x};
    \addplot[name path=A0,forget plot,draw=blue!50,domain=1/3:0.5]{x};
    \addplot[name path=B0,forget plot,draw=blue!50, domain=1/3:0.5]{1-2*x};
    \addplot[blue!50,forget plot]fill between[of= A0 and B0, soft clip={domain=1/3:0.5}];
     \addplot[blue!50]fill between[of= F1 and axis, soft clip={domain=0.5:1}];
            \addlegendentry{$\Delta_u^*(a_1)$}
       \addplot[name path=C0,forget plot,draw=red!30, domain=0:1/3]{(1-x)/2};
    \addplot[red!30,forget plot]fill between[of= F and C0, soft clip={domain=0:1/3}];
        \addplot[red!30]fill between[of= F and A0, soft clip={domain=1/3:0.5}];
                    \addlegendentry{$\Delta_u^*(a_2)$}
        \addplot[orange!30,forget plot]fill between[of=C0 and axis, soft clip={domain=0:1/3}];
                \addplot[orange!30]fill between[of=B0 and axis, soft clip={domain=1/3:0.5}];
                            \addlegendentry{$\Delta_u^*(a_3)$}
                \addplot[name path=A1, draw=black!25,domain=0:1]{8/9-x};
\addplot[name path=B1, draw=black!25,domain=0:8/9]{1/9};
\draw[name path=C1, black!25](1/9,0)--(1/9,8/9);
\addplot[black!25,opacity=0.75]fill between [of= A1 and B1, soft clip={domain=1/9:7/9}];
\addplot[name path=A0,draw=black!50,domain=0:1]{1.5-2*x};
\addplot[name path=A,draw=black!50,domain=0:1]{0.75-0.5*x};
\addplot[name path=B,draw=black!50,domain=0:1]{1/4-0.5*x};
\addplot[name path=C,draw=black!50,domain=0:1]{1/2-2*x};
\addplot[name path=D,draw=black!50,domain=0:1]{1/2+x};
\addplot[name path=E,draw=black!50,domain=0:1]{x-1/2};
\addplot[black!50,opacity=0.75] fill between[of=A and B, soft clip={domain=1/6:0.5}];
\addplot[black!50,opacity=0.75] fill between[of=D and C, soft clip={domain=0:1/6}];
\addplot[black!50,opacity=0.75] fill between[of=E and A0, soft clip={domain=0.5:2/3}];
\addplot[thick,black,domain=1/6:5/18]{0.75-0.5*x};
\addplot[thick,black,domain=1/9:1/6]{1/2+x};
\addplot[thick,black] coordinates {(1/9,11/18) (1/9,5/18)};
\addplot[thick,black,domain=1/9:1/6]{1/2-2*x};
\addplot[thick,black,domain=1/6:5/18]{1/4-0.5*x};
\addplot[thick,black,domain=5/18:11/18]{1/9};
\addplot[thick,black,domain=5/18:11/18]{8/9-x};
\addplot[thick,black,domain=11/18:2/3]{x-1/2};
\addplot[thick,black,domain=11/18:2/3]{1.5-2*x};
\end{axis}
\end{tikzpicture}}\label{fig:ms-match}
}
\caption{The set \ppmpair\ in \autoref{example:match} for $\marginal=(1/3,1/3,1/3)$}\label{fig:ex-match}
\end{figure}
\end{example}
The Minkowski-sum representation of the set \ppmpair\ in \autoref{eq:minkowski-representation} has several implications. First, a prior exists such that \pair\ is \bce-consistent given \payoff\ if and only if the support of \marginal\ does not include strictly dominated actions, that is, actions $\aaction\in\Actions$ for which \opta\ is empty. Letting $\Actions_+$ denote the support of \marginal, we assume in what follows that $\Actions_+$ contains no strictly dominated actions. Second, because the sets \opta\ are full-dimensional polytopes, \ppmpair\ is itself a full-dimensional polytope  (cf. \autoref{fig:ms-match}).\footnote{To be fully precise, \ppmpair\ is full-dimensional so long as an $\aaction\in\Actions_+$ exists such that \opta\ is full-dimensional.}
 As such, it can be represented either as the convex hull of its extreme points (the so-called $V$-representation) or as the intersection of finitely many halfspaces (the so-called $H$-representation). In fact, letting $\ext(\gfset)$ denote the set of extreme points of a subset \gfset\ of $\realst$, we have that
\begin{align}\label{eq:extreme-ms}\tag{$V$-rep}
\ext(\minkowski(\payoff,\marginal))\subset\sum_{\aaction\in\Actions}\marginal(\aaction)\ext(\opta).
\end{align}
In words, any extreme point of \ppmpair\ is a \marginal-weighted combination of extreme points in the sets \opta, but the opposite may not hold.

Although \autoref{eq:minkowski-representation} characterizes the set of \bce-consistent distributions for a fixed utility function \payoff\ and action distribution \marginal, it does not capture the joint restrictions on the prior and utility that arise from requiring \marginal\ to be rationalizable via information. Moreover, even for fixed \payoff, computing \ppmpair\ is nontrivial\textemdash let alone doing so for every \payoff\ the analyst may wish to consider. Most algorithms for Minkowski sums require the $V$-representation of the sets \opta, which is generally unavailable, with the notable exception of \cite{bergemann2015limits}.\footnote{See \cite*{das2024worst}. In fact, obtaining the $V$-representation of the set \opta\ from its $H$-representation in \autoref{eq:opt-a} is known to be a computationally complex problem  \citep{weibel2007minkowski}.} 


In the remainder of the paper, we focus on studying the $H$-representation of the set \ppmpair. First, as our results below highlight, the $H$-representation allows us to understand which pairs of prior and utility function can jointly rationalize the \marginal\ via information (cf. \autoref{sec:applications}). Second, our characterization of the $H$-representation of the set \ppmpair\ relies on the support function of this set, an object that has received increasing attention in the econometrics literature on random sets for partial identification \citep{molchanov2018random}.\footnote{That literature uses the support function to study the \emph{Aumann expectation} of a random set, with the Minkowski sum serving as its empirical analog. Because \Actions\ is finite, the Minkowski sum exactly represents \ppmpair. The key difference is that in our setting, the data\textemdash the action distribution\textemdash determines the weights of the Minkowski sum, not its summands, whereas in econometrics, data informs the summands, with weights given by empirical frequencies.}
%
%
%

\paragraph{The support function of \ppmpair} Because the set \ppmpair\ is convex, it can be characterized via its support function. Moreover, the support function of \ppmpair\ is the \marginal-weighted sum of the support function of the sets \opta\ \citep{rockafellar1970convex}. Together, these arguments lead to the following statement, which we record for future reference:
\begin{observation}[Support function of \ppmpair]\label{observation:ms}
The prior \prior\ is an element of \ppmpair\ if and only if for all vectors $\bdirection\in\realst$, the following holds:
\begin{align}\label{eq:support}
\max_{\beliefv\in\ppmpair}\pricev\beliefv=\sum_{\aaction\in\Actions}\marginal(\aaction)\max_{\beliefv\in\opta}\bdirection\bbelief\geq\bdirection\boldsymbol{\prior}.
\end{align}
\end{observation}
In words, \pair\ is \bce-consistent given \payoff\ if and only if for all vectors $\bdirection\in\realst$, \prior\ lies below the support function of the set \ppmpair. The left-hand side of \autoref{eq:support} represents the support function of the set \ppmpair\ via the support functions of the sets \opta.\footnote{Whereas \autoref{observation:ms} is an immediate consequence of the Minkowski-sum representation of \ppmpair, it can alternatively be obtained from the dual representation of the system of equations that define the set \bceu\ in \autoref{definition:bce-c} (see \autoref{sec:utility}) and from \citet[Theorem 3]{strassen1965existence}.} An immediate consequence of \autoref{observation:ms} is that to define the $H$-representation of \ppmpair\textemdash defined by a collection of normal vector-height pairs $(\pricev_i,b_i)_{i\in I}$\textemdash the normal vectors alone suffice. By \autoref{eq:support}, the height of the halfspace with normal vector \pricev\ is the \marginal-weighted sum of the values of the support functions of the sets \opta\ in direction \pricev. Finally, we note that \autoref{eq:support}  uses our ongoing assumption that the support of \marginal\ contains no strictly dominated actions. However, by changing $\max$ for $\sup$ in \autoref{eq:support}, the support function of \ppmpair\ also accounts for whether \ppmpair\ is empty by following the convention that the supremum over an empty set is $-\infty$. 
%
%
%

\paragraph{Test functions} \autoref{observation:ms} provides \emph{an} $H$-representation of the set \ppmpair, albeit not the most useful one, because checking whether $\prior\in\ppmpair$ requires verifying \autoref{eq:support} holds for infinitely many vectors $\pricev\in\realst$. The results that follow refine the result in \autoref{observation:ms} by describing sets of \emph{test} functions $\Prices\subsetneq\realst$ such that if \autoref{eq:support} holds for all vectors in $\pricev\in\Prices$, it holds for all vectors $\pricev\in\realst$.\footnote{Ideally, one would like to obtain the minimal $H$-representation of the set \ppmpair. Whereas the results in Theorems \ref{proposition:2-3}--\ref{proposition:2-step} provide the minimal-$H$ representation of \ppmpair\ in some environments, we use the term \emph{test functions} because these vectors are often a superset of the minimal-$H$ representation.}
%


The Minkowski-sum representation of the set \ppmpair\ immediately implies the face structure of the sets \opta\ determine that of the set \ppmpair\ and hence its $H$-representation. The concepts of the \emph{normal fan} of a polytope and the \emph{common refinement} of the normal fans of a collection of polytopes, to which we turn next, allow us to describe how the face structure of the sets \opta\ determine the $H$-representation of the set \ppmpair:
\begin{definition}[Normal fan and common refinement \protect{\citep[Lecture 7]{ziegler2012lectures}}]\label{definition:normal-fan}
Let \gfset\ denote a finite set and \polytope\ a nonempty full-dimensional polytope in $\reals^\gfset$. For each face \face\ of \polytope, let 
\begin{align*}
\normal_\face=\left\{\bdirection\in\reals^\gfset:\face\subseteq\{\gbold\in\polytope:\bdirection\gbold=\max_{\gbbold\in\polytope}\bdirection\gbbold\}\right\},
\end{align*}
denote the linear functions \bdirection\ that attain their maximum over \polytope\ on \face. 
\begin{enumerate}
\item The normal fan of \polytope, $\normal(\polytope)$, is the collection of the cones $\normal_\face$ over the faces of \polytope. 
\item Given a collection of polytopes, $(\polytope_k)_{1\leq k\leq K}$, the common refinement of their normal fans is given by
\begin{align*}
\wedge_{1\leq k\leq K}N(P_k)\equiv\left\{\cap_{1\leq k\leq K} \normal_k:(\forall k\in\{1,\dots, K\})\normal_k\in\normal(\polytope_k)\right\}.
\end{align*}
\end{enumerate}
\end{definition}
Intuitively, the normal fan of a polytope ``partitions'' normal vectors \bdirection\ according to the faces at which the objective, $\pricev\gbold$, attains its maximum value on \polytope.\footnote{Technically speaking, it is not a partition, which can be easily seen in \autoref{fig:match-1}. The direction $(1,1)$ belongs in the normal cone of the face defined by $\belief_1+\belief_2=1$, and also in the normal cone of the face defined by $(\belief_1,\belief_2)=(\nicefrac{1}{2},\nicefrac{1}{2})$.}  In line with the partition interpretation, the common refinement of a collection of normal fans is the meet of these partitions. 

We illustrate \autoref{definition:normal-fan} in the context of \autoref{example:match} in \autoref{fig:normal-fan}. In \autoref{fig:match-1}, the blue arrows denote the normal vectors defining the polytope $\opt(\aaction_1)$. Clockwise from the left, these vectors are given by  the utility differences $\payoff(\aaction_\aindex,\cdot)-\payoff(\aaction_1,\cdot)$ for $j\in\{2,3\}$\textemdash viewed as vectors in $\reals^2$\textemdash the direction $(1,1)$\textemdash arising from the constraint that $\belief_1+\belief_2\leq1$\textemdash and the (negative of the) canonical direction $e_{\type_2}=(0,1)$\textemdash arising from the nonnegativity constraint on $\belief_2$.  The normal fan of $\opt(\aaction_1)$ is given by the collection of cones depicted in dotted red, including the one-dimensional ones. \autoref{fig:partition} depicts the vectors that generate the cones in the normal fan. \autoref{fig:partition} illustrates how the normal fan partitions $\reals^2$, with vectors in the same cell representing linear functions in $\reals^\Types$ that attain their maximum on the same face of $\opt(\aaction_1)$.  \autoref{fig:match-2} illustrates the normal fan of $\opt(\aaction_2)$. \autoref{fig:common-r} illustrates the common refinement of the normal fans of $\opt(\aaction_1)$ and $\opt(\aaction_2)$, which \emph{refines} the partitions of $\reals^2$ induced by each of the normal fans. 

\begin{figure}[t!]
\horizfig[Normal fan of $\opt(\aaction_1)$]{\scalebox{0.7}{%
\begin{tikzpicture}
\begin{axis}[xmin=-0.15,xmax=1.15,ymin=-0.15,ymax=1.15,xticklabels={},yticklabels={},xlabel=$\mu_1$,ylabel=$\mu_2$,x label style={at={(axis description cs:1,-0.01)}},
    y label style={at={(axis description cs:-0.01,1)},rotate=-90}
    ,width=9cm,height=9cm,]
        \addplot[name path=naxis,black,forget plot,domain=-0.15:1.15]{-0.15};
                        \addplot[name path=nyaxis,black,forget plot,domain=-0.15:1.15]{1.15};
    \addplot[name path=axis,black!10,dashed,forget plot,domain=1/2:1]{0};
            \addplot[name path=F1,forget plot,draw=blue!50,domain=0.5:1]{1-x};
    \addplot[name path=A0,forget plot,draw=blue!50,domain=1/3:0.5]{x};
    \addplot[name path=B0,forget plot,draw=blue!50, domain=1/3:0.5]{1-2*x};
    \addplot[blue!50,forget plot]fill between[of= A0 and B0, soft clip={domain=1/3:0.5}];
     \addplot[blue!50]fill between[of= F1 and axis, soft clip={domain=0.5:1}];
            \addlegendentry{$\Delta_u^*(a_1)$}
            \addplot[draw=blue,mark=*] coordinates {(1/2,1/2)};
                        \addplot[draw=blue,mark=*] coordinates {(1/3,1/3)};
                                    \addplot[draw=blue,mark=*] coordinates {(1/2,0)};
                                                \addplot[draw=blue,mark=*] coordinates {(1,0)};
\addplot[thick,blue,->] coordinates {(5/12, 5/12) (1/3,1/2)};
\addplot[thick,dotted,red] coordinates {(5/12,5/12) (-0.15,0.983333)};

\addplot[thick,blue,->] coordinates{(5/12,1/6) (0.325,29/240)};
\addplot[thick,dotted,red] coordinates {(5/12,1/6) (-0.15,-0.11667)};

\addplot[->,blue,thick] coordinates {(1/3,1/3) (0.25,7/24)};
\addplot[dotted,red,domain=-0.15:1/3,name path=cone2b]{0.5*x+1/6};
\addplot[->,blue, thick] coordinates {(1/3,1/3) (1/4,5/12)};
\addplot[dotted,red,domain=-0.15:1/3,name path=cone2]{2/3-x};
\addplot[pattern=dots,pattern color=red]fill between[of=cone2 and cone2b, soft clip={domain=-0.15:1/3}];
\addplot[thick,blue,->] coordinates {(0.75,0) (0.75,-0.1)};
\addplot[thick,dotted,red] coordinates {(0.75,0) (0.75,-0.15)};

\addplot[thick,blue,->] coordinates {(0.75,0.25) (0.85,0.35)};
\addplot[thick,dotted,red] coordinates {(0.75,0.25) (1.15,0.65)};

\addplot[->,blue,thick] coordinates {(1/2,0) (1/2,-0.1)};
\addplot[dotted,red,name path=cone1] coordinates {(1/2,0) (1/2,-0.15)};
\addplot[->,blue,thick] coordinates {(1/2,0) (0.4,-0.05)};
\addplot[dotted,red, name path=cone1b,domain=-0.15:1/2] {0.5*x-0.25};
\addplot[pattern=dots,pattern color=red]fill between[of=cone1b and naxis, soft clip={domain=-0.15:1/2}];

\addplot[->,blue,thick] coordinates {(1,0) (1,-0.1)};
\addplot[dotted,red,name path=cone3] coordinates {(1,0) (1,-0.15)};
\addplot[->,blue,thick] coordinates {(1,0) (1.1,0.1)};
\addplot[dotted,red,name path=cone3b, domain=1:1.15]{x-1};
\addplot[pattern=dots,pattern color=red]fill between[of=cone3b and naxis, soft clip={domain=1:1.15}];

\addplot[->,blue,thick] coordinates {(0.5,0.5) (0.6,0.6)};
\addplot[dotted,red,name path=cone4,domain=0.5:1.15] {x};
\addplot[->,blue,thick] coordinates {(0.5,0.5) (0.4,0.6)};
\addplot[dotted,red,name path=cone4b,domain=-0.15:0.5] {1-x};
\addplot[pattern=dots,pattern color=red]fill between[of=nyaxis and cone4, soft clip={domain=0.5:1.15}];
\addplot[pattern=dots,pattern color=red]fill between[of=nyaxis and cone4b, soft clip={domain=-0.15:0.5}];
            \end{axis}
\end{tikzpicture}}\label{fig:match-1}}\hspace{2cm}
\horizfig[The normal fan partitions $\reals^2$]{
\scalebox{1.4}{
\begin{tikzpicture}
\draw[blue,thick,->](0,0)--(1,1);
\draw[blue,thick,->](0,0)--(-1,1);
\draw[blue,thick,->](0,0)--(-1,-0.5);
\draw[blue,thick,->](0,0)--(0,-1);
\end{tikzpicture}}\label{fig:partition}}

\horizfig[The normal fan of $\opt(\aaction_2)$]{\scalebox{0.7}{%
\begin{tikzpicture}
\begin{axis}[xmin=-0.15,xmax=1.15,ymin=-0.15,ymax=1.15,xticklabels={},yticklabels={},xlabel=$\mu_1$,ylabel=$\mu_2$,x label style={at={(axis description cs:1,-0.01)}},
    y label style={at={(axis description cs:-0.01,1)},rotate=-90}
    ,width=9cm,height=9cm]
\addplot[name path=C0,forget plot,draw=red!30, domain=0:1/3]{(1-x)/2};
    \addplot[red!30,forget plot]fill between[of= F and C0, soft clip={domain=0:1/3}];
        \addplot[red!30]fill between[of= F and A0, soft clip={domain=1/3:0.5}];
                    \addlegendentry{$\Delta_u^*(a_2)$}
            \addplot[draw=red,mark=*] coordinates {(1/2,1/2)};
                        \addplot[draw=red,mark=*] coordinates {(1/3,1/3)};
                                    \addplot[draw=red,mark=*] coordinates {(0,1/2)};
                                                \addplot[draw=red,mark=*] coordinates {(0,1)};

\addplot[red,->,thick] coordinates {(0, 3/4) (-0.1,3/4)};
\addplot[thick,red,dotted] coordinates {(0,3/4) (-0.15,3/4)};
\addplot[red,->,thick] coordinates {(1/4,3/4) (0.35,0.85)};
\addplot[thick,red,dotted] coordinates {(1/4,3/4) (1.15,1.65)};
\addplot[red,->,thick] coordinates {(5/12,5/12) (31/60,19/60)};
\addplot[thick,red,dotted] coordinates {(5/12,5/12) (0.983333,-0.15)};
\addplot[red,->,thick] coordinates {(1/4,3/8) (0.2,11/40)};
\addplot[red,dotted,thick] coordinates {(1/4,3/8) (-0.0125,-0.15)};

\addplot[dotted,red,domain=0.5:1.15,name path=cone2b]{x};
\addplot[dotted,red,domain=0.5:11.15,name path=cone2]{1-x};
\addplot[red,->,thick] coordinates {(1/2, 1/2) (0.6,0.6)};
\addplot[red,->,thick] coordinates {(1/2, 1/2) (0.6,0.4)};
\addplot[pattern=dots,pattern color=red]fill between[of=cone2 and cone2b,soft clip={domain=0.5:1.15}];

\addplot[domain=-0.15:1.15,name path=xaxis]{-0.15};
\addplot[dotted,red,domain=-0.15:1/3,name path=cone3]{2*x-1/3};
\addplot[dotted,red,domain=1/3:1.15,name path=cone3b]{2/3-x};
\addplot[red,->,thick] coordinates {(1/3, 1/3) (13/30,7/30)};
\addplot[red,->,thick] coordinates {(1/3, 1/3) (17/60,7/30)};
\addplot[pattern=dots,pattern color=red]fill between[of=cone3 and naxis,soft clip={domain=-0.15:1/3}];
\addplot[pattern=dots,pattern color=red]fill between[of=naxis and cone3b,soft clip={domain=1/3:1}];

\addplot[dotted,red,domain=-0.15:0,name path=cone4]{2*x+1/2};
\addplot[dotted,red,domain=-0.15:0,name path=cone4b]{1/2};
\addplot[red,->,thick] coordinates {(0,1/2) (-0.1,1/2)};
\addplot[red,->,thick] coordinates {(0,1/2) (-0.05,0.4)};
\addplot[pattern=dots,pattern color=red]fill between[of=cone4 and cone4b,soft clip={domain=-0.15:0}];

\addplot[domain=-0.15:1.15,name path=yaxis]{1.15};
\addplot[dotted,red,domain=0:1/3,name path=cone5]{x+1};
\addplot[dotted,red,domain=-0.15:0,name path=cone5b]{1};
\addplot[red,->,thick] coordinates {(0,1) (-0.1,1)};
\addplot[red,->,thick] coordinates {(0,1) (0.1,1.1)};
\addplot[pattern=dots,pattern color=red]fill between[of=yaxis and cone5b,soft clip={domain=-0.15:0}];
\addplot[pattern=dots,pattern color=red]fill between[of=yaxis and cone5,soft clip={domain=0:1/3}];

\end{axis}
\end{tikzpicture}}\label{fig:match-2}
}
\hspace{2cm}\horizfig[Common refinement of $\normal(\opt(\aaction_1))$ and $\normal(\opt(\aaction_2))$]{
\scalebox{1.4}{\begin{tikzpicture}
\draw[red,thick,->](0,0)--(-1,0);
\draw[red,thick,->](0,0)--(-0.5,-1);
\draw[red,thick,->](0,0)--(1,-1);
\draw[violet,thick,->](0,0)--(1,1);
\draw[blue,thick,->](0,0)--(-1,1);
\draw[blue,thick,->](0,0)--(-1,-0.5);
\draw[blue,thick,->](0,0)--(0,-1);
\end{tikzpicture}}\label{fig:common-r}}
\caption{Illustration of \autoref{definition:normal-fan} in \autoref{example:match}.}\label{fig:normal-fan}
\end{figure}
The normal vectors defining the minimal $H$-representation of a full-dimensional polytope correspond to the extreme rays of the one-dimensional cones in the normal fan of that polytope \citep{ziegler2012lectures}. Figures \ref{fig:match-1} and \ref{fig:partition} illustrate this observation. For ease of reference, below, we denote the extreme rays of the one-dimensional cones in the normal fan of a polytope \polytope\ by $\Prices_{\ext}(\normal(\polytope))$. Consequently, the minimal $H$-representation of \ppmpair\ obtains from the extreme rays of the one-dimensional cones in the normal fan of \ppmpair.\footnote{Recall that we mean one-dimensional in $\reals^{\tnum-1}$.} The normal fan of \ppmpair\ is in fact the common refinement of the normal fans of the sets \opta\ \citep{ziegler2012lectures}. For instance, the normal fan in \autoref{fig:common-r} is the normal fan of $\opt(\aaction_1)+\opt(\aaction_2)$. It follows that the one-dimensional cones in the normal fan of \ppmpair\ are the one-dimensional cones in the common refinement of the normal fans of the sets \opta.  

%
%
%
%
\autoref{theorem:h-representation} summarizes this discussion and presents our basic characterization of the set \bceu. 

\begin{theorem}[$H$-representation of \ppmpair]\label{theorem:h-representation} Suppose $\Actions_+$ contains no strictly dominated actions and hence \ppmpair\ is nonempty. The pair $\pair\in\Posteriors\times\Delta(\Actions)$ is \bce-consistent given \payoff\ if and only if \autoref{eq:support} holds for all $\pricev\in\Prices_{\ext}\left(\wedge_{\aaction\in\Actions_+}\normal(\opta)\right)$.\footnote{Formally, the extreme rays of the one-dimensional cones are vectors $\tilde\pricev\in\reals^{\tnum-1}$, which can be embedded in $\reals^\tnum$ as $\pricev=(0,\tilde{\pricev})$.}
\end{theorem}
The proof of \autoref{theorem:h-representation} and other results can be found in \autoref{appendix:proofs}.

\autoref{theorem:h-representation} characterizes the set of priors \prior\ such that \pair\ is \bce-consistent given \payoff\ via a finite system of inequalities \prior\ must satisfy. In practice, the analyst knows neither the prior \prior\ nor the \dm's utility \payoff. From this perspective, \autoref{theorem:h-representation} describes the joint restrictions on the pairs $(\prior,\payoff)$ for which an information structure exists that rationalizes the given action marginal \marginal. Thus, the approach in this paper also offers an alternative perspective on how to carry out the identification exercise. Typically, the analyst specifies the prior and the utility function up to a finite-dimensional parameter. Instead, for a given utility function \payoff, the set \ppmpair\ provides a nonparametric representation of all priors for which \pair\ is \bce-consistent given \payoff, which is useful whenever the analyst has no information on what the prior should be, but may have auxiliary data on the \dm's payoffs.

%
%

The characterization in \autoref{theorem:h-representation} leaves open the question of how to obtain the vectors that generate the one-dimensional cones in the common refinement of the normal fans $\{\normal(\opta):\aaction\in\Actions_+\}$. The answer to this question is the focus of the rest of the paper. The main challenge is that the one-dimensional cones of the normal fan of \ppmpair\ may not be determined solely from the one-dimensional cones of the normal fan of the sets \opta\textemdash the extreme rays of which we know\textemdash but from intersections of higher-dimensional cones in these normal fans\textemdash the extreme rays of which a priori we do not know. This issue is exacerbated by the cardinality of the set of states and actions, or by how \emph{complex} the utility differences $\payoff(\aaction,\cdot)-\payoff(\aactionb,\cdot)$ are. For that reason, our results place restrictions either on the cardinality of the states or the utility function. When \Types\ contains at most three states, \autoref{proposition:2-3} shows the one-dimensional cones of the normal fan of \ppmpair\ obtain directly from those of the normal fans $\{\normal(\opta):\aaction\in\Actions_+\}$. \autoref{sec:utility} characterizes the extreme rays of the one-dimensional cones of $\normal(\ppmpair)$ for canonical classes of utility functions. 

\paragraph{Connecting the $H$-representations of \opta\ and of \ppmpair} Suppose the set \ppmpair\ is nonempty, and let \pricev\ denote an extreme ray of 
a one-dimensional cone of \opta, for some action $\aaction\in\Actions_+$. Then, \pricev\ is an extreme ray of a one-dimensional cone of \ppmpair.\footnote{When \ppmpair\ has dimension less than $\tnum-1$, this statement applies to those \opta\ with dimension equal to that of \ppmpair.}  
%
%
%
%
Hence, the minimal set of test functions always includes the extreme rays of the one-dimensional cones of the normal fan of \opta\ whenever $\aaction\in\Actions_+$, that is,
\begin{align}\label{eq:inclusion}
\cup_{\aaction\in\Actions_+}\Prices_{\ext}(\normal(\opta))\subset\Prices_{\ext}(\normal(\ppmpair)).
\end{align}
Thus, asking when this inclusion is an equality is natural. \autoref{proposition:2-3} below shows this is the case whenever $\cardt\leq3$. To state \autoref{proposition:2-3}, define
\begin{align}\label{eq:test-small}
\Prices_{\opt}=\cup_{\aactionb\in\Actions_+}\{\payoff(\aactionbb,\cdot)-\payoff(\aactionb,\cdot):\aactionbb\in\Actions\}\cup\{-e_\type:\type\in\Types\},
\end{align}
where  $\payoff(\aactionbb,\cdot)-\payoff(\aactionb,\cdot)\in\realst$ is the vector that collects the payoff differences between \aactionbb\ and \aactionb\ as a function of the state, and $e_\type$ is the vector in \realst\ with a 1 in the \type-coordinate and $0$ elsewhere. For a given $\aactionb\in\Actions$, these vectors define the $H$-representation of \optab\ and hence contain $\Prices_{\ext}(\normal(\optab))$. Taking union over the different actions $\aactionb$ yields the set $\Prices_{\opt}$.\footnote{We could refine the set $\Prices_{\opt}$ by eliminating normal directions to \opta\ that do not define a facet of \opta.}

\autoref{proposition:2-3} summarizes the above discussion:
\begin{theorem}[Simple state spaces]\label{proposition:2-3}
Suppose $|\Types|\leq 3$. The pair $\pair\in\Posteriors\times\Delta(\Actions)$ is \bce-consistent given \payoff\ if and only if \autoref{eq:support} holds for all $\pricev\in\Prices_{\opt}$, that is, if and only if for all states $\type\in\Types$,
\begin{align}
\sum_{\aaction\in\Actions}\marginal(\aaction)\min_{\belief\in\opta}\belief(\type)&\leq\prior(\type),\label{eq:bce-c-states}\tag{BM}
\intertext{and for all pairs of actions $\aactionb\in\Actions_+$ and $\aactionbb\in\Actions$,}
\label{eq:bce-c-actions}
\sum_{\aaction\in\Actions}\marginal(\aaction)\max_{\belief\in\opta}\sum_{\type\in\Types}\belief(\type)\left[\payoff(\aactionbb,\type)-\payoff(\aactionb,\type)\right]&\geq\sum_{\type\in\Types}\prior(\type)\left[\payoff(\aactionbb,\type)-\payoff(\aactionb,\type)\right].\tag{PM}
\end{align}
\end{theorem}
In \autoref{proposition:2-3}, \autoref{eq:bce-c-states} corresponds to \autoref{eq:support} evaluated at the canonical vectors $-e_\type$, whereas \autoref{eq:bce-c-actions} corresponds to \autoref{eq:support} evaluated at the vectors $\payoff(\aactionbb,\cdot)-\payoff(\aactionb,\cdot)$ for different action pairs $(\aactionb,\aactionbb)$. 

We remark that as long as a direction $\pricev\in\Prices_{\opt}$ is an extreme ray in the one-dimensional cone of the normal fan of \emph{some} \opta\ such that $\aaction\in\Actions_+$, the corresponding equation in the statement of \autoref{proposition:2-3} is part of the $H$-representation of \ppmpair\ independently of the cardinality of the states. Thus, whereas any $\pricev\in\realst$ defines through \autoref{eq:support} a necessary condition for \pair\ to be \bce-consistent given \payoff, \autoref{eq:inclusion} gives us a precise sense in which equations \ref{eq:bce-c-states} and \ref{eq:bce-c-actions} are \emph{the} necessary conditions and can be used to rule out pairs \pair\ that are not \bce-consistent given \payoff.  In fact, we can  provide intuition for Equations \ref{eq:bce-c-states} and \ref{eq:bce-c-actions} by reasoning about their necessity, starting from \autoref{eq:bce-c-states}. If \pair\ is \bce-consistent given \payoff, then we can find a belief system that satisfies for each $\aaction\in\Actions$, 
\begin{align*}
\sum_{\aaction\in\Actions}\marginal(\aaction)\belief(\type|\aaction)=\prior(\type)\Rightarrow \sum_{\aaction\in\Actions}\marginal(\aaction)\min_{\belief\in\opta}\belief(\type)\leq\prior(\type).
\end{align*}
The above implication holds because \bce-consistency of \pair\ implies that for each $\aaction\in\Actions_+$, $\belief(\cdot|\aaction)$ is an element of \opta. Thus, if \pair\ is \bce-consistent, the left-hand side of \autoref{eq:bce-c-states} is a lower bound on the prior. In particular, \autoref{eq:bce-c-states} implies the non-negativity constraints on \prior. Whenever the left-hand side of \autoref{eq:bce-c-states} is positive, rationalizing \marginal\ via information requires the \dm\ to assign strictly positive probability to some states.

Intuitively, \autoref{eq:bce-c-states} verifies whether finding a belief system $\{\belief(\cdot|\aaction):\aaction\in\Actions\}$ that satisfies the \emph{belief martingale} property relative to \prior\ is possible. Indeed, if the inequality in \autoref{eq:bce-c-states} failed for some state, either the frequency with which the \dm\ is taking a given action, or the minimum probability the \dm\ needs to assign to this state so that taking a given action is optimal, reflects that the \dm\ is more optimistic than at the prior. In this case, \marginal\ cannot be rationalized via information.

To provide intuition for \autoref{eq:bce-c-actions}, considering the binary-action case is useful. Assume $\Actions=\{\aaction_1,\aaction_2\}$:  \bce-consistency of \pair\ implies a belief system exists that satisfies at most two obedience constraints, which we can write as the following chain of inequalities:
\begin{align*}
\marginal(\aaction_1)\sum_{\type\in\Types}\belief(\type|\aaction_1)\left[\payoff(\aaction_1,\type)-\payoff(\aaction_2,\type)\right]\geq0\geq\marginal(\aaction_2)\sum_{\type\in\Types}\belief(\type|\aaction_2)\left[\payoff(\aaction_1,\type)-\payoff(\aaction_2,\type)\right].
\end{align*}
That is, information resulting in beliefs $\belief(\cdot|\aaction_1)$ and $\belief(\cdot|\aaction_2)$ alters the relative ranking of $\aaction_1$ and $\aaction_2$. Equation \ref{eq:state-mg-b} implies that if we add up both sides of the above chain, we obtain the following:
\begin{align*}
&\marginal(\aaction_1)\sum_{\type\in\Types}\belief(\type|\aaction_1)\left[\payoff(\aaction_1,\type)-\payoff(\aaction_2,\type)\right]+\marginal(\aaction_2)\sum_{\type\in\Types}\belief(\type|\aaction_2)\left[\payoff(\aaction_1,\type)-\payoff(\aaction_2,\type)\right]=\\
=&\sum_{\type\in\Types}\prior(\type)\left[\payoff(\aaction_1,\type)-\payoff(\aaction_2,\type)\right].
\end{align*}
In other words, although information can alter the relative ranking between the two actions, it cannot systematically do so: on average, the ranking between $\aaction_1$ and $\aaction_2$ must coincide with how the \dm\ ranks these two actions at the prior \prior. This observation is the analogue to the belief martingale condition, albeit in terms of the \dm's payoffs; hence we refer to it as a \emph{payoff martingale} condition. Using once again the property that $\belief(\cdot|\aaction)\in\opta$, the above equality implies \autoref{eq:bce-c-actions}. 

Consider now the case in which the \dm\ has three actions, $\Actions=\{\aaction_1,\aaction_2,\aaction_3\}$, and once again, \autoref{eq:bce-c-actions} for the pair $\aaction_1,\aaction_2$. Whereas the obedience constraints feature the comparison between these two actions when $\aaction_1$ or $\aaction_2$ is recommended, no such inequality arises when $\aaction_3$ is recommended. Still, \autoref{eq:bce-c-actions} adds up over all actions, including $\aaction_3$. The reason is that when ensuring $\aaction_3$ is optimal, the obedience constraints place no restrictions on the relative ranking of $\aaction_1$ and $\aaction_2$. Yet, the ranking of these two actions has to average to their ranking at the prior over \emph{all} beliefs the \dm\ has. \autoref{eq:bce-c-actions} checks that the payoff martingale condition can be satisfied by placing bounds on the induced relative rankings for each action pair $(\aactionb,\aactionbb)$ across all action recommendations.

We illustrate \autoref{proposition:2-3} using \autoref{example:match}:
\setcounter{example}{0}
\begin{example}[continued]\label{example:match-h}
\begin{figure}[t!]
\centering
\subfloat[Belief-martingale conditions \eqref{eq:bce-c-states}]{
\scalebox{0.65}{%
\begin{tikzpicture}
\begin{axis}[xmin=0,xmax=1,ymin=0,ymax=1,xticklabels={},yticklabels={},xlabel=$\mu_1$,ylabel=$\mu_2$,x label style={at={(axis description cs:1,-0.01)}},
    y label style={at={(axis description cs:-0.01,1)},rotate=-90}
    ,width=9cm,height=9cm]
    \addplot[name path=axis,forget plot,domain=0:1]{0};
    \addplot[name path=F,forget plot,draw=red!50,domain=0:0.5]{1-x};
        \addplot[name path=F1,forget plot,draw=blue!50,domain=0.5:1]{1-x};
    \addplot[name path=A0,forget plot,draw=blue!50,domain=1/3:0.5]{x};
    \addplot[name path=B0,forget plot,draw=blue!50, domain=1/3:0.5]{1-2*x};
    \addplot[blue!50,forget plot]fill between[of= A0 and B0, soft clip={domain=1/3:0.5}];
     \addplot[blue!50]fill between[of= F1 and axis, soft clip={domain=0.5:1}];
            \addlegendentry{$\Delta_u^*(a_1)$}
       \addplot[name path=C0,forget plot,draw=red!30, domain=0:1/3]{(1-x)/2};
    \addplot[red!30,forget plot]fill between[of= F and C0, soft clip={domain=0:1/3}];
        \addplot[red!30]fill between[of= F and A0, soft clip={domain=1/3:0.5}];
                    \addlegendentry{$\Delta_u^*(a_2)$}
        \addplot[orange!30,forget plot]fill between[of=C0 and axis, soft clip={domain=0:1/3}];
                \addplot[orange!30]fill between[of=B0 and axis, soft clip={domain=1/3:0.5}];
                            \addlegendentry{$\Delta_u^*(a_3)$}
\addplot[name path=A1, draw=black!25,domain=0:1]{8/9-x};
\addplot[name path=B1, draw=black!25,domain=0:8/9]{1/9};
\draw[name path=C1, black!25](1/9,0)--(1/9,8/9);
\addplot[black!25,opacity=0.75]fill between [of= A1 and B1, soft clip={domain=1/9:7/9}];
\end{axis}
\end{tikzpicture}}\label{fig:bce-c-states}
}
\subfloat[Payoff-martingale conditions \eqref{eq:bce-c-actions}]{\scalebox{0.65}{%
\begin{tikzpicture}
\begin{axis}[xmin=0,xmax=1,ymin=0,ymax=1,xticklabels={},yticklabels={},xlabel=$\mu_1$,ylabel=$\mu_2$,x label style={at={(axis description cs:1,-0.01)}},
    y label style={at={(axis description cs:-0.01,1)},rotate=-90}
    ,width=9cm,height=9cm]
    \addplot[name path=axis,forget plot,domain=0:1]{0};
    \addplot[name path=F,forget plot,draw=red!50,domain=0:0.5]{1-x};
        \addplot[name path=F1,forget plot,draw=blue!50,domain=0.5:1]{1-x};
    \addplot[name path=A0,forget plot,draw=blue!50,domain=1/3:0.5]{x};
    \addplot[name path=B0,forget plot,draw=blue!50, domain=1/3:0.5]{1-2*x};
    \addplot[blue!50,forget plot]fill between[of= A0 and B0, soft clip={domain=1/3:0.5}];
     \addplot[blue!50]fill between[of= F1 and axis, soft clip={domain=0.5:1}];
            \addlegendentry{$\Delta_u^*(a_1)$}
       \addplot[name path=C0,forget plot,draw=red!30, domain=0:1/3]{(1-x)/2};
    \addplot[red!30,forget plot]fill between[of= F and C0, soft clip={domain=0:1/3}];
        \addplot[red!30]fill between[of= F and A0, soft clip={domain=1/3:0.5}];
                    \addlegendentry{$\Delta_u^*(a_2)$}
        \addplot[orange!30,forget plot]fill between[of=C0 and axis, soft clip={domain=0:1/3}];
                \addplot[orange!30]fill between[of=B0 and axis, soft clip={domain=1/3:0.5}];
                            \addlegendentry{$\Delta_u^*(a_3)$}
\addplot[name path=A0,draw=black!50,domain=0:1]{1.5-2*x};
\addplot[name path=A,draw=black!50,domain=0:1]{0.75-0.5*x};
\addplot[name path=B,draw=black!50,domain=0:1]{1/4-0.5*x};
\addplot[name path=C,draw=black!50,domain=0:1]{1/2-2*x};
\addplot[name path=D,draw=black!50,domain=0:1]{1/2+x};
\addplot[name path=E,draw=black!50,domain=0:1]{x-1/2};
\addplot[black!50,opacity=0.75] fill between[of=A and B, soft clip={domain=1/6:0.5}];
\addplot[black!50,opacity=0.75] fill between[of=D and C, soft clip={domain=0:1/6}];
\addplot[black!50,opacity=0.75] fill between[of=E and A0, soft clip={domain=0.5:2/3}];
\end{axis}
\end{tikzpicture}}\label{fig:bce-c-actions}}
\subfloat[The set \ppmpair]{\scalebox{0.65}{%
\begin{tikzpicture}
\begin{axis}[xmin=0,xmax=1,ymin=0,ymax=1,xticklabels={},yticklabels={},xlabel=$\mu_1$,ylabel=$\mu_2$,x label style={at={(axis description cs:1,-0.01)}},
    y label style={at={(axis description cs:-0.01,1)},rotate=-90}
    ,width=9cm,height=9cm]
    \addplot[name path=axis,forget plot,domain=0:1]{0};
    \addplot[name path=F,forget plot,draw=red!50,domain=0:0.5]{1-x};
        \addplot[name path=F1,forget plot,draw=blue!50,domain=0.5:1]{1-x};
    \addplot[name path=A0,forget plot,draw=blue!50,domain=1/3:0.5]{x};
    \addplot[name path=B0,forget plot,draw=blue!50, domain=1/3:0.5]{1-2*x};
    \addplot[blue!50,forget plot]fill between[of= A0 and B0, soft clip={domain=1/3:0.5}];
     \addplot[blue!50]fill between[of= F1 and axis, soft clip={domain=0.5:1}];
            \addlegendentry{$\Delta_u^*(a_1)$}
       \addplot[name path=C0,forget plot,draw=red!30, domain=0:1/3]{(1-x)/2};
    \addplot[red!30,forget plot]fill between[of= F and C0, soft clip={domain=0:1/3}];
        \addplot[red!30]fill between[of= F and A0, soft clip={domain=1/3:0.5}];
                    \addlegendentry{$\Delta_u^*(a_2)$}
        \addplot[orange!30,forget plot]fill between[of=C0 and axis, soft clip={domain=0:1/3}];
                \addplot[orange!30]fill between[of=B0 and axis, soft clip={domain=1/3:0.5}];
                            \addlegendentry{$\Delta_u^*(a_3)$}
                \addplot[name path=A1, draw=black!25,domain=0:1]{8/9-x};
\addplot[name path=B1, draw=black!25,domain=0:8/9]{1/9};
\draw[name path=C1, black!25](1/9,0)--(1/9,8/9);
\addplot[black!25,opacity=0.75]fill between [of= A1 and B1, soft clip={domain=1/9:7/9}];
\addplot[name path=A0,draw=black!50,domain=0:1]{1.5-2*x};
\addplot[name path=A,draw=black!50,domain=0:1]{0.75-0.5*x};
\addplot[name path=B,draw=black!50,domain=0:1]{1/4-0.5*x};
\addplot[name path=C,draw=black!50,domain=0:1]{1/2-2*x};
\addplot[name path=D,draw=black!50,domain=0:1]{1/2+x};
\addplot[name path=E,draw=black!50,domain=0:1]{x-1/2};
\addplot[black!50,opacity=0.75] fill between[of=A and B, soft clip={domain=1/6:0.5}];
\addplot[black!50,opacity=0.75] fill between[of=D and C, soft clip={domain=0:1/6}];
\addplot[black!50,opacity=0.75] fill between[of=E and A0, soft clip={domain=0.5:2/3}];
\addplot[thick,black,domain=1/6:5/18]{0.75-0.5*x};
\addplot[thick,black,domain=1/9:1/6]{1/2+x};
\addplot[thick,black] coordinates {(1/9,11/18) (1/9,5/18)};
\addplot[thick,black,domain=1/9:1/6]{1/2-2*x};
\addplot[thick,black,domain=1/6:5/18]{1/4-0.5*x};
\addplot[thick,black,domain=5/18:11/18]{1/9};
\addplot[thick,black,domain=5/18:11/18]{8/9-x};
\addplot[thick,black,domain=11/18:2/3]{x-1/2};
\addplot[thick,black,domain=11/18:2/3]{1.5-2*x};
\end{axis}
\end{tikzpicture}}\label{fig:2-3-int}}
\caption{Illustrating \autoref{proposition:2-3} in \autoref{example:match} for $\marginal=(1/3,1/3,1/3)$.}\label{fig:2-3-match}
\end{figure}
\autoref{fig:2-3-match} illustrates the intersection of the halfspaces defined by the equations in \autoref{proposition:2-3} in the case of a uniform marginal. \autoref{fig:bce-c-states} illustrates the belief martingale equations \ref{eq:bce-c-states}. When the \dm\ is taking each action with probability $1/3$, these equations require that the \dm\ assigns probability of at least $1/9$ to each of the states. This is intuitive: the \dm\ taking all actions with equal probability reflects they assign enough probability to each of the states. \autoref{fig:bce-c-actions} illustrates the payoff martingale conditions \ref{eq:bce-c-actions}. In this example, for any two actions $\aaction_\aindex,\aaction_k$, the payoff difference at belief \belief\ is given by the difference in beliefs $\belief(\type_\aindex)-\belief(\type_k)$. The payoff martingale conditions place bounds on the difference in beliefs at the prior, $\prior(\type_\aindex)-\prior(\type_k)$ for any two pair of states. In particular, they imply $\prior(\type_\aindex)-\prior(\type_k)\in[-1/2,1/2]$. Again, this is intuitive: the \dm\ taking all actions with equal probability reflects that at the prior, the \dm's beliefs do not significantly favor any of the states. 

\autoref{fig:2-3-int} illustrates the intersection of all halfspaces defined by Equations \ref{eq:bce-c-states} and \ref{eq:bce-c-actions}. In \autoref{example:match}, each of these equations define a face of the polytope. In other words, in this example, each of the test functions, $\Prices_{\opt}$, is needed to define the set \ppmpair; hence, the characterization in \autoref{proposition:2-3} provides the minimal $H$-representation of the set of priors \prior\ such that \pair\ is \bce-consistent given \payoff.
\end{example}
Whereas \autoref{example:match-h} provides an instance in which all equations in \autoref{proposition:2-3} are needed to define \ppmpair, this is not always the case. Indeed, when $\Types=\{\type_1,\type_2\}$, the belief martingale equations \ref{eq:bce-c-states} alone define this set. In the case of two states, the prior is summarized by the probability of $\type_2$, $\prior(\type_2)$. It is immediate that the set \ppmpair\ is an interval; hence, it is defined by two inequalities. The same is true of the sets $\{\opta:\aaction\in\Actions\}$, which in a slight abuse of notation, we define as $\opta=[\underline{\belief}(\type_2|\aaction),\overline{\belief}(\type_2|\aaction)]$. \autoref{corollary:2-state} shows the lower and upper bounds of \opta\ define the lower and upper bounds of \ppmpair:

\begin{corollary}[Binary states]\label{corollary:2-state} 
Suppose $|\Types|=2$. Then, \pair\ is \bce-consistent given \payoff\ if and only if
\begin{align*}
\sum_{\aaction\in\Actions}\marginal(\aaction)\underline{\belief}(\type_2|\aaction)\leq\prior(\type_2)\leq\sum_{\aaction\in\Actions}\marginal(\aaction)\overline{\belief}(\type_2|\aaction).
\end{align*}
\end{corollary}
We make two remarks. First, the right-hand side of the expression in \autoref{corollary:2-state} obtains from \autoref{eq:bce-c-states} at $\type_1$. Second, the representation of \ppmpair\ via its extreme points  provides one way of understanding \autoref{corollary:2-state} (cf. \autoref{eq:extreme-ms}). With binary states, the beliefs $\{\underline{\belief}(\type_2|\aaction),\overline{\belief}(\type_2|\aaction)\}$ are the extreme points of \opta. Furthermore, the extreme points of \ppmpair\ are \marginal-weighted convex combinations of the extreme points of the sets \opta, but not the reverse. \autoref{corollary:2-state} states that only the \marginal-weighted convex combination of the minimal extreme points and of the maximal extreme points can be extreme in \ppmpair.

\paragraph{Beyond simple state spaces} As anticipated above, the result in \autoref{proposition:2-3} does not extend to larger state spaces: when the cardinality of \Types\ is at least four, an extreme ray in a one-dimensional cone of the normal fan of \ppmpair\ may not be an extreme ray in a one-dimensional cone of any of the normal fans $\{\normal(\opta):\aaction\in\Actions\}$.\footnote{Recall the normal fan of \ppmpair\ is the common refinement of the normal fans $\normal(\opta)$. Hence, a normal cone of \ppmpair\ obtains by intersecting normal cones in $\{\normal(\opta):\aaction\in\Actions\}$. In three or more dimensions, the extreme rays of the intersection of two cones need not be extreme rays in any of the cones in the intersection.} We illustrate this possibility with a simple binary-action, four-state example:
\begin{example}[Test functions in $\Prices_{\opt}$ do not suffice when $\cardt\geq4$]\label{example:counter}
Consider the following decision problem with four states, $\Types=\{\type_1,\dots,\type_4\}$, and two actions, $\Actions=\{\aaction_1,\aaction_2\}$. The utility is given by:
\begin{align*}
\payoff(\aaction_1,\cdot)=(0,0,0,0),\;\;\payoff(\aaction_2,\cdot)&=(-9,-5,-1,5).
\end{align*}
That is, $\aaction_1$ is preferred when the \dm\ assigns high probability to states other than $\type_4$, and $\aaction_2$ is preferred when the \dm\ assigns high probability to state $\type_4$. We denote by $\belief_\tindex$ the probability of state $\type_\tindex$.  Using the property that $\belief_1=1-\belief_2-\belief_3-\belief_4$, \autoref{fig:counter-opt} depicts the sets \opta\ for $\aaction_1$ in grey and $\aaction_2$ in white.
    \begin{figure}[h] 
\centering 
\subfloat[The sets \opta]{
\begin{tikzpicture}[scale = 0.4]

\begin{axis}[%
width=5.638in,
height=4.754in,
at={(1.206in,0.642in)},
scale only axis,
plot box ratio=1 1 1,
xmin=0,
xmax=1,
tick align=outside,
xlabel style={font=\color{white!15!black}},
xlabel={$\belief_2$},
ymin=0,
ymax=1,
ylabel style={font=\color{white!15!black}},
ylabel={$\belief_3$},
zmin=0,
zmax=1,
zlabel style={font=\color{white!15!black}},
zlabel={$\belief_4$},
view={141}{8},
axis x line*=bottom,
axis y line*=left,
axis z line*=left,
xmajorgrids,
ymajorgrids,
zmajorgrids
]

\addplot3[area legend, draw=black, fill=white!90!black, fill opacity=0.5, forget plot]
table[row sep=crcr] {%
x	y	z\\
-0	-0	1\\
0	0.833333333333333	0.166666666666667\\
-0	-0	0.642857142857143\\
}--cycle;

\addplot3[area legend, draw=black, fill=white!90!black, fill opacity=0.5, forget plot]
table[row sep=crcr] {%
x	y	z\\
0.5	0	0.5\\
-0	-0	0.642857142857143\\
-0	-0	1\\
}--cycle;

\addplot3[area legend, draw=black, fill=white!90!black, fill opacity=0.5, forget plot]
table[row sep=crcr] {%
x	y	z\\
-0	-0	1\\
0	0.833333333333333	0.166666666666667\\
0.5	0	0.5\\
}--cycle;

\addplot3[area legend, draw=black, fill=white!90!black, fill opacity=0.5, forget plot]
table[row sep=crcr] {%
x	y	z\\
-0	-0	0.642857142857143\\
0	0.833333333333333	0.166666666666667\\
0.5	0	0.5\\
}--cycle;

\addplot3[area legend, draw=black, fill=gray, fill opacity=0.5, forget plot]
table[row sep=crcr] {%
x	y	z\\
-0	-0	-0\\
0	1	0\\
-0	0.833333333333333	0.166666666666667\\
-0	-0	0.642857142857143\\
}--cycle;

\addplot3[area legend, draw=black, fill=gray, fill opacity=0.5, forget plot]
table[row sep=crcr] {%
x	y	z\\
1	0	0\\
0.5	0	0.5\\
-0	-0	0.642857142857143\\
-0	-0	-0\\
}--cycle;

\addplot3[area legend, draw=black, fill=gray, fill opacity=0.5, forget plot]
table[row sep=crcr] {%
x	y	z\\
-0	-0	-0\\
1	0	0\\
0	1	0\\
}--cycle;

\addplot3[area legend, draw=black, fill=gray, fill opacity=0.5, forget plot]
table[row sep=crcr] {%
x	y	z\\
0	1	0\\
-0	0.833333333333333	0.166666666666667\\
0.5	0	0.5\\
1	0	0\\
}--cycle;

\addplot3[area legend, draw=black, fill=gray, fill opacity=0.5, forget plot]
table[row sep=crcr] {%
x	y	z\\
-0	-0	0.642857142857143\\
-0	0.833333333333333	0.166666666666667\\
0.5	0	0.5\\
}--cycle;
\end{axis}
\end{tikzpicture}\label{fig:counter-opt}}
\subfloat[The set \ppmpair\ for the uniform marginal]{
        \begin{tikzpicture}[scale = 0.4]
\begin{axis}[%
width=5.638in,
height=4.754in,
at={(1.206in,0.642in)},
scale only axis,
plot box ratio=1 1 1,
xmin=0,
xmax=1,
tick align=outside,
xlabel style={font=\color{white!15!black}},
xlabel={$\belief_2$},
ymin=0,
ymax=1,
ylabel style={font=\color{white!15!black}},
ylabel={$\belief_3$},
zmin=0,
zmax=1,
zlabel style={font=\color{white!15!black}},
zlabel={$\belief_4$},
view={141}{8},
axis x line*=bottom,
axis y line*=left,
axis z line*=left,
xmajorgrids,
ymajorgrids,
zmajorgrids
]

\addplot3[area legend, draw=black, fill=white!70!black, fill opacity=0.3, forget plot]
table[row sep=crcr] {%
x	y	z\\
-1.07938349616335e-16	3.08148791101958e-33	0.321428571428571\\
-1.11689431940299e-16	-1.53731817792012e-17	0.821428571428572\\
0.25	-2.20399014709994e-17	0.75\\
0.75	-2.4392496726595e-17	0.25\\
0.25	-6.66671969179822e-18	0.25\\
}--cycle;

\addplot3[area legend, draw=black, fill=white!70!black, fill opacity=0.3, forget plot]
table[row sep=crcr] {%
x	y	z\\
0.25	-6.66671969179822e-18	0.25\\
6.9166538021343e-16	0.416666666666666	0.0833333333333333\\
-2.70481200845768e-17	0.416666666666668	0.0833333333333327\\
-2.7048120084577e-17	0.416666666666667	0.0833333333333333\\
-1.07938349616335e-16	3.08148791101958e-33	0.321428571428571\\
}--cycle;

\addplot3[area legend, thick, draw=red, pattern color=red, pattern=north east lines, forget plot]
table[row sep=crcr] {%
x	y	z\\
0.75	-2.4392496726595e-17	0.25\\
0.5	0.416666666666667	0.0833333333333332\\
-2.70481200845768e-17	0.416666666666668	0.0833333333333327\\
6.9166538021343e-16	0.416666666666666	0.0833333333333333\\
-2.7048120084577e-17	0.416666666666667	0.0833333333333332\\
0.25	-6.66671969179822e-18	0.25\\
}--cycle;

\addplot3[area legend, draw=black, fill=white!70!black, fill opacity=0.3, forget plot]
table[row sep=crcr] {%
x	y	z\\
-1.07938349616335e-16	3.08148791101958e-33	0.321428571428571\\
-2.7048120084577e-17	0.416666666666667	0.0833333333333332\\
-2.7048120084577e-17	0.416666666666667	0.0833333333333333\\
-2.70481200845768e-17	0.416666666666668	0.0833333333333327\\
6.78766797398383e-17	0.916666666666667	0.0833333333333333\\
-3.07992024085415e-17	0.416666666666666	0.583333333333333\\
-1.11689431940299e-16	-1.53731817792012e-17	0.821428571428572\\
}--cycle;

\addplot3[area legend, draw=black, fill=white!70!black, fill opacity=0.3, forget plot]
table[row sep=crcr] {%
x	y	z\\
-2.7048120084577e-17	0.416666666666667	0.0833333333333333\\
-2.7048120084577e-17	0.416666666666667	0.0833333333333332\\
6.78766797398383e-17	0.916666666666667	0.0833333333333333\\
0.5	0.416666666666667	0.0833333333333332\\
6.9166538021343e-16	0.416666666666666	0.0833333333333333\\
}--cycle;

\addplot3[area legend, draw=black, fill=white!70!black, fill opacity=0.3, forget plot]
table[row sep=crcr] {%
x	y	z\\
-1.11689431940299e-16	-1.53731817792012e-17	0.821428571428572\\
-3.07992024085415e-17	0.416666666666666	0.583333333333333\\
0.25	-2.20399014709994e-17	0.75\\
}--cycle;

\addplot3[area legend, draw=black, fill=white!70!black, fill opacity=0.3, forget plot]
table[row sep=crcr] {%
x	y	z\\
0.75	-2.4392496726595e-17	0.25\\
0.5	0.416666666666667	0.0833333333333332\\
6.78766797398383e-17	0.916666666666667	0.0833333333333333\\
-3.07992024085415e-17	0.416666666666666	0.583333333333333\\
0.25	-2.20399014709994e-17	0.75\\
}--cycle;
\end{axis}
\end{tikzpicture}\label{fig:counter-ms}}%
\caption{\autoref{proposition:2-3} does not hold when $\cardt\geq4$.}
\label{fig:counter}
\end{figure}

Consider now the uniform marginal, $\marginal(\aaction_1)=1/2$. The result in \autoref{proposition:2-3} suggests the set \ppmpair\ is defined by the intersection of six halfspaces. However, the polytope \ppmpair\ has seven facets as illustrated in \autoref{fig:counter-ms}. Indeed, whereas the grey facets correspond to the test functions in the statement of \autoref{proposition:2-3}, the hatched red facet has normal vector $(0,-2,-5)$, which corresponds to the test function $\pricev=(0,0,-2,-5)$. This facet obtains from two edges\textemdash one-dimensional faces\textemdash one in $\opt(\aaction_1)$ and the other in $\opt(\aaction_2)$. The former corresponds to the intersection of the non-negativity constraints on $\belief_3$ and $\belief_4$; the latter corresponds to the intersection of the obedience constraint and the constraint $\belief_2+\belief_3+\belief_4\leq1$. Associated with each of these edges is a two-dimensional normal cone, the intersection of which yields a one-dimensional normal cone in the normal fan of \ppmpair\ with $(0,-2,-5)$ as an extreme ray.
\end{example}

\section{\bce-consistency in monotone and concave decision problems}\label{sec:utility}
In this section, we identify assumptions on the utility function \payoff, which allow us to refine the basic characterization in \autoref{theorem:h-representation} without imposing assumptions on the cardinality of the state space. Concretely, we focus on \emph{monotone and concave decision problems} in which the utility function \payoff\ satisfies concavity and increasing differences assumptions (\autoref{assumption:mcv}). \autoref{proposition:mcv} identifies a set of test functions under this assumption. We next consider two special cases: affine utility differences (\autoref{sec:affine}) and two-step utility differences (\autoref{sec:two-step}). In each case, we characterize the set of \bce-consistent marginals via a system of finitely many inequalities. In contrast to the results in \autoref{sec:main}, for which we relied on the properties of the Minkowski sum, the characterization in this section relies on the dual of the program induced by checking the feasibility of Equations \ref{eq:obedience}, \ref{eq:state-mg}, and \ref{eq:action-mg}. This dual approach allows us to identify both the test functions and the value of the support function of the set \ppmpair\textemdash the left-hand side of \autoref{eq:support}\textemdash and thus provide a more succinct characterization of the set \ppmpair. \autoref{sec:foa} extends the results in this section to the case in which the sets of actions and states are compact Polish spaces.



\paragraph{Monotone and concave decision problems} We focus on decision problems in which the utility function \payoff\ satisfies a combination of concavity and increasing-differences assumptions. To state these assumptions, the ordering of the states and the actions matters. Recall we are indexing the actions with $\aindex\in\{1,\dots,\anum\}$ and the states with $\tindex\in\{1,\dots,\tnum\}$. Definitions \ref{definition:id} and \ref{definition:cv} below place restrictions on the utility difference across adjacent actions, which we denote by
\begin{align}
\diff(\aaction_{\aindex+1},\aaction_\aindex,\type_\tindex)\equiv \payoff(\aaction_{\aindex+1},\type_\tindex)-\payoff(\aaction_\aindex,\type_\tindex).
\end{align}

\begin{definition}[Increasing differences]\label{definition:id}
    The utility function $\payoff:\Actions\times\Types\to\reals$ has increasing differences if for all $j\in\{1,\dots,\anum-1\}$, $\diff(\aaction_{\aindex+1},\aaction_\aindex,\type_\tindex)$ is increasing in $i$.
\end{definition}
In words, under increasing differences, the \dm\ finds higher index actions more attractive than lower index ones in higher states.

\begin{definition}[Concavity$^*$]\label{definition:cv}
    The utility function $\payoff:\Actions\times\Types\to\reals$ is concave$^*$ if the following hold:
    \begin{enumerate}
        \item For all $\type\in\Types$, the function $\diff(\aaction_{\aindex+1},\aaction_\aindex,\type)$ is decreasing in $j\in\{1,\dots,\anum-1\}$, and
        \item For all $j\in\{2,\dots,\anum-1\}$ and all $\belief\in\Posteriors$, if $\sum_\type\belief(\type)\diff(\aaction_{\aindex+1},\aaction_\aindex,\type)=0$, \newline then~$\sum_\type\belief(\type)\diff(\aaction_\aindex,\aaction_{\aindex-1},\type)\neq0$.
    \end{enumerate}
\end{definition}
In words, a utility function is concave$^*$ if, for a given state, there are decreasing returns to increasing the actions\textemdash that is, $\payoff(\cdot,\type)$ is concave\textemdash and, for any belief, the \dm\ can be indifferent between at most two actions.

When the \dm's utility function satisfies the above definitions, we say the decision problem is monotone and concave. We record this in \autoref{assumption:mcv} for ease of reference:
\begin{assumption}[Monotone and concave decision problems]\label{assumption:mcv} 
The decision problem is monotone and concave if \payoff\ satisfies Definitions \ref{definition:id} and \ref{definition:cv}.
\end{assumption}
\autoref{assumption:mcv} affords the following simplification in determining whether  a pair of distributions is \bce-consistent given \payoff. Recall that \bce-consistency of \pair\ is equivalent to the feasibility of the system defined by Equations \ref{eq:obedience}, \ref{eq:state-mg}, and \ref{eq:action-mg}. \autoref{assumption:mcv} implies we can ignore all obedience constraints not involving \emph{adjacent} action pairs, $(\aaction_\aindex,\aaction_{\aindex+1})$ and $(\aaction_\aindex,\aaction_{j-1})$, thereby reducing the number of obedience constraints to $2\anum-2$. Furthermore, \autoref{assumption:mcv} also implies that, given an action $\aaction_\aindex$, the adjacent obedience constraints, $(\aaction_\aindex,\aaction_{\aindex+1})$ and $(\aaction_\aindex,\aaction_{j-1})$, cannot simultaneously bind whenever $\aindex\in\{2,\dots,\anum-1\}$.

\paragraph{Dual formulation of \bce-consistency} Recall our goal is to identify a set of test functions for \autoref{eq:support}.
 Whereas the results in \autoref{sec:main} identified such functions by relying on properties of the Minkowski sum, the results in this section rely on the analysis of a problem dual to determining the feasibility of the equations that define \bce-consistency of \pair, Equations \ref{eq:obedience}, \ref{eq:state-mg}, and \ref{eq:action-mg}. 

Consider the problem of choosing $\joint\in\Delta(\Actions\times\Types)$ to maximize $0$ subject to Equations \ref{eq:obedience}, \ref{eq:state-mg}, and \ref{eq:action-mg}. This problem has value $0$ and hence, the system of equations in \autoref{definition:bce-c} is feasible, if and only if program \ref{eq:dual} below has nonnegative value:
\begin{align}\label{eq:dual}\tag{D}
&V_D\pair\equiv\min_{\pricev\in\realst,\priceav\in\reals^\Actions,\multvp\in\reals_{\geq0}^\Actions,\multvm\in\reals_{\geq0}^\Actions}\priceav\marginalv-\pricev\priorv\\
&\text{s.t.}\left\{\begin{array}{ll}
(\forall 1\leq j\leq\anum)(\forall 1\leq i\leq\tnum)&\pricea(\aaction_\aindex)\geq\price(\type_\tindex)+\multm_\aindex\diff(\aaction_\aindex,\aaction_{j-1},\type_\tindex)+\multp_\aindex\diff(\aaction_\aindex,\aaction_{\aindex+1},\type_\tindex)\\
(\forall j\in\{2,\dots,\anum-1\})&\multp_\aindex\multm_\aindex=0
\end{array}\right..\nonumber
\end{align}
In the above program, the vectors \pricev\ and \priceav\ are the Lagrange multipliers on the constraints \ref{eq:state-mg} and \ref{eq:action-mg}, respectively. That the notation for the multiplier on \ref{eq:state-mg} coincides with that of the vectors in \autoref{eq:support} is not a coincidence: the analysis that follows shows the dual variables \pricev\ that solve this program are intimately related to the test functions in \autoref{eq:support}. The vectors \multvp\ and \multvm\ are the multipliers on the adjacent obedience constraints: conditional on a recommendation to take $\aaction_\aindex$, $\multp_\aindex$ is the multiplier on the upward-looking constraint that $\aaction_\aindex$ is better than $\aaction_{\aindex+1}$, and $\multm_\aindex$ is the multiplier on the downward-looking constraint that $\aaction_{j}$ is better than $\aaction_{j-1}$. Finally, because these constraints cannot simultaneously bind, complementary slackness implies the condition $\multp_\aindex\multm_\aindex=0$ must hold at a solution. Below, we follow the convention that $\multm_1=\multp_{\carda}=0$.


We make two observations.\footnote{Recall that we are assuming no action in $\Actions_+$ is strictly dominated. Absent this assumption, the value of the dual is $-\infty$ and hence, the primal is unfeasible.} First, program \ref{eq:dual} is always feasible because we can always set to $0$ the coordinates of each of the vectors \pricev, \priceav, \multvp, and \multvm\ and still satisfy the conditions of the program. Therefore, a solution exists. Second, in any solution to program \ref{eq:dual}, the dual variable \pricev\ must be of the form
\begin{align}\label{eq:basic-p-identity}
p(\type)=\min_{\aaction_\aindex}\left[q(\aaction_\aindex)+\multp_\aindex\diff(\aaction_{\aindex+1},\aaction_{j},\type)-\multm_\aindex\diff(\aaction_\aindex,\aaction_{j-1},\type)\right],
\end{align}
regardless of the choice of \priceav, \multvp, and \multvm. \autoref{proposition:mcv} below provides a building block for the rest of this section. It shows that in analyzing the value of program \ref{eq:dual}, we need only consider solutions in which either \multvp\ or \multvm\ is zero. In other words, we need only consider solutions in which either all adjacent upward-looking or all adjacent downward-looking constraints are non-binding. By \autoref{eq:basic-p-identity}, this in turn has implications for the directions \pricev\ at which testing \autoref{eq:support} is sufficient.

To state \autoref{proposition:mcv}, define
\begin{align}\label{eq:mcv-test}
    \Prices^\uparrow&\equiv\left\{\pricev\in\realst:(\exists\priceav\in\realsa)(\exists\multvp\in\realsa_{\geq0})(\forall\type\in\Types)\price(\type)=\min_{\aindex}\left[\pricea_\aindex+\multp_\aindex  \diff(\aaction_{\aindex+1},\aaction_\aindex,\type)\right]\right\},\\
    \Prices^\downarrow&\equiv\left\{\pricev\in\realst:(\exists\priceav\in\realsa)(\exists\multvm\in\realsa_{\geq0})(\forall\type\in\Types)\price(\type)=\min_{\aindex}\left[\pricea_\aindex-\multm_\aindex  \diff(\aaction_{j},\aaction_{j-1},\type)\right]\right\}.\nonumber
\end{align}
In words, directions $\pricev\in\Prices^\uparrow$ correspond to directions for which the downward-looking constraints are nonbinding, whereas directions $\pricev\in\Prices^\downarrow$ correspond to directions for which the upward-looking constraints are nonbinding. Increasing differences implies that every vector in $\Prices^\uparrow$ is an increasing function of \type, whereas every vector in $\Prices^\downarrow$ is a decreasing function of $\type$. 

\begin{lemma}[Monotone and concave decision problems]\label{proposition:mcv}
    \pair\ is \bce-consistent given \payoff\ if and only if  \autoref{eq:support} holds for all $\pricev\in\Prices^\uparrow\cup\Prices^\downarrow$.
\end{lemma}

\autoref{proposition:mcv} identifies the vectors in $\Pricesup\cup\Pricesdown$ as test functions for the \bce-consistency of \pair\ given \payoff. From an information design perspective, this result uncovers an interesting property of the \emph{extremal} information structures that implement a given action distribution in monotone and concave decision problems. From the point of view of characterizing the extreme points of the set of joint distributions $\joint(\aaction,\type)$ that obediently implement a given \marginal, \autoref{proposition:mcv} says we can restrict attention to those in which either no downward-looking obedience constraint binds, or no upward-looking obedience constraint binds.

In the following sections, we use \autoref{proposition:mcv} together with additional assumptions on the decision problem to identify a finite subset of the test functions in $\Prices^\uparrow\cup\Prices^\downarrow$.

\subsection{Affine utility differences}\label{sec:affine}
Throughout this section, we assume the utility function satisfies the following condition:
\begin{definition}[Affine utility differences]\label{definition:aud}
The decision problem has affine utility differences if \payoff\ is concave$^*$ and vectors $\diffv\in\realst$, $\slopev\in\reals_{>0}^{\carda-1}$, and $\cttv\in\reals^{\carda-1}$ exist such that for all $\aindex\in\{1,\dots,\carda-1\}$, 
\begin{align}\label{eq:aud}\tag{AUD}
\payoff(\aaction_{\aindex+1},\type)-\payoff(\aaction_\aindex,\type)=\slope_\aindex \diff(\type)+\ctt_\aindex.
\end{align}    
\end{definition}

We make three observations. First, in the case of binary actions, affine utility difference entails no loss of generality. Second, by reordering the states, it is without loss of generality to assume $\diff$ is increasing in \type. Hence, decision problems with affine utility differences satisfy increasing differences. Thus, decision problems with affine utility differences are monotone and concave decision problems. Finally, we note the connection between affine utility differences when \Actions\ and \Types\ are finite sets, and quadratic loss in the general case in which they are compact, convex subsets of \reals. The analogue to the payoff difference across adjacent actions in quadratic loss is the derivative of the utility function $-(\aaction-\type)^2$ with respect to \aaction:
\begin{align*}
\frac{\partial}{\partial\aaction}\left[-\left(\aaction-\type\right)^2\right]=-2(\aaction-\type)=2\type-2\aaction,
\end{align*}
which is analogous to condition \ref{eq:aud} when \diff\ is the identity, $\slope=2$, and $\ctt=-2\aaction$. Below, this connection becomes apparent once we note the analogy between \autoref{proposition:aud} and the characterization of feasible action distributions under quadratic loss as mean-preserving contractions of the prior (cf. \citealp{strassen1965existence}).
%
%

\paragraph{Test functions for affine utility differences} \autoref{proposition:aud} characterizes the test functions under the assumption of affine utility differences. We first state the result and then provide intuition for it. To state the result, we first define a family of vectors $\pricev\in\realst$ and $\priceav\in\realsa$, which below play the role of the test functions and the value of the support function at those test functions, respectively.

For any $\candt\in\Types$, define the vectors $\pricevup_{\candt},\pricevdown_{\candt}\in\realst$ as
\begin{align*}
\priceup_{\candt}(\type)&=\min\left\{\diff(\type),\diff(\candt)\right\}\;\;
\pricedown_{\candt}(\type)=\min\left\{-\diff(\type),-\diff(\candt)\right\}.
\end{align*}
and let $\Prices_{\aud}=\{\pricevup_{\candt},\pricevdown_{\candt}:\candt\in\Types\}$. The notation highlights that $\pricevup_{\candt}$ and $\pricevdown_{\candt}$ are elements of the sets $\Prices^\uparrow$ and $\Prices^\downarrow$, respectively (cf. \autoref{eq:mcv-test}). Note the vectors $\pricevup_{\type_\tnum}$ and $\pricevdown_{\type_1}$ correspond to the (normalized) utility differences, while the vectors $\pricevup_{\type_1}$ and $\pricevdown_{\type_\tnum}$ are proportional to the vectors $\{\mathbf{1},-\mathbf{1}\}$, where $\mathbf{1}$ is the vector with $1$ in every coordinate.

To each vector in $\Prices_{\aud}$, associate the vectors $\priceavup_{\candt},\priceavdown_{\candt}\in\realsa$, defined as follows:
\begin{align*}
(\forall\aindex\in\{1,\dots,\carda-1\})\priceaup_{\candt}(\aaction_\aindex)&=\min\{\diff(\candt),-\kappa_\aindex/\gamma_\aindex\}\text{ and } \priceaup_{\candt}(\aaction_\anum)=\diff(\candt),\\
(\forall\aindex\in\{2,\dots,\carda\})\priceadown_{\candt}(\aaction_\aindex)&=\min\{-\diff(\candt),\kappa_{j-1}/\gamma_{j-1}\}\text{ and } \priceadown_{\candt}(\aaction_1)=-\diff(\candt).
\end{align*}
\begin{theorem}[Affine utility differences]\label{proposition:aud}
Assume the decision problem has affine utility differences. The pair $\pair\in\Posteriors\times\Delta(\Actions)$ is \bce-consistent given \payoff\ if and only if \autoref{eq:support} holds for all $\pricev\in\Prices_{\aud}$, that is, if and only if for all $\candt\in\Types$,
\begin{align}\label{eq:prop-aud}
\priceavup_{\candt}\marginalv\geq\pricevup_{\candt}\priorv\text{ and }\priceavdown_{\candt}\marginalv\geq\pricevdown_{\candt}\priorv.
\end{align}

%
%
\end{theorem}
\autoref{proposition:aud} reduces the question of whether \pair\ is \bce-consistent given \payoff\ to checking at most $2\cardt$ linear inequalities, together with $\cardt+1$  constraints implied by $\prior\in\Posteriors$.\footnote{To be precise, checking that the coordinates of the prior add up to one is implied by \autoref{eq:prop-aud} evaluated at $\pricevup_{\type_1}$ and $\pricevdown_{\type_\tnum}$.} Interestingly, the slopes, $\slope_\aindex$, and constants, $\ctt_\aindex$, determine the heights of the halfspaces that define \ppmpair, but not their normal vectors. 

To provide some intuition for why $\Prices_{\aud}$ is the minimal set of test functions, consider the case of quadratic utility. In this case, we know \pair\ is \bce-consistent given \payoff\ if and only if \prior\ dominates \marginal\ in the convex order. In other words, if and only if, for all concave functions $f$, the expected value of $f$ under \marginal\ dominates that under \prior. Any concave function can be obtained as affine combinations of the functions $\min\{\type-c,0\}$ and $\min\{c-\type,0\}$ for different values of $c\in\reals$ \citep{border1991functional}. In fact, testing that the expected value of these functions under \marginal\ is greater than that under \prior\ is enough to determine whether \prior\ dominates \marginal\ in the convex order. (The representation of the convex order via the comparison of the cumulative distributions of \marginal\ and \prior\ comes from this finding.) The test functions $\priceup_{\candt}$ and $\pricedown_{\candt}$ are the analogue of the $\min\{\type-c,0\}$ and $\min\{c-\type,0\}$ functions in the case in which \diff\ is not the identity. 

We illustrate \autoref{proposition:aud} in the context of \autoref{example:counter}:    
    \setcounter{example}{1}
 \begin{example}[continued]\label{example:counter-aud}
Recall that in this example, the utility difference between $\aaction_2$ and $\aaction_1$ is given by
\begin{align*}
\diffv=\payoff(\aaction_2,\cdot)-\payoff(\aaction_1,\cdot)=(-9,-5,-1,5),
\end{align*}
where $\ctt=0$ and $\slope=1$. \autoref{proposition:aud} implies the test functions are given by:
\begin{align*}
\Prices_{\aud}&=\left\{\begin{array}{l}(-9,-5,-1,5),(-9,-5,-1,-1),(-9,-5,-5,-5),(-9,-9,-9,-9),\\(9,5,1,-5),(5,5,1,-5),(1,1,1,-5),(-5,-5,-5,-5)\end{array}\right\}\\
 \priceavup_{\candt}&\in\{(0,5),(-1,-1),(-5,-5),(-9,-9)\}\\
\priceavdown_{\candt}&\in\{(9,0),(5,0),(1,0),(-5,-5)\}.
\end{align*}
The first line in the definition of $\Prices_{\aud}$ corresponds to $\pricevup_{\candt}$\textemdash in decreasing order of the states\textemdash and the second line corresponds to $\pricevdown_{\candt}$. The first vectors in each line of $\Prices_{\aud}$ are the utility differences, so that the inequalities evaluated at those directions correspond to the payoff martingale equations \eqref{eq:bce-c-actions}. 

In \autoref{sec:main}, we numerically computed the set \ppmpair\ and concluded the characterization in \autoref{proposition:2-3} was missing a facet, defined by the vector $(0,-2,-5)$, which in $\reals^4$ we can take to be $(0,0,-2,-5)$. Consider now the test function $\pricevdown_{\type_2}=(5,5,1,-5)$ and note that this is an affine transformation of $(0,0,-2,-5)$; indeed, $(0,0,-2,-5)=0.5*(5,5,1,-5)-5/2$.  

 \autoref{proposition:aud} allows us to recover the set \ppmpair\ in \autoref{example:counter}. Indeed, using the test functions $\Prices_{\aud}$ and the constraints that arise from $\prior$ being a probability distribution, we obtain that \ppmpair\ is defined by the following inequalities:
 \begin{align}
9\belief_1+5\belief_2+\belief_3-5\belief_4&\leq 9\marginal(\aaction_1)\nonumber\\
-9\belief_1-5\belief_2-\belief_3+5\belief_4&\leq 5\marginal(\aaction_2)\nonumber\\
\belief_1+\belief_2+\belief_3-5\belief_4&\leq\marginal(\aaction_1)\nonumber\\
5\belief_1+5\belief_2+\belief_3-5\belief_4&\leq 5\marginal(\aaction_1)\;\;\left(\star\right)\nonumber\\
-\belief_1\leq0,\; -\belief_2\leq0,\; -\belief_3&\leq0\nonumber
\\
 \belief_1+\belief_2+\belief_3+\belief_4&=1.\nonumber
 \end{align}
The first four inequalities correspond to the test functions $\price\in\Prices_{\aud}$, with equation $\left(\star\right)$ corresponding to the new facet relative to the representation in \autoref{proposition:2-3}. The final set of inequalities corresponds to the probability constraints. We do not include the non-negativity constraint on $\belief_4$ because it is implied by the third inequality and the probability constraints. In fact, the third inequality coincides with the belief martingale equation at $\type_4$.
    \end{example}

\paragraph{Binary actions} When the \dm\ only has two actions, assuming affine utility differences, together with $\slope=1$ and $\ctt=0$, is without loss.  Thus, \autoref{proposition:aud} characterizes the set of \bce-consistent distributions for all decision problems with binary actions. We record the corresponding characterization in \autoref{corollary:bin} below, where we take advantage of the binary action assumption to provide more explicit expressions for the conditions in \autoref{eq:prop-aud}. 

When $\Actions=\{\aaction_1,\aaction_2\}$, \autoref{proposition:aud} implies \pair\ is \bce-consistent given \payoff\ if and only if for all $\candt\in\Types_+\equiv\{\type\in\Types:\diff(\type)>0\}$, 
\begin{align*}
    \marginal(\aaction_1)&\leq \sum_{\type\leq\candt} \prior(\type)\left[1-\frac{\diff(\type)}{\diff(\candt)}\right],
\intertext{and for all $\candt\in\Types_-\equiv\{\type\in\Types:\diff(\type)<0\}$,}
    \marginal(\aaction_2)&\leq\sum_{\type\geq\candt}\prior(\type) \left[1-\frac{\diff(\type)}{\diff(\candt)}\right],
\end{align*}
where in the above expressions, $\type\geq\candt$ and $\type\leq\candt$ signify states with higher and lower indices than \candt, respectively.

We can refine the above expressions by figuring out the states \candt\ that minimize the right-hand side of these inequalities. To this end, define 
\begin{align*}
\type_{\aaction_1}(\prior)&=\min\{\candt\in\Types_+:\sum_{\type\leq\candt} \prior(\type)\diff(\type)>0\},\;\;
\type_{\aaction_2}(\prior)=\max\{\candt\in\Types_-:\sum_{\type\geq\candt} \prior(\type)\diff(\type)<0\},
\end{align*}
whenever these sets are nonempty. Otherwise, let $\type_{\aaction_1}(\prior)=\type_\tnum$ when the first set is empty, and let $\type_{\aaction_2}(\prior)=\type_1$ when the second is empty. We note the following: First, in the above expressions, $\min$ and $\max$ are over the state indices. Second, unless the \dm\ is indifferent between both actions at \prior, at least one of the two sets is nonempty. To see this, suppose that $\aaction_1$ is uniquely optimal at the prior; hence, the first set is empty. Then, $\type_1$ satisfies the conditions defining the set on the right-hand side.

To understand the definition of $\type_{\aaction_1}(\prior)$, consider the information structures that maximize the probability the \dm\ takes $\aaction_1$. Intuitively, such an information structure should pool states below a threshold state, \candt, and this threshold state is an element of $\Types_+$. Then, $\type_{\aaction_1}(\prior)$ is the smallest state such that the recommendation to take $\aaction_1$ when states below $\type_{\aaction_1}(\prior)$ are pooled is disobedient. In other words, the information structure that maximizes the probability of taking $\aaction_1$ pools states strictly below $\type_{\aaction_1}(\prior)$ with probability $1$, and $\type_{\aaction_1}(\prior)$ with a probability determined by the \dm's binding obedience constraints.  The intuition for $\type_{\aaction_2}(\prior)$ is similar. With these definitions, the characterization for the case of binary actions is as follows:

\begin{corollary}[Binary actions]\label{corollary:bin}
Suppose $\carda=2$. The pair \pair\ is  \bce-consistent given \payoff\ if and only if $\marginal(\aaction_2)\in[\mathrm{LB}(\prior,\diffv),\mathrm{UB}(\prior,\diffv)]$, where
\begin{align}
\mathrm{LB}(\prior,\diffv)&\equiv\max\left\{0,1-\sum_{\type\leq\type_{\aaction_1}(\prior)}\prior(\type)\left[1-\frac{\diff(\type)}{\diff(\type_{\aaction_1}(\prior))}\right]\right\},\label{eq:lb}
\intertext{and}
\mathrm{UB}(\prior,\diffv)&\equiv\min\left\{1,\sum_{\type\geq\type_{\aaction_2}(\prior)}\prior(\type)\left[1-\frac{\diff(\type)}{\diff(\type_{\aaction_2}(\prior))}\right]\right\}.\label{eq:ub}
\end{align}
    \end{corollary}    
Whenever $\aaction_1$ is optimal at the prior, the right-hand side of \autoref{eq:lb} is $0$. Instead, when $\aaction_2$ is optimal at the prior, the right-hand side of \autoref{eq:ub} is 1. Consequently, when the \dm\ is indifferent between both actions at the prior, all action distributions can be rationalized.

%

\subsection{Two-step utility differences}\label{sec:two-step}
\begin{definition}[Two-step utility differences]\label{definition:2-step}
    The decision problem has two-step utility differences if \payoff\ is concave$^*$ and for all $j\in\{1,\dots,\carda-1\}$, $\diff(\aaction_{\aindex+1},\aaction_\aindex,\type)$ takes exactly two values $\lowdj<0<\highdj$.\end{definition}
By reordering the states, decision problems with two-step utility differences satisfy increasing differences. Thus, decision problems with two-step utility differences are monotone and concave decision problems. Decision problems with two-step utility differences are pinned down by the values the vector $\diff(\aaction_{\aindex+1},\aaction_\aindex,\cdot)$ takes on $\type_1$ and $\type_\tnum$, and the state at which $\diff(\aaction_{\aindex+1},\aaction_\aindex,\cdot)$ switches between those values. In what follows, we denote by $\candi(\aindex)$ the highest index state at which $\diff(\aaction_{\aindex+1},\aaction_\aindex,\cdot)$ coincides with $\lowdj$. Concavity implies $\candi(\aindex)$ is increasing in \aindex. 

The discrete analog of the absolute loss function is a limiting case of two-step utility differences.\footnote{To be sure, this special case does not satisfy the second requirement of concavity*. We can instead consider a perturbed decision problem with $\tilde{d}(\aaction_{\aindex+1},\aaction_\aindex,\type)={d}(\aaction_{\aindex+1},\aaction_\aindex,\type)-j\varepsilon$ for some very small $\varepsilon>0$. This perturbed decision problem has two-step utility differences, so our characterization below applies.}\footnote{\cite{yang2024monotone} and \cite{kolotilin2024distributions} characterize the distributions over posterior quantiles consistent with a prior distribution over the states.}
 To see this, let $\Types_\aindex$ denote the set of states in which action \aindex\ is optimal, $\Types_\aindex=\{\type:(\forall\aaction\in\Actions)\payoff(\aaction_\aindex,\type)\geq\payoff(\aaction,\type)\}$. Under \autoref{assumption:mcv}, these sets are ordered in that $\Types_\aindex$ lies to the left of $\Types_{\aindex+1}$\textemdash in terms of the state indices. We can then define
\begin{align*}
\diff(\aaction_{\aindex+1},\aaction_\aindex,\type)=-c\mathbbm{1}[\type\in\Types_k,k\leq j]+c\mathbbm{1}[\type\in\Types_k,k\geq \aindex+1],
\end{align*}
or $\payoff(\aaction_\aindex,\type)=-c|k-j|$ for $\type\in\Types_k$.

%

\paragraph{Test functions for two-step utility differences} \autoref{proposition:2-step} characterizes the test functions under the assumption of two-step utility differences. In particular, we show the test functions for two-step utility differences are basically those in $\Prices_{\opt}$. 

For each $\aindex\in\Actions$, we define $\priceup_\aindex(\type)=\diff(\aaction_{\aindex+1},\aaction_j,\type)$ and $\pricedown_\aindex=-\diff(\aaction_{\aindex+1},\aaction_{j},\type)$. Denote by $\Prices_2$ the set of vectors $\{\pricevup_j,\pricevdown_j:j\in\{1,\dots,\anum-1\}\}$ and note it is a subset of $\Prices_{\opt}$\textemdash it is the left-most set in \autoref{eq:test-small}. To each of these vectors, associate the vectors $\priceavup_\aindex,\priceavdown_\aindex\in\realsa$, defined as follows:
\begin{align*}
\priceaup_\aindex(\aaction_k)&=\max_{\beliefv\in\opt(\aaction_k)}\beliefv\diff(\aaction_{\aindex+1},\aaction_j,\cdot)=\highdj-\frac{\highdj-\lowdj}{\highdk-\lowdk}\highdk\mathbbm{1}[\candi(k)\leq\candi(j)],\\
\priceadown_\aindex(\aaction_k)&=\max_{\beliefv\in\opt(\aaction_k)}-\beliefv\diff(\aaction_{\aindex+1},\aaction_j,\cdot)=-\lowdj+\frac{\highdj-\lowdj}{\highd_{k,k-1}-\lowd_{k,k-1}}\lowd_{k,k-1}\mathbbm{1}[\candi(k-1)\geq\candi(j)].
\end{align*}
For instance, in the case of absolute loss,  $\priceaup_\aindex(\aaction_k)=c\mathbbm{1}_{i^\star(k)>i^\star(j)}$ and $\priceadown_\aindex(\aaction_k)=-c\mathbbm{1}_{i^\star(k-1)<i^\star(j)}$.

\begin{theorem}[Two-step utility differences]\label{proposition:2-step}
Assume the decision problem has two-step utility differences. The $\pair\in\Posteriors\times\Delta(\Actions)$ is \bce-consistent given \payoff\ if and only if \autoref{eq:support} holds for all $\pricev\in\Prices_2$. That is, if and only if for all $j\in\{1,\dots,\anum-1\}$, 
\begin{align}\label{eq:2-step}
    \priceavup_\aindex\marginalv\geq\diffv(\aaction_{\aindex+1},\aaction_\aindex)\priorv\geq-\priceavdown_\aindex\marginalv.
\end{align}
\end{theorem}
\autoref{proposition:2-step} reduces the question of whether \pair\ is \bce-consistent given \payoff\ to checking $2(\carda-1)$ linear inequalities, together with $\cardt+1$  constraints implied by $\prior\in\Posteriors$.

Two-step utility differences is a case in which (i) the vectors in $\Prices_{\opt}$ are test functions, without restricting the cardinality of the state space, and (ii) the belief martingale equations are implied by either the payoff martingale equations or the non-negativity condition that \priorv\ is a distribution. We obtain (i) because of the simple structure of the utility differences. Implicitly, \autoref{proposition:2-step} shows that under two-step utility differences no new extreme rays are generated when we intersect the normal cones of \opta. Regarding (ii), fix a state $\type_i$ and consider the belief martingale condition for that state:
\begin{align*}
\sum_{\aaction\in\Actions}\marginal(\aaction)\min_{\belief\in\opta}\belief(\type_i)\leq\prior(\type_i).
\end{align*}
Recall that $\Types_j=\{\type_{\candi(j-1)+1},\dots,\type_{\candi(j)}\}$. If $i\neq\candi(j)$ for any $j$, then the left-hand side of the above expression is $0$, so that the above equation reduces to the non-negativity constraint on $\prior(\type_i)$. The same is true when $i=\candi(j)$ for some $j$, but $\Types_j$ is not a singleton. Instead, when $i=\candi(j)$ for some $j$ and $\Types_j$ is a singleton, then one can verify the above equation is implied by \autoref{eq:2-step}.

We conclude this section with \autoref{remark:foa}, which summarizes the extension of our results to the case in which \Types\ and \Actions\ are compact, convex, Polish spaces and the first-order approach applies (cf. \citealp{kolotilin2023persuasion}). Readers interested in applications can jump to \autoref{sec:applications}, with little loss of continuity.

\begin{remark}[\bce-consistency under first-order approach]\label{remark:foa}
In \autoref{sec:foa}, we extend the results in this section to the case in which \Types\ and \Actions\ are compact Polish spaces and the first-order approach applies. In this case, the set \opta\ is the set of all beliefs such that the \dm's first-order condition holds at \aaction. In a slight abuse of notation, letting $\diff(\aaction,\type)$ denote the partial derivative of \payoff\ with respect to the first coordinate, the set \opta\ is the set of beliefs such that $\mathbb{E}_\belief[\diff(\aaction,\type)]=0$. \citet[Theorem 3]{strassen1965existence} implies the support function of the set \ppmpair\ is defined as in \autoref{eq:support}, after appropriately replacing sums with integrals. In \autoref{sec:foa}, we show \autoref{proposition:mcv} extends, and use this result to recover the characterizations of \bce-consistency in a generalization of quadratic loss. 
\end{remark}

\section{Applications}\label{sec:applications}
In this section, we apply our results (i) to elucidate comparative statics of the set of \bce-consistent marginals (\autoref{sec:cs}) and (ii) to study \bce-consistency across decision problems (\autoref{sec:across}). \autoref{sec:games} further illustrates our results in the context of simple multi-agent settings.

\subsection{Comparative statics with affine utility differences}\label{sec:cs}
Under the assumption of affine utility differences, we consider in this section changes to the prior or the utility function that preserve the rationalization of a given marginal distribution over actions. 

\paragraph{Changes to the prior} Suppose \pair\ is \bce-consistent given \payoff\ and let \priorb\ denote another prior belief. When can we say $(\priorb,\marginal)$ are \bce-consistent given \payoff? The following definition is key:
\begin{definition}[$\diff$-mean-preserving spread] The prior \priorb\ is a \diff-mean-preserving spread of  \prior\ if $\priorb\circ\diff^{-1}$ is a mean-preserving spread of $\prior\circ\diff^{-1}$.
\end{definition}
Note $\prior\circ\diff^{-1}$ is the distribution of payoffs\textemdash measured by the utility difference \diff\textemdash faced by the \dm\ under the prior distribution \prior. When \priorb\ is a \diff-mean-preserving spread of \prior, the \dm\ faces a riskier payoff distribution under \priorb\ than under \prior. Conversely, we can obtain the payoff distribution under \prior\ by a garbling of the payoff distribution under \priorb. A fortiori, if \pair\ is \bce-consistent given \payoff, so is $(\priorb,\marginal)$.

\autoref{proposition:aud-cs-mps} summarizes this discussion:
\begin{proposition}\label{proposition:aud-cs-mps}
Suppose \payoff\ satisfies affine utility differences and that \pair\ is \bce-consistent given \payoff.  If \priorb\ is a \diff-mean preserving spread of \prior, $(\priorb,\marginal)$ is \bce-consistent given \payoff.
\end{proposition}
The result follows from the shape of the test functions $\Prices_{\aud}$ and noting that if \priorb\ is a \diff-mean preserving spread of \prior, then for all \candt, $\pricevup_{\candt}\priorb\leq\pricevup_{\candt}\prior$ and $\pricevdown_{\candt}\priorb\leq\pricevdown_{\candt}\prior$.

We illustrate \autoref{proposition:aud-cs-mps} with the following example:
\begin{example}[Safe vs. risky project]\label{example:2-sided}
The \dm\ is choosing between a safe ($\aaction_1$) and a risky ($\aaction_2$) project. Let $\Types_+,\Types_0,$ and $\Types_-$ denote the set of states for which the risky project dominates, is payoff equivalent to, and is dominated by the safe project, respectively. Then, the \dm's utility difference equals $1$ on $\Types_+$, $0$ on $\Types_0$, and $-1$ otherwise. For a given prior \prior, the cumulative distribution of $\prior\circ\diff^{-1}$ is as follows:
\begin{align*}
F_{\prior\circ\diff^{-1}}(x)=\left\{\begin{array}{ll}0&\text{ if }x<-1\\
\prior(\Types_-)&\text{ if }x\in[-1,0)\\
1-\prior(\Types_+)&\text{ if }x\in[0,1)\\
1&\text{otherwise}\end{array}\right..
\end{align*}
A prior \priorb\ is a \diff-mean-preserving spread of \prior\ if $\priorb(\Types_+)-\priorb(\Types_-)=\prior(\Types_+)-\prior(\Types_-)$ and $\priorb(\Types_+)\geq\prior(\Types_+)$. That is, under \priorb, the \dm\ assigns more probability to the extreme states in which the risky project pays off or fails, while maintaining the prior difference between the probability that the risky project pays off or fails. By \autoref{proposition:aud-cs-mps}, if \pair\ is \bce-consistent given \payoff, so is $(\priorb,\marginal)$.
\end{example}


\paragraph{Binary actions and payoff shifters} Consider now the case of binary actions, $\Actions=\{\aaction_1,\aaction_2\}$, so that without loss, we can take $\payoff(\aaction_2,\type)-\payoff(\aaction_1,\type)=\diff(\type)$. Suppose we parameterize the utility differences by $\diff(\cdot,\param)$ such that $\param<\paramb$ implies $\diff(\cdot,\param)\leq\diff(\cdot,\paramb)$. That is, as we move from \param\ to \paramb, action $\aaction_2$ becomes more attractive.

\autoref{proposition:aud-cs-shift} characterizes how the lower and upper bounds on $\marginal(\aaction_2)$ change as the utility differences shift in favor of the higher action.
\begin{proposition}[Payoff shifters]\label{proposition:aud-cs-shift}
Assume $\Actions=\{\aaction_1,\aaction_2\}$ and fix a prior $\prior$. If $\paramb>\param$, then the lower and upper bounds on $\marginal(\aaction_2)$ are smaller under \param\ than under \paramb. That is, $\mathrm{LB}(\prior,\diffv(\param))\leq\mathrm{LB}(\prior,\diffv(\paramb))$ and $\mathrm{UB}(\prior,\diffv(\param))\leq\mathrm{UB}(\prior,\diffv(\paramb))$.
\end{proposition}
The result is intuitive: as $\aaction_2$ becomes more attractive, both the minimal probability with which the \dm\ must take $\aaction_2$ and the maximal probability with which the \dm\ may take $\aaction_2$ so that \marginal\ is consistent with information are higher. 

The result has the following implication: suppose \pair\ is \bce-consistent given $\diff(\cdot,\param)$ for some real-valued parameter \param. Then, \pair\ is \bce-consistent given $\diff(\cdot,\paramb)$ for all $\paramb\in[\underline{\param}_{\prior},\overline{\param}_{\prior}]$, where $\mathrm{UB}(\prior,\diffv(\underline{\param}_{\prior}))=\marginal(\aaction_2)$ and $\mathrm{LB}(\prior,\diffv(\overline{\param}_{\prior}))=\marginal(\aaction_2)$. That is, for each \prior\ such that \pair\ is \bce-consistent for some parameter, we can identify an interval of parameter values for which \pair\ is \bce-consistent.

We illustrate \autoref{proposition:aud-cs-shift} with \autoref{example:shift}, based on \cite{bergemann2022counterfactuals}:
\begin{example}[Payoff shifters]\label{example:shift}
Suppose $\Actions=\{\aaction_1,\aaction_2\}$, and label the states as $\Types=\{-1,1\}$. Normalizing the payoff from action $\aaction_1$ to $0$, the parameterized utility of action $\aaction_2$ equals the utility difference across the actions, $\diff(\type,\param)=\type+\param$. In what follows, suppose $\param\in[-1,1]$. When $\param>1$, $\aaction_2$ is dominant and the only rationalizable marginal is $\marginal=(0,1)$. Similarly, when $\param<-1$, $\aaction_1$ is dominant and the only rationalizable marginal is $\marginal=(1,0)$. 

\autoref{proposition:2-3} and \autoref{corollary:bin} both imply \pair\ is \bce-consistent given \payoff\ if and only if 
\begin{align}\label{eq:single}
\prior(1)\in\left[\marginal(\aaction_2)\frac{1-\param}{2},\frac{1-\param}{2}+\marginal(\aaction_2)\frac{1+\param}{2}\right].
\end{align}
Interestingly, for a fixed probability that the \dm\ takes $\aaction_2$, $\marginal(\aaction_2)$, the lower and the upper bounds on the prior probability the state is $1$ fall as \param\ increases. As $\aaction_2$ becomes more attractive, the \dm\ can be less optimistic about $\type=1$ and still take $\aaction_2$ with the same probability. Similarly, as $\aaction_2$ becomes more attractive, the \dm\ must be less optimistic about the high state than before in order to still take $\aaction_2$ with the same probability.
%
%
\end{example}

\paragraph{Binary actions and ratio-ordered payoffs} Say the parameterized utility difference $\diff(\cdot,\param)$ is \emph{ratio ordered} if for $\type<\typeb$, the ratio $\diff(\type,\param)/\diff(\typeb,\param)$ is increasing in \param. We have the following result:
\begin{proposition}[Ratio-ordered payoffs]\label{proposition:aud-cs-ratio}
Suppose the parameterized utility difference is ratio ordered. Then, the lower and upper bounds on the probability of $\aaction_2$ are increasing in \param.
\end{proposition}

We illustrate \autoref{proposition:aud-cs-ratio} with \autoref{example:ht}:
\begin{example}[Hypothesis testing]\label{example:ht}
The \dm\ wants to test the hypothesis that the state belongs in the set $\hat\Types$. Let $\Actions=\{\aaction_1,\aaction_2\}$, where $\aaction_1$ represents rejecting the hypothesis. Let $\costr(\param)>0$ and $\costa(\param)>0$ denote the costs of type $I$ and type $II$ errors. Then, the \dm's parameterized utility difference is given by
\begin{align*}
\diff(\type,\param)=\left\{\begin{array}{ll}\costr(\param) &\text{ if }\type\in\hat\Types\\
-\costa(\param)&\text{otherwise}\end{array}\right..
\end{align*}
The payoff differences are ratio-ordered if $\param<\paramb$ implies  $-\costa(\param)/\costr(\param)\leq-\costa(\paramb)/\costr(\paramb)$.

Theorems \ref{proposition:aud} and \ref{proposition:2-step} both imply \pair\ is \bce-consistent given \payoff\ if and only if the following holds:
\begin{align*}
\prior(\type\in\hat\Types)\left(1+\frac{\costa(\param)}{\costr(\param)}\right)-\frac{\costa(\param)}{\costr(\param)}\leq\marginal(\aaction_2)\leq\prior(\type\in\hat\Types)\left(1+\frac{\costr(\param)}{\costa(\param)}\right).
\end{align*}
As we increase \param, the lower and upper bounds on the probability of action $\aaction_2$ go up, reflecting that as \param\ increases, the relative cost of not rejecting the hypothesis decreases. Conversely, for a fixed probability of taking action $\aaction_2$, the lower and upper bounds on the prior probability the hypothesis is correct decrease, for similar reasons to those in \autoref{example:shift}. For instance, as \param\ increases, not rejecting the hypothesis becomes more attractive; hence, the \dm\ is willing to not reject it at lower values of $\prior(\type\in\hat\Types)$.
%
%
\end{example}

\subsection{\bce-consistency across decision problems}\label{sec:across}
Suppose the analyst observes the \dm's choices across $N$ different decision problems, each indexed by a set of actions $\Actionsi$ and utility function $\payoffi:\Actionsi\times\Types\to\reals$. The analyst's data are now the joint distribution over actions across decision problems, which we denote by $\marginalb\in\Delta(\Actions_1\times\dots\Actions_N)$. Our results so far allow us to understand whether $(\prior,\marginalb|_{\Actionsi})$ is \bce-consistent in decision problem $\decision_\playerindex=\langle\Types,\Actionsi,\payoffi\rangle$, where $\marginalb|_{\Actionsi}$ is the marginal of \marginalb\ over \Actionsi.  Instead, we now study when a \emph{single} information structure exists that rationalizes the choices \marginal\ made by the \dm\ across all decision problems. 

%
%
%
%
%
%

Consider now an auxiliary decision problem $\bar\decision=\langle\Types, \bar\Actions,\bar u\rangle$, where $\bar\Actions=\times_{i=1}^N\Actionsi$. In this decision problem,  choices are action \emph{profiles}, $\aaction\in\times_{\playerindex\in\nplayers}\Actionsi$, and payoffs are $\bar u(\aaction,\type)=\sum_{\playerindex=1}^\nplayers\payoffi(\actioni,\type)$. The results in \cite{bergemann2022counterfactuals} imply that for a given prior distribution \prior, an information structure exists that rationalizes the \dm's choices \marginalb\ \emph{across} decision problems $\left(\decision_\playerindex\right)_{\playerindex\leq N}$ if and only if $(\prior,\marginalb)$ is \bce-consistent in the auxiliary decision problem $\bar\decision$:

\begin{proposition}[\bce-consistency across decision problems]\label{proposition:across}
Consider the collection of $N$ decision problems $\{D_n=\langle\Types,\Actionsi,\payoffi\rangle:n\in\{1,\dots,N\}\}$ and a joint distribution over actions $\marginalb\in\Delta(\times_{\playerindex=1}^N\Actionsi)$. The pair $(\prior,\marginalb)$ is \bce-consistent across decision problems $(\decision_\playerindex)$ if and only if $(\prior,\marginalb)$ is \bce-consistent in decision problem $\bar\decision$.
\end{proposition}
Consequently, the results in Sections \ref{sec:main} and \ref{sec:utility} can be used to study \bce-consistency across decision problems. 
We illustrate this result using a variation of \autoref{example:shift}:
\begin{example}\label{example:shift-b}
Consider the following variation of \autoref{example:shift}. Suppose we observe the \dm\ in two decision problems, the first with payoffs corresponding to $\param=0$ and the second corresponding to some $\param\in(0,1]$. Formally, in both decision problems, $\Actions_1=\Actions_2=\{\aaction_1,\aaction_2\}$ and payoff differences are as follows:
\begin{align*}
\payoff_1(\aaction_2,\type)-\payoff_1(\aaction_1,\type)&=\type,\\
\payoff_2(\aaction_2,\type)-\payoff_2(\aaction_1,\type)&=\type+\param.
\end{align*}
Let $\bar\payoff=\payoff_1+\payoff_2$. An action distribution is now $\marginalb\in\Delta\left(\{\aaction_1,\aaction_2\}^2\right)$. For \marginalb\ to be rationalizable via information, we need $\marginalb(\aaction_2,\aaction_1)=0$. Because $\param>0$, if the \dm\ has access to the same information structure under both decision problems, whenever they take $\aaction_2$ in the first problem, they should also take $\aaction_2$ in the second one. \autoref{proposition:across} implies $(\prior,\marginalb)$ is  \bce-consistent across decision problems if and only if
\begin{align}\label{eq:across}
\prior(1)\in\left[\marginalb(\aaction_1,\aaction_2)\frac{1-\param}{2}+\marginalb(\aaction_2,\aaction_2)\frac{1}{2},\frac{1-\param}{2}+\marginalb(\aaction_1,\aaction_2)\frac{\param}{2}+\marginalb(\aaction_2,\aaction_2)\frac{1+\param}{2}\right].
\end{align}
An easy way to see why this is the case is to note $\Delta_{\bar\payoff}^*(\aaction_1,\aaction_1)=[0,\nicefrac{1-\param}{2}]$, $\Delta_{\bar\payoff}^*(\aaction_1,\aaction_2)=[\nicefrac{1-\param}{2},\nicefrac{1}{2}]$, and $\Delta_{\bar\payoff}^*(\aaction_2,\aaction_2)=[\nicefrac{1}{2},1]$ and apply \autoref{corollary:2-state}. 

Compared with the single-decision-problem bounds we derived in \autoref{eq:single}, the bounds in \autoref{eq:across} adjust for the correlation in actions across decision problems. For illustration, suppose the probability that the \dm\ chooses $\aaction_2$ in at least one of the decision problems, $\marginalb(\aaction_1,\aaction_2)+\marginalb(\aaction_2,\aaction_2)$, coincides with the probability that they choose $\aaction_2$ if they were facing a single decision problem with payoff shifter \param, which is denoted by $\marginal(\aaction_2)$ in \autoref{eq:single}. Keeping this total probability fixed, consider the case in which $\marginalb(\aaction_1,\aaction_2)=0$; hence, the \dm\ chooses $\aaction_2$ in either decision problem only if they choose $\aaction_2$ in both decision problems. In this case, the upper bound on $\prior(1)$ is the same as that in \autoref{eq:single}, where the \dm\ is facing decision problem 2 alone, but the lower bound is higher, reflecting that the \dm\ must be optimistic enough about the high state to choose $\aaction_2$ in both decision problems. Similarly, consider the case in which $\marginalb(\aaction_2,\aaction_2)=0$; hence, the \dm\ only chooses $\aaction_2$ in the decision problem 2. In this case, the lower bound is the same as that in \autoref{eq:single}, where the \dm\ faces decision problem 2 alone, but the upper bound is smaller, reflecting the \dm\ is never optimistic enough to choose $\aaction_2$ in both decision problems.
\end{example}
Note we can interpret the joint distribution over action profiles in this section as coming from a game between $N$ players, where player \playerindex\ has action set \Actionsi\ and payoffs $\payoffi$. Consider now the question of whether we can find a public information structure that rationalizes \marginalb\ as if the players observe the realization of the public information structure \emph{prior} to non-cooperatively playing the game. As we show in \autoref{sec:games}, \autoref{proposition:across} characterizes the pairs $(\prior,\marginalb)$ that admit such rationalization. \autoref{sec:games} also applies our results to study ring-network games \citep{kneeland2015identifying}.

\section{Rationalizing information structures}\label{sec:core}
Whereas the results in the previous sections characterize the set of \pair\ that are \bce-consistent given \payoff, they remain silent as to the set of information structures that rationalize the marginal \marginal\ given the \dm's prior \prior\ and utility function \payoff.\footnote{In a sense, this observation is consistent with the paper's motivation: the applied literature treats information as a nuisance parameter and hence, is not necessarily interested in estimating the (parameters of the) information structure.} In this section, we tackle this problem in the single-agent setting, where an information structure can be identified with the distribution over posteriors $\bsplit\in\Delta(\Posteriors)$ it induces \citep{kamenica2011bayesian}. Given $\pair$ and \payoff, we characterize in this section the set of distributions over posteriors (if any) that rationalize the \dm's distribution over actions. Without loss of generality, the analysis that follows assumes the distribution over posteriors \bsplit\ has finite support \citep{myerson1982optimal,kamenica2011bayesian}; we denote the support of \bsplit\ by $\mathrm{supp}\,\bsplit$.

For a distribution over posteriors \bsplit\ to rationalize \marginal, two conditions must be satisfied. First, the mean of \bsplit\ must equal  the prior, \prior; that is, \bsplit\ must be Bayes plausible. Second, the distribution over posteriors \bsplit\ must induce the distribution over actions \marginal. Formally, let $\aaction^*(\belief)$ denote the \dm's optimal set of actions when their belief is \belief. Then, \bsplit\ induces \marginal\ if a decision rule $\strat:\Posteriors\to\Delta(\Actions)$ exists such that for every $\aaction\in\Actions$, 
\begin{align}\label{eq:demand-supply}\marginal(\aaction)=\sum_{\belief\in\Posteriors}\bsplit(\belief)\strat(\belief)(\aaction),\end{align}
where $\aaction\in\text{ supp }\strat(\belief)(\cdot)$ only if $\aaction\in\aaction^*(\belief)$. 

Whereas in the previous section we interpreted the marginal distribution \marginal\ as a Bayes plausible distribution over posteriors, \autoref{eq:demand-supply} shows this analogy is perhaps incomplete. Indeed, whenever the distribution over posteriors \bsplit\ induces beliefs such that $\aaction^*(\belief)$ is not a singleton, specifying the \dm's tie-breaking rule \strat\ is necessary to determine whether the frequency with which the \dm\ takes actions under \bsplit\ matches that under \marginal. 

\paragraph{A demand and supply problem} The problem of determining whether a decision rule \strat\ exists satisfying \autoref{eq:demand-supply} admits the following interpretation (cf. \citealp{gale1957theorem}): beliefs $\belief\in\mathrm{supp}\,\bsplit$ are supplied in quantities $\bsplit(\belief)$ and demanded by actions $\aaction\in\Actions$ in quantities $\marginal(\aaction)$. The demand $\marginal(\aaction)$ can only be satisfied by certain beliefs\textemdash those that satisfy $\belief\in\opta$. The decision rule \strat\ describes how much of a given belief $\belief\in\mathrm{supp}\,\bsplit$ is allocated to action \aaction. That \bsplit\ implements \marginal\ is equivalent to being able to allocate the supply of beliefs to satisfy the action demands in a \emph{market-clearing} way.

As we argue in \autoref{appendix:ext}, the above is an instance of the supply and demand problem studied in \cite{gale1957theorem}, with the distributions \marginal\ and \bsplit\ determining the demanded and supplied quantities, respectively. Building on the main theorem in that paper, \autoref{proposition:core-bp} below characterizes when \bsplit\ implements \marginal. To state \autoref{proposition:core-bp}, one final piece of notation is needed. Given a Bayes plausible  $\bsplit\in\Delta(\Posteriors)$, we construct a measure over subsets \Actionsb\ of the set of actions \Actions\ as follows. For each $\Actionsb\subseteq\Actions$, define the push-forward measure $\bsplit_\Actions(\Actionsb)$ as 
\begin{align}\label{eq:tau-consideration-sets}
    \bsplit_\Actions(\Actionsb)=\bsplit\left(\left[\aaction^*\right]^{-1}(\Actionsb)\right).
\end{align}
In words, each action subset \Actionsb\ has mass equal to the probability that \bsplit\ induces a belief under which \Actionsb\ is the set of optimal actions. \autoref{proposition:core-bp} characterizes the distribution over posteriors that implement \marginal\ when the \dm's prior is \prior\ and their utility function is \payoff:
\begin{proposition}\label{proposition:core-bp}
Suppose $\pair$ are \bce-consistent. A Bayes plausible distribution over posteriors, $\bsplit\in\Delta(\Posteriors)$, implements $\marginal$ if and only if  the following holds 
    \begin{align}\label{eq:core}
(\forall\Actionsb\subseteq\Actions)\sum_{\aaction\in\Actionsb}\marginal(\aaction)\geq\sum_{\Actionsc\subseteq\Actionsb}\bsplit_\Actions(\Actionsc).
    \end{align}
\end{proposition}
To interpret \autoref{eq:core}, note the following. The left-hand side is the probability under which the agent takes \emph{some} action \aaction\ in the set \Actionsb. Instead, the right-hand side is the probability under which the agent finds \emph{some} action in the set \Actionsb\ optimal, but no action that is not in \Actionsb. \autoref{eq:core} then states the frequency with which the agent takes actions in \Actionsb\ has to be at least the frequency with which an action in \Actionsb\ is optimal.\footnote{\autoref{eq:core} is intimately connected to the Border-Matthews-Maskin-Riley characterization of reduced-form implementation in auctions. The latter states that a reduced-form auction (a collection of interim probabilities of trade for each buyer) has an auction implementation if and only if for all subsets of buyer-type profiles, the probability the reduced form auction allocates the good to types in that set is no larger than the prior probability of buyer types in that set. Letting $\bsplit_\Actions$ and \marginal\ play the role of the reduced-form auction and of the type distribution, respectively, \autoref{eq:core} is morally related to the Border-Matthews-Maskin-Riley inequalities.}

\begin{remark}[A core interpretation]\label{remark:core} \autoref{eq:core} implies \marginal\ is in the \emph{core of the game} induced by the measure $\bsplit_\Actions$.\footnote{\cite{azrieli2022marginal} also note the connection between consistent marginals in the context of stochastic menu choice and cooperative games.} Indeed, given $\bsplit_\Actions$, define the following cooperative game. The set of players is the set of actions \Actions, so that a coalition of players is a subset of actions $\Actionsb\subset\Actions$. The worth of coalition \Actionsb\ is given by $w_{\bsplit_\Actions}(\Actionsb)=\sum_{\Actionsc\subseteq\Actionsb}\bsplit_\Actions(\Actionsc)$. Because $w_{\bsplit_{\Actions}}\geq0$, the core of the game $(\Actions,w_{\bsplit_\Actions})$ is given by
 \[\mathrm{Core}(w_{\bsplit_{\Actions}})=\left\{p\in\Delta\left(\Actions\right):(\forall\Actionsb\subseteq\Actions)\sum_{\aaction\in\Actionsb}p(\aaction)\geq w_{\bsplit_\Actions}(\Actionsb)\right\}.\] 
 \autoref{eq:core} states that \marginal\ is a payment rule for each player in the game that covers the worth of each coalition and hence, belongs to the core of the game.
 \end{remark}
The proof of \autoref{proposition:core-bp} is in \autoref{appendix:ext} and follows from three steps. First, we show how to map our problem into that in \cite{gale1957theorem}. Second, whereas \cite{gale1957theorem} allows for the supply-demand equations to hold as weak inequalities, we show that any solution to Gale's problem clears the market exactly and, hence, satisfies \autoref{eq:demand-supply}. This result follows from \marginal\ and \bsplit\ being probability distributions. \citeauthor{gale1957theorem} refers to such solutions as \emph{maximal} flows. Finally, we show the necessary and sufficient condition for the existence of a feasible and maximal flow in \cite{gale1957theorem} is equivalent to \autoref{eq:core}.

\paragraph{Connection with stochastic choice} We now draw a connection with stochastic choice from menus, which, among other things, allows us to illustrate why \autoref{eq:core} implies \bsplit\ implements \marginal. 

It follows from the results in \cite{gale1957theorem} that \autoref{eq:core} implies a system of conditional probabilities $\{\scr(\cdot|\Actionsb)\in\Delta(\Actions):\Actionsb\subset\Actions\}$ exists such that (i) for all $\Actionsb\subseteq\Actions$, $\scr(\Actionsb|\Actionsb)=1$ and (ii)
\begin{align}\label{eq:stochastic-menu-choice}
    \marginal(\aaction)=\sum_{\Actionsb\subseteq\Actions:\aaction\in\Actionsb}\bsplit_\Actions(\Actionsb)\scr(\aaction|\Actionsb).
\end{align}
The conditional probabilities $\left(\scr(\cdot|\Actionsb)\right)_{\Actionsb\subseteq\Actions}$ can be interpreted as the \dm's stochastic choice from the menus $\{\Actionsb:\Actionsb\subseteq\Actions\}$. \autoref{eq:stochastic-menu-choice} states that the probability that the \dm\ chooses action \aaction\ under marginal is the probability that the agent faces a menu \Actionsb\ that has \aaction\ available and the \dm\ chooses \aaction\ out of \Actionsb.\footnote{ \autoref{eq:stochastic-menu-choice} is \emph{another} demand and supply problem, where $\bsplit_\Actions$ is the supply of action subsets. The results in \cite{gale1957theorem} imply \autoref{eq:core} is equivalent to the existence of $\scr$. The existence of $\scr$ can also be established using \citeauthor{azrieli2022marginal}'s extension of Hall's marriage theorem (see their Proposition 9).} 

Consider now the following ``experiment'': we first draw a menu $\Actionsb$ using $\bsplit_\Actions$ and then draw an action $\aaction\in\Actionsb$ according to $\scr(\cdot|\Actionsb)$. We only inform the \dm\ of the drawn action, and not the menu from which it was drawn. Because we draw menu \Actionsb\ only when it is the optimal set of actions, we only recommend \aaction\ when following the recommendation is optimal for the \dm. As we explain below, \autoref{eq:stochastic-menu-choice} then implies this ``experiment'' induces the \dm\ to take actions with the desired frequency. 

Technically, we have not described an experiment\textemdash a collection of signal distributions conditional on the state of the world\textemdash but one can do so immediately as follows: for each $\type\in\Types$ and $\aaction\in\Actions$,
\begin{align}\label{eq:state-dependent-stochastic-choice}
    \tilde{\scr}(\aaction|\type)=\sum_{\Actionsb:\aaction\in\Actionsb}\sum_{\{\belief\in\mathrm{supp}\,\bsplit:\aaction^*(\belief)=\Actionsb\}}\frac{\belief(\type)}{\prior(\type)}\bsplit(\belief)\scr(\aaction|\Actionsb).
\end{align}
In this experiment, the \dm\ receives an action recommendation conditional on the state of the world, so that \autoref{eq:state-dependent-stochastic-choice} describes the \dm's state-dependent stochastic choice.

The previous discussion connects two sets of conditional distributions over choices that arise in the stochastic choice literature: stochastic choices conditional on a state of the world (\autoref{eq:state-dependent-stochastic-choice}) and stochastic choices out of a menu\textemdash$\scr$ in \autoref{eq:stochastic-menu-choice}. Indeed, the measure $\bsplit_\Actions$ can be interpreted as the frequency with which the agent faces different menus\textemdash action subsets in this case\textemdash whereas the measure \marginal\ represents the frequency with which the agent makes different choices. In other words, the pair $(\bsplit_\Actions,\marginal)$ is analogous to the dataset in \cite{azrieli2022marginal}. Whereas they show the core condition in \autoref{eq:core} characterizes the existence of $\left(\scr(\cdot|\Actionsb)\right)_{\Actionsb\subseteq\Actions}$ in \autoref{eq:stochastic-menu-choice} given $(\bsplit_\Actions,\marginal)$, we show the core condition  characterizes the set of Bayes plausible distribution over posteriors that induce \marginal.
{\singlespacing{
\bibliographystyle{ecta}
\bibliography{id-core}}}

\begin{thebibliography}{51}
\newcommand{\enquote}[1]{``#1''}
\expandafter\ifx\csname natexlab\endcsname\relax\def\natexlab#1{#1}\fi

\bibitem[\protect\citeauthoryear{Aliprantis and Border}{Aliprantis and
  Border}{2013}]{aliprantis2013infinite}
\textsc{Aliprantis, C.~D. and K.~C. Border} (2013): \emph{Infinite Dimensional
  Analysis: A Hitchhiker's Guide}, Springer-Verlag Berlin and Heidelberg GmbH
  \& Company KG.

\bibitem[\protect\citeauthoryear{Arieli and Babichenko}{Arieli and
  Babichenko}{2019}]{arieli2019private}
\textsc{Arieli, I. and Y.~Babichenko} (2019): \enquote{Private bayesian
  persuasion,} \emph{Journal of Economic Theory}, 182, 185--217.

\bibitem[\protect\citeauthoryear{Arieli, Babichenko, Sandomirskiy, and
  Tamuz}{Arieli et~al.}{2021}]{arieli2021feasible}
\textsc{Arieli, I., Y.~Babichenko, F.~Sandomirskiy, and O.~Tamuz} (2021):
  \enquote{Feasible joint posterior beliefs,} \emph{Journal of Political
  Economy}, 129, 2546--2594.

\bibitem[\protect\citeauthoryear{Azrieli and Rehbeck}{Azrieli and
  Rehbeck}{2022}]{azrieli2022marginal}
\textsc{Azrieli, Y. and J.~Rehbeck} (2022): \enquote{Marginal stochastic
  choice,} \emph{arXiv preprint arXiv:2208.08492}.

\bibitem[\protect\citeauthoryear{Barseghyan, Molinari, O'Donoghue, and
  Teitelbaum}{Barseghyan et~al.}{2013}]{barseghyan2013nature}
\textsc{Barseghyan, L., F.~Molinari, T.~O'Donoghue, and J.~C. Teitelbaum}
  (2013): \enquote{The nature of risk preferences: Evidence from insurance
  choices,} \emph{American economic review}, 103, 2499--2529.

\bibitem[\protect\citeauthoryear{Barseghyan, Prince, and Teitelbaum}{Barseghyan
  et~al.}{2011}]{barseghyan2011risk}
\textsc{Barseghyan, L., J.~Prince, and J.~C. Teitelbaum} (2011): \enquote{Are
  risk preferences stable across contexts? Evidence from insurance data,}
  \emph{American Economic Review}, 101, 591--631.

\bibitem[\protect\citeauthoryear{Beresteanu, Molchanov, and
  Molinari}{Beresteanu et~al.}{2011}]{beresteanu2011sharp}
\textsc{Beresteanu, A., I.~Molchanov, and F.~Molinari} (2011): \enquote{Sharp
  identification regions in models with convex moment predictions,}
  \emph{Econometrica}, 79, 1785--1821.

\bibitem[\protect\citeauthoryear{Bergemann, Brooks, and Morris}{Bergemann
  et~al.}{2015}]{bergemann2015limits}
\textsc{Bergemann, D., B.~Brooks, and S.~Morris} (2015): \enquote{The limits of
  price discrimination,} \emph{American Economic Review}, 105, 921--957.

\bibitem[\protect\citeauthoryear{Bergemann, Brooks, and Morris}{Bergemann
  et~al.}{2022}]{bergemann2022counterfactuals}
---\hspace{-.1pt}---\hspace{-.1pt}--- (2022): \enquote{Counterfactuals with
  Latent Information,} \emph{American Economic Review}, 112, 343--368.

\bibitem[\protect\citeauthoryear{Bergemann and Morris}{Bergemann and
  Morris}{2016}]{bergemann2016bayes}
\textsc{Bergemann, D. and S.~Morris} (2016): \enquote{Bayes correlated
  equilibrium and the comparison of information structures in games,}
  \emph{Theoretical Economics}, 11, 487--522.

\bibitem[\protect\citeauthoryear{Border}{Border}{1991}]{border1991functional}
\textsc{Border, K.~C.} (1991): \enquote{Functional analytic tools for expected
  utility theory,} in \emph{Positive Operators, Riesz Spaces, and Economics},
  ed. by C.~D. Aliprantis, K.~C. Border, and W.~A. Luxemburg, Springer, 69--88.

\bibitem[\protect\citeauthoryear{Caplin and Dean}{Caplin and
  Dean}{2015}]{caplin2015revealed}
\textsc{Caplin, A. and M.~Dean} (2015): \enquote{Revealed preference, rational
  inattention, and costly information acquisition,} \emph{American Economic
  Review}, 105, 2183--2203.

\bibitem[\protect\citeauthoryear{Caplin, Dean, and Leahy}{Caplin
  et~al.}{2022}]{caplin2017rationally}
\textsc{Caplin, A., M.~Dean, and J.~Leahy} (2022): \enquote{Rationally
  inattentive behavior: Characterizing and generalizing Shannon entropy,}
  \emph{Journal of Political Economy}, 130, 1676--1715.

\bibitem[\protect\citeauthoryear{Caplin, Martin, and Marx}{Caplin
  et~al.}{2023}]{caplin2023rationalizable}
\textsc{Caplin, A., D.~J. Martin, and P.~Marx} (2023): \enquote{Rationalizable
  Learning,} Tech. rep., National Bureau of Economic Research.

\bibitem[\protect\citeauthoryear{Chambers, Liu, and Rehbeck}{Chambers
  et~al.}{2020}]{chambers2020costly}
\textsc{Chambers, C.~P., C.~Liu, and J.~Rehbeck} (2020): \enquote{Costly
  information acquisition,} \emph{Journal of Economic Theory}, 186, 104979.

\bibitem[\protect\citeauthoryear{Cohen and Einav}{Cohen and
  Einav}{2007}]{cohen2007estimating}
\textsc{Cohen, A. and L.~Einav} (2007): \enquote{Estimating risk preferences
  from deductible choice,} \emph{American economic review}, 97, 745--788.

\bibitem[\protect\citeauthoryear{Das, Dev, and Sarvottamananda}{Das
  et~al.}{2024}]{das2024worst}
\textsc{Das, S., S.~R. Dev, and S.~Sarvottamananda} (2024): \enquote{A
  worst-case optimal algorithm to compute the Minkowski sum of convex
  polytopes,} \emph{Discrete Applied Mathematics}, 350, 44--61.

\bibitem[\protect\citeauthoryear{De~Oliveira and Lamba}{De~Oliveira and
  Lamba}{2022}]{de2022rationalizing}
\textsc{De~Oliveira, H. and R.~Lamba} (2022): \enquote{Rationalizing dynamic
  choices,} \emph{Available at SSRN 3332092}.

\bibitem[\protect\citeauthoryear{Denti}{Denti}{2022}]{denti2022posterior}
\textsc{Denti, T.} (2022): \enquote{Posterior separable cost of information,}
  \emph{American Economic Review}, 112, 3215--3259.

\bibitem[\protect\citeauthoryear{Dewan and Neligh}{Dewan and
  Neligh}{2020}]{dewan2020estimating}
\textsc{Dewan, A. and N.~Neligh} (2020): \enquote{Estimating information cost
  functions in models of rational inattention,} \emph{Journal of Economic
  Theory}, 187, 105011.

\bibitem[\protect\citeauthoryear{Dickstein, Jeon, and Morales}{Dickstein
  et~al.}{2024}]{dickstein2024patient}
\textsc{Dickstein, M.~J., J.~Jeon, and E.~Morales} (2024): \enquote{Patient
  costs and physicians' information,} Tech. rep., National Bureau of Economic
  Research.

\bibitem[\protect\citeauthoryear{Dickstein and Morales}{Dickstein and
  Morales}{2018}]{dickstein2018exporters}
\textsc{Dickstein, M.~J. and E.~Morales} (2018): \enquote{What do exporters
  know?} \emph{The Quarterly Journal of Economics}, 133, 1753--1801.

\bibitem[\protect\citeauthoryear{Dillenberger, Krishna, and
  Sadowski}{Dillenberger et~al.}{2023}]{dillenberger2023subjective}
\textsc{Dillenberger, D., R.~V. Krishna, and P.~Sadowski} (2023):
  \enquote{Subjective information choice processes,} \emph{Theoretical
  Economics}, 18, 529--559.

\bibitem[\protect\citeauthoryear{Dillenberger, Lleras, Sadowski, and
  Takeoka}{Dillenberger et~al.}{2014}]{dillenberger2014theory}
\textsc{Dillenberger, D., J.~S. Lleras, P.~Sadowski, and N.~Takeoka} (2014):
  \enquote{A theory of subjective learning,} \emph{Journal of Economic Theory},
  153, 287--312.

\bibitem[\protect\citeauthoryear{Enke and Graeber}{Enke and
  Graeber}{2023}]{enke2023cognitive}
\textsc{Enke, B. and T.~Graeber} (2023): \enquote{Cognitive uncertainty,}
  \emph{The Quarterly Journal of Economics}, 138, 2021--2067.

\bibitem[\protect\citeauthoryear{Ergin and Sarver}{Ergin and
  Sarver}{2010}]{ergin2010unique}
\textsc{Ergin, H. and T.~Sarver} (2010): \enquote{A unique costly contemplation
  representation,} \emph{Econometrica}, 78, 1285--1339.

\bibitem[\protect\citeauthoryear{Gale}{Gale}{1957}]{gale1957theorem}
\textsc{Gale, D.} (1957): \enquote{A theorem on flows in networks,}
  \emph{Pacific J. Math}, 7, 1073--1082.

\bibitem[\protect\citeauthoryear{Galichon and Henry}{Galichon and
  Henry}{2011}]{galichon2011set}
\textsc{Galichon, A. and M.~Henry} (2011): \enquote{Set identification in
  models with multiple equilibria,} \emph{The Review of Economic Studies}, 78,
  1264--1298.

\bibitem[\protect\citeauthoryear{Gualdani and Sinha}{Gualdani and
  Sinha}{2024}]{gualdani2019identification}
\textsc{Gualdani, C. and S.~Sinha} (2024): \enquote{Identification in discrete
  choice models with imperfect information,} \emph{Journal of Econometrics},
  244, 105854.

\bibitem[\protect\citeauthoryear{Kamenica and Gentzkow}{Kamenica and
  Gentzkow}{2011}]{kamenica2011bayesian}
\textsc{Kamenica, E. and M.~Gentzkow} (2011): \enquote{Bayesian persuasion,}
  \emph{American Economic Review}, 101, 2590--2615.

\bibitem[\protect\citeauthoryear{Khaw, Li, and Woodford}{Khaw
  et~al.}{2021}]{khaw2021cognitive}
\textsc{Khaw, M.~W., Z.~Li, and M.~Woodford} (2021): \enquote{Cognitive
  imprecision and small-stakes risk aversion,} \emph{The Review of Economic
  Studies}, 88, 1979--2013.

\bibitem[\protect\citeauthoryear{Kneeland}{Kneeland}{2015}]{kneeland2015identifying}
\textsc{Kneeland, T.} (2015): \enquote{Identifying higher-order rationality,}
  \emph{Econometrica}, 83, 2065--2079.

\bibitem[\protect\citeauthoryear{Kolotilin}{Kolotilin}{2018}]{kolotilin2018optimal}
\textsc{Kolotilin, A.} (2018): \enquote{Optimal information disclosure: A
  linear programming approach,} \emph{Theoretical Economics}, 13, 607--635.

\bibitem[\protect\citeauthoryear{Kolotilin, Corrao, and Wolitzky}{Kolotilin
  et~al.}{2025}]{kolotilin2023persuasion}
\textsc{Kolotilin, A., R.~Corrao, and A.~Wolitzky} (2025): \enquote{Persuasion
  and matching: Optimal productive transport,} 133.

\bibitem[\protect\citeauthoryear{Kolotilin and Wolitzky}{Kolotilin and
  Wolitzky}{2024}]{kolotilin2024distributions}
\textsc{Kolotilin, A. and A.~Wolitzky} (2024): \enquote{Distributions of
  Posterior Quantiles via Matching,} \emph{Theoretical Economics}, 19,
  1399--1413.

\bibitem[\protect\citeauthoryear{Lu}{Lu}{2016}]{lu2016random}
\textsc{Lu, J.} (2016): \enquote{Random choice and private information,}
  \emph{Econometrica}, 84, 1983--2027.

\bibitem[\protect\citeauthoryear{Magnolfi and Roncoroni}{Magnolfi and
  Roncoroni}{2023}]{magnolfi2019estimation}
\textsc{Magnolfi, L. and C.~Roncoroni} (2023): \enquote{Estimation of discrete
  games with weak assumptions on information,} \emph{The Review of Economic
  Studies}, 90, 2006--2041.

\bibitem[\protect\citeauthoryear{Molchanov and Molinari}{Molchanov and
  Molinari}{2018}]{molchanov2018random}
\textsc{Molchanov, I. and F.~Molinari} (2018): \emph{Random sets in
  econometrics}, vol.~60, Cambridge University Press.

\bibitem[\protect\citeauthoryear{Morris}{Morris}{2020}]{morris2020no}
\textsc{Morris, S.~E.} (2020): \enquote{No trade and feasible joint posterior
  beliefs,} Tech. rep., Massachusetts Institute of Technology.

\bibitem[\protect\citeauthoryear{Myerson}{Myerson}{1982}]{myerson1982optimal}
\textsc{Myerson, R.~B.} (1982): \enquote{Optimal coordination mechanisms in
  generalized principal--agent problems,} \emph{Journal of Mathematical
  Economics}, 10, 67--81.

\bibitem[\protect\citeauthoryear{Rambachan}{Rambachan}{2024}]{rambachan2022identifying}
\textsc{Rambachan, A.} (2024): \enquote{Identifying prediction mistakes in
  observational data,} \emph{The Quarterly Journal of Economics}, 139,
  1665--1711.

\bibitem[\protect\citeauthoryear{Rehbeck}{Rehbeck}{2023}]{rehbeck2023revealed}
\textsc{Rehbeck, J.} (2023): \enquote{Revealed Bayesian expected utility with
  limited data,} \emph{Journal of Economic Behavior \& Organization}, 207,
  81--95.

\bibitem[\protect\citeauthoryear{Rockafellar}{Rockafellar}{1970}]{rockafellar1970convex}
\textsc{Rockafellar, R.~T.} (1970): \enquote{Convex Analysis,} \emph{Convex
  Analysis}, 28.

\bibitem[\protect\citeauthoryear{Serfozo}{Serfozo}{1982}]{serfozo1982convergence}
\textsc{Serfozo, R.} (1982): \enquote{Convergence of Lebesgue integrals with
  varying measures,} \emph{Sankhy{\=a}: The Indian Journal of Statistics,
  Series A}, 380--402.

\bibitem[\protect\citeauthoryear{Strack and Yang}{Strack and
  Yang}{2024}]{strack2024privacy}
\textsc{Strack, P. and K.~H. Yang} (2024): \enquote{Privacy Preserving
  Signals,} \emph{Available at SSRN 4467608}.

\bibitem[\protect\citeauthoryear{Strassen}{Strassen}{1965}]{strassen1965existence}
\textsc{Strassen, V.} (1965): \enquote{The existence of probability measures
  with given marginals,} \emph{The Annals of Mathematical Statistics}, 36,
  423--439.

\bibitem[\protect\citeauthoryear{Syrgkanis, Tamer, and Ziani}{Syrgkanis
  et~al.}{2017}]{syrgkanis2017inference}
\textsc{Syrgkanis, V., E.~Tamer, and J.~Ziani} (2017): \enquote{Inference on
  auctions with weak assumptions on information,} \emph{arXiv preprint
  arXiv:1710.03830}.

\bibitem[\protect\citeauthoryear{Vohra, Toikka, and Vohra}{Vohra
  et~al.}{2023}]{toikka2022bayesian}
\textsc{Vohra, A., J.~Toikka, and R.~Vohra} (2023): \enquote{Bayesian
  persuasion: Reduced form approach,} \emph{Journal of Mathematical Economics},
  102863.

\bibitem[\protect\citeauthoryear{Weibel}{Weibel}{2007}]{weibel2007minkowski}
\textsc{Weibel, C.} (2007): \enquote{Minkowski sums of polytopes: combinatorics
  and computation,} Tech. rep., EPFL.

\bibitem[\protect\citeauthoryear{Yang and Zentefis}{Yang and
  Zentefis}{2024}]{yang2024monotone}
\textsc{Yang, K.~H. and A.~K. Zentefis} (2024): \enquote{Monotone function
  intervals: Theory and applications,} \emph{American Economic Review}, 114,
  2239--2270.

\bibitem[\protect\citeauthoryear{Ziegler}{Ziegler}{2012}]{ziegler2012lectures}
\textsc{Ziegler, G.~M.} (2012): \emph{Lectures on polytopes}, vol. 152,
  Springer Science \& Business Media.

\end{thebibliography}

\appendix
\section{Omitted proofs}\label{appendix:proofs}
\subsection{Omitted proofs from \autoref{sec:main}}\label{appendix:main}

We begin this section with the proof of \autoref{theorem:h-representation}. When \ppmpair\ has dimension $\tnum-1$ and we regard it as a subset of $\reals^{\tnum-1}$, the proof of \autoref{theorem:h-representation} follows from the arguments in the main text. We focus here on the case in which \ppmpair\ has dimension $\dimm\leq\tnum-1$ and present the characterization result when we regard $\ppmpair$ as a subset of $\reals^\tnum$ in \autoref{theorem:h-representation-app}. After presenting the proof of \autoref{theorem:h-representation-app}, we show how to modify the statement of \autoref{theorem:h-representation} in the main text to regard \ppmpair\ as a subset of $\reals^\tnum$.

Regarding $\ppmpair\subset \Reals^\tnum$, let $\dimm$ denote the dimension of \ppmpair. Let $\mathrm{M}^\perp$ denote the kernel space of $\ppmpair$, defined as follows
\begin{align*}
\mathrm{M}^\perp:=\{\pricev\in  \Reals^\tnum:\pricev\cdot (\beliefv-\beliefvb)=0,\forall \beliefv,\beliefvb \in \ppmpair\}.
\end{align*}
Note that $\mathrm{M}^\perp$ is a subspace of $\Reals^\tnum$. As a result, a basis of $\mathrm{M}^\perp$ exists. Let $\Basis=\{\basisv_1,...,\basisv_{\tnum-\dimm}\}$ be a basis of $\mathrm{M}^\perp$. Let $P_\Basis$ be the projection matrix onto $\mathrm{M}^\perp$ defined by \Basis.\footnote{Formally, $P_\Basis=\Basis(\Basis^T\Basis)^{-1}\Basis^T$ where we abuse notation by treating \Basis\ as a matrix whose columns are $\basisv_i$ and $T$ denotes the transpose.} For a set $\gsubset\subset \Reals^\tnum$, we let $\gsubset/\Basis$ denote the projection of $\gsubset$ on the subspace defined by $\ppmpair$:
\begin{align*}
\gsubset/\Basis:=\{\gbold-P_\Basis\gbold:\gbold\in \gsubset\}.\end{align*}

\begin{theorem}\label{theorem:h-representation-app}
Suppose \ppmpair\ is nonempty.  The pair $\pair\in\Posteriors\times\Delta(\Actions)$ is \bce-consistent given \payoff\ if and only if \autoref{eq:support} holds for all $\pricev\in \Prices_{\ext}\left(\big(\wedge_{\aaction\in\Actions_{+}} \normal(\opta)\big)/\Basis\right)\cup \Basis\cup (-\Basis)$.
\end{theorem}
As we explain after the proof of \autoref{theorem:h-representation-app}, when $\dimm=\tnum-1$, we can take \Basis\ to be the constant vector, $\mathbf{1}$. 
\begin{proof}[Proof of \autoref{theorem:h-representation-app}] 
Necessity follows directly from \autoref{observation:ms}. Next, we prove sufficiency. Let $\N:=\wedge_{\aaction\in\Actions_+} \normal(\opta)$. By \citet[Theorem 7.12]{ziegler2012lectures}, \N\ is the normal fan of \ppmpair. Let \face\ denote a facet of \ppmpair, that is, a $\dimm-1$-dimensional face of \ppmpair. We say a halfspace $H_\face:=\{\beliefv\in \reals^\tnum: \pricev_\face\cdot \beliefv\leq c_\face\}$ is a \emph{facet-defining halfspace of $F$},  if $\ppmpair\subset H_\face$ and $\{\beliefv\in\reals^\tnum: \pricev_\face\cdot \beliefv= c_\face\}\cap \ppmpair=\face$. Our proof relies on the following lemma.

        \begin{lemma}[\protect{\citet[Theorem 2.15, (7)]{ziegler2012lectures}}]
            Let $\polytope\subset \Reals^\tnum$ be a polytope. Then, \polytope\ is the intersection of its facet-defining halfspaces and its affine hull, $\aff(\polytope)$.\footnote{The affine hull of a polytope \polytope\ is the set of all affine combinations of elements of \polytope.}
        \end{lemma}
        In what follows, we construct a facet-defining halfspace for each facet \face\ of $\ppmpair$, and show that if \prior\ satisfies the condition in \autoref{theorem:h-representation-app}, $\prior\in (\cap_{\face\in \text{Facets}(\ppmpair)}H_\face)\cap \aff(\ppmpair)$. Invoking the lemma completes the proof.

For any facet \face\ of $\ppmpair$, let $\normal_\face \in \N$ be the normal cone at \face. Note the dimension of $\normal_\face$ is $\tnum-(\dimm-1)$ and $\Basis\subset \normal_\face$. Thus, $\normal_\face/\Basis$ is 1-dimensional. It follows that $\normal_\face/\Basis \in \Prices_{\ext}(\N/\Basis)$. Let $\pricev_\face\in \normal_\face/\Basis$ be its extreme ray (i.e., $\pricev_\face\neq 0$), and $\beliefv_\face$ be any point on \face. Next, we argue that $H_\face:=\{\beliefv\in \reals^\tnum: \pricev_\face \beliefv\le \pricev_\face\beliefv_\face\}$ is a facet-defining halfspace for $F$.

        First, because $\pricev_\face\in \normal_\face/\Basis$, we know that $\pricev_\face = n_\face-P_\Basis n_\face$ for some $n_\face\in \normal_\face$. Because $n_\face$ is in the normal cone of \face\ and $\beliefv_\face$ is a point in \face, we have $n_\face\cdot \beliefv\le n_\face\cdot \beliefv_\face$ for all $\belief\in \ppmpair$. Moreover, because \Basis\ is a basis for $\mathrm{M}^\perp$, we have $(P_\Basis n_\face)\beliefv = (P_\Basis n_\face)\beliefv_\face$ for all $\beliefv\in \ppmpair$. Therefore, $\pricev_\face\beliefv\le \pricev_\face\beliefv_\face$ for all $\beliefv\in \ppmpair$. Thus, $\ppmpair\subset H_\face$. Second, $\face\subset \{\beliefv\in \reals^\tnum: \pricev_\face \beliefv= \pricev_\face \beliefv_\face\}\cap \ppmpair$ because $\face\subset \ppmpair$ and $\pricev_\face\beliefv = \pricev_\face\beliefv_\face$ for all $\beliefv\in F$. Third, suppose $\face\subsetneq \{\beliefv\in \reals^\tnum: \pricev_\face\beliefv= \pricev_\face\beliefv_\face\}\cap \ppmpair$. Because $\ppmpair\subset H_\face$, we know $\{\beliefv\in \reals^\tnum: \pricev_\face\beliefv= \pricev_\face\beliefv_\face\}\cap \ppmpair$ is a face of $\ppmpair$, and it contains facet \face\ as a strictly subset. This means that $\{\beliefv\in \reals^\tnum: \pricev_\face\beliefv= \pricev_\face\beliefv_\face\}\cap \ppmpair = \ppmpair$ and thus $\pricev_\face\in \mathrm{M}^\perp$. However, this can only be the case when $\pricev_\face=0$; a contradiction.
        
        Summarizing the three observations, we conclude that $H_\face:=\{\beliefv\in \reals^\tnum: \pricev_\face\beliefv\le \pricev_\face\beliefv_\face\}$ is a facet-defining halfspace of \face. The condition in \autoref{theorem:h-representation-app} implies that \autoref{eq:support} holds for all such $\pricev_\face$, so we conclude that $\prior \in \cap_{\face\in \text{Facets}(\ppmpair)} H_\face$.

        Finally, we argue that $\prior\in \aff(\ppmpair)$. Note that $\gbold \in \aff(\ppmpair)$ if and only if there exists $\beliefv^*\in \ppmpair$ such that $\basisv (\gbold-\beliefv^*) = 0$ for all $\basisv\in \Basis$. Take any $\beliefv^*\in \ppmpair$. Because \autoref{eq:support} holds for all $\pricev\in \Basis\cup (-\Basis)$ and $\basisv \beliefv=\basisv \beliefv^*$ for all $\beliefv\in \ppmpair$, we conclude that $\basisv \priorv=\basisv \beliefv^*$ for all $\basisv\in\Basis$ and thus $\prior \in \aff(\ppmpair)$.
\end{proof}

When \ppmpair\ is $\tnum-1$ dimensional, we can take $\Basis$ to be any constant vector.  Because checking \autoref{eq:support} holds for this vector (and its negative) is equivalent to requiring the coordinates of \prior\ to add up to $1$, we can simplify the condition to only checking one-dimensional normal cones after taking the quotient with $B$. Because we are now regarding \ppmpair\ as a subset of $\reals^\tnum$, we denote by $\tilde{\mathrm{M}}(\payoff,\marginal)$ the full-dimensional polytope in $\reals^{\tnum-1}$. Similarly,  denote by $\tilde{\mu}_0$ the $(\tnum-1)$-dimensional vector $(\prior(\type_2),\dots,\prior(\type_I))$. Next, we argue that if \autoref{eq:support} holds for all the extreme rays of the one-dimensional normal cones of $\tilde{\mathrm{M}}(\payoff,\marginal)$ at $\tilde{\mu}_0$, then it also holds for $\Prices_{\ext}(\N/\Basis)$ at $\prior$.

Let $\tilde\face$ denote a facet of $\tilde{\mathrm{M}}(\payoff,\marginal)$. Let $\tilde \pricev_{\tilde \face}$ be an extreme ray of the one-dimensional normal cone corresponding to $\tilde \face$. Let $\beliefvb_{\tilde\face}$ be a point in $\face$. Because \autoref{eq:support} holds for $\tilde \pricev_{\tilde \face}$ at $\tilde{\mu}_0$, we have
\begin{equation}\label{eq:normal I-1}
    \tilde \pricev_{\tilde \face} \tilde{\beliefv}_0\le \tilde\pricev_{\tilde\face}\beliefv_{\tilde\face}.
\end{equation}
Define $\face\subset\reals^\tnum$ as follows:
\begin{align*}
\face:=\{(1-\sum_{i=2}^I \tilde \mu(\type_i),\tilde\mu(\type_2),...,\tilde\mu(\type_I)):\beliefvb\in \tilde \face\}.
\end{align*}
Note that \face\ is a facet of $\ppmpair$. Let $\pricev_\face:=(0,\tilde p_{\tilde F}(\type_2),...,\tilde p_{\tilde F}(\type_I))$. Because $\tilde\pricev_{\tilde\face}$ is a normal vector at $\tilde\face$, $\pricev_\face$ is a normal vector at $\face$. Let $\beliefv_\face:=(1-\sum_{i=2}^I \tilde \mu_{\tilde F}(\type_i),\tilde \mu_{\tilde F}(\type_2),...,\tilde \mu_{\tilde F}(\type_I))$. Then we know $\mu_\face\in F$ and \autoref{eq:normal I-1} implies that $\pricev_\face\cdot \priorv\le \pricev_\face\cdot \beliefv_\face$. This means \autoref{eq:support} is satisfied for $\pricev_\face$ and $\prior$.

Because $\priorv,\beliefv_\face\in \Delta\Omega$, we know that \autoref{eq:support} is also satisfied for vector $\pricev_\face^*:=\pricev_\face-\sum_{\tindex=2}^\tnum \pricev_\face(\type_\tindex)$. We observe that: (1) $\pricev_\face^*$ is also a normal vector at \face; (2) $\pricev_\face^*$ is not a constant vector because $\tilde \pricev_{\tilde F}\neq 0$; (3) the inner product between $\pricev_\face^*$ and any constant vector is $0$. Therefore, we conclude that $\pricev_\face^*\in \normal_\face/\Basis$. Because this argument holds for any facet $F$ of $\ppmpair$, we conclude that \autoref{eq:support} holds for all $\pricev\in \Prices_{\ext}(\N/\Basis)$. This completes the argument.

\begin{proof}[Proof of \autoref{proposition:2-3}]
When $d=0$, $\mathrm{M}^\perp=\reals^\tnum$ and we can take $\Basis=\{-e_\type:\type\in\Types\}$ in \autoref{theorem:h-representation-app}. By \autoref{theorem:h-representation-app}, that \autoref{eq:support} holds for all $\pricev\in\Basis\cup(-\Basis)$ implies $\prior\in\ppmpair$. In fact, because $\prior\in\Posteriors$, that is, the entries of $\prior$ sum up to 1, the fact that \autoref{eq:support} holds for all $\pricev\in\Basis$, i.e., that \autoref{eq:bce-c-states} holds, fully pins down \ppmpair.

When $d=1$, \ppmpair\ must be obtained from \marginal-convex combinations of some \opta's of dimension 1 that are parallel to each other (i.e., share the same kernel) and other \opta's of dimension 0 (i.e., singletons), for $\aaction\in\Actions_+$. Hence \ppmpair\ is also parallel to those 1-dimensional \opta's and share the same normal vectors that are contained in $\Prices_{\opt}$. 

When $d=2$ and thus $\tnum=3$, the facets of \ppmpair\ are then 1-dimensional and are obtained from \marginal-convex combinations of faces of the sets \opta. In particular, some of them are 1-dimensional faces from \opta's for some $\aaction\in\Actions_1\subset\Actions_+$ and parallel to each other, while the others are extreme points of the remaining sets \optab's for $\aactionb\in\Actions_+\setminus\Actions_1$. Because \opta's are at most 2-dimensional, those 1-dimensional faces are either \opta's themselves or their facets. It then follows that the normals of the facets of \ppmpair\ are exactly the normals of some facet of some \opta, which are contained in $\Prices_{\opt}$. This completes the proof.
\end{proof}
\subsection{Omitted proofs from \autoref{sec:utility}}\label{appendix:utility}
The proofs in \autoref{sec:utility} build on two companion programs to the dual program \ref{eq:dual}: 
\begin{align}
    V_D^\uparrow\pair&=\min_{\pricev\in\realst,\priceav\in\realsa,\multvp\in\realsa_{\geq0}}\priceav\marginalv-\pricev\priorv \tag{D$^\uparrow$}\label{eq:dualup}\\
    &\text{s.t. }(\forall\tindex)(\forall\aindex)\pricea(\aaction_\aindex)\geq\price(\type_\tindex)+\multp_\aindex\diff(\aaction_\aindex,\aaction_{\aindex+1},\type_\tindex),\nonumber\\
        V_D^\downarrow\pair&=\min_{\pricev\in\realst,\priceav\in\realsa,\multvm\in\realsa_{\geq0}}\priceav\marginalv-\pricev\priorv\tag{D$^\downarrow$}\label{eq:dualdown}\\
        &\text{s.t. }(\forall\tindex)(\forall\aindex)\pricea(\aaction_\aindex)\geq\price(\type_\tindex)+\multm_\aindex \diff(\aaction_{j},\aaction_{j-1},\type_\tindex).\nonumber
\end{align}
The programs \ref{eq:dualup} and \ref{eq:dualdown} are analogous to the dual program \ref{eq:dual}, but restricting attention to multipliers in which either the adjacent downward constraints or the adjacent upward constraints are not binding. Recall we follow the convention that $\multm_1=\multp_{\carda}=0$. It is immediate that $V_D\pair\leq\min\{V_D^\uparrow\pair,V_D^\downarrow\pair\}$. \autoref{proposition:mcv} shows that, in fact, $V_D\pair=\min\{V_D^\uparrow\pair,V_D^\downarrow\pair\}$.

In what follows, when we refer to candidate solutions to programs \ref{eq:dual}, \ref{eq:dualup}, and \ref{eq:dualdown}, we assume without loss of generality that \autoref{eq:basic-p-identity} holds, which we repeat here for ease of reference:
\begin{align}\label{eq:p-binds}
p(\type)=\min_{\aaction_\aindex\in\Actions} \left[q(\aaction_\aindex)+\multp_j\diff(\aaction_{\aindex+1},\aaction_{j},\type)-\multm_j\diff(\aaction_j,\aaction_{j-1},\type)\right].
\end{align}
Depending on the program under consideration, $\multvm=0$ \eqref{eq:dualup} or $\multvp=0$ \eqref{eq:dualdown} in \autoref{eq:p-binds}.

\begin{proof}[Proof of \autoref{proposition:mcv}]
Toward a contradiction, assume $\min\{V_D^\uparrow\pair,V_D^\downarrow\pair\}\geq0$, but $V_D\pair<0$. Then, a tuple $(\tilde\pricev,\tilde\priceav,\tilde{\multv}^\uparrow,\tilde{\multv}^\downarrow)\in\realst\times\realsa\times\realsa_{\geq0}\times\realsa_{\geq0}$ exists such that the following holds: the tuple is feasible for program \ref{eq:dual}, $\tilde{\multv}^\uparrow,\tilde{\multv}^\downarrow\neq0$, and $\tilde\pricev\priorv>\tilde\priceav\marginalv$.

 Let $\Actions_\uparrow=\{j\in\{1,\dots,\carda\}:\tilde{\mult}^\uparrow_\aindex>0\}$, $A_\downarrow=\{j\in\{1,\dots,\carda\}:\tilde{\mult}_\aindex^\downarrow>0\}$, and let $\Actions_0=\{1,\dots,\carda\}\setminus\left(\Actions_\uparrow\cup\Actions_\downarrow\right)$. Recalling that $\tilde\pricev$ satisfies \autoref{eq:p-binds}, we can write
 \begin{align*}
\tilde\price(\type)=\min\left\{\min_{j\in\Actions_\uparrow}\tilde\pricea(\aaction_\aindex)+\tilde{\mult}^\uparrow_\aindex \diff(a_{\aindex+1},a_\aindex,\type),\min_{j\in\Actions_\downarrow}\tilde\pricea(\aaction_\aindex)-\tilde{\mult}_\aindex^\downarrow \diff(a_{j},a_{j-1},\type),\min_{j\in\Actions_0}\tilde\pricea(\aaction_\aindex)\right\}.
 \end{align*}
Because $\tilde{\multv}^\uparrow,\tilde{\multv}^\downarrow\neq0$, $\tilde{\pricev}$ is single-peaked. Let $\tilde{\price}^*=\max\{\tilde{\price}(\type):\type\in\Types\}$ and let $i^*$ denote the  index of the smallest maximizer. Without loss of generality, we can assume $\tilde\pricea(\aaction_\aindex)=\tilde{\price}^*$ for $\aindex\in\Actions_0$: feasibility implies $\tilde\pricea(\aaction_j)\geq\tilde{\price}^*$ for $\aindex\in\Actions_0$ and satisfying the constraint with equality only lowers the value of the objective.

Define $(\pricevup,\priceavup)$ and $(\pricevdown,\priceavdown)$ as follows:
\begin{align*}
    \priceup(\type_\tindex)&=\mathbbm{1}[i\leq i^*]\tilde{\price}(\type_\tindex)+\mathbbm{1}[i\geq i^*+1]\tilde{\price}^*&&    \pricedown(\type_\tindex)=\mathbbm{1}[i\geq i^*+1]\tilde{\price}(\type_\tindex)+\mathbbm{1}[i\leq i^*]\tilde{\price}^*\\
    \priceaup(\aaction_\aindex)&=\mathbbm{1}[j\in\Actions\setminus\Actions_\downarrow]\tilde\pricea(\aaction_\aindex)+\mathbbm{1}[j\in\Actions_{\downarrow}]\tilde{\price}^*&&
    \priceadown(\aaction_\aindex)=\mathbbm{1}[j\in\Actions\setminus\Actions_\uparrow]\tilde\pricea(\aaction_\aindex)+\mathbbm{1}[j\in\Actions_{\uparrow}]\tilde{\price}^*.
\end{align*}
In words, \pricevup\ truncates $\tilde{\pricev}$ to the right of $i^*$, making it equal to $\tilde{\price}^*$, whereas $\priceavup$ corrects $\tilde\priceav$ on those actions $j\in\Actions_\downarrow$ so that the dual constraints are satisfied when making the multiplier \multm\ equal to zero for those actions. Instead,  $\pricevdown$ truncates $\tilde{\pricev}$ to the left of $i^*$, making it equal to $\tilde{\price}^*$, whereas $\priceavdown$ corrects $\tilde\priceav$ on those actions $j\in\Actions_\uparrow$ so that the dual constraints are satisfied when making the multiplier \multp\ equal to zero for those actions. 

Note  $(\pricevup,\priceavup,\tilde{\multv}^\uparrow)$ is feasible for the first program, and $(\pricevdown,\priceavdown,\tilde{\multv}^\downarrow)$ is feasible for the second. Moreover, $\pricevup+\pricevdown=\tilde{\pricev}+\tilde{\price}^*$ and $\priceavdown+\priceavup=\tilde{\priceav}+\tilde{\price}^*$. Thus, either $\pricevup\priorv>\priceavup\marginalv$ or $\pricevdown\priorv>\priceavdown\marginalv$, contradicting that $\min\{V_D^\uparrow\pair,V_D^\downarrow\pair\}\geq0$.
\end{proof}
\begin{proof}[Proof of \autoref{proposition:aud}]
That the conditions are necessary follows from the feasibility of $(\pricevup_{\candt},\priceavup_{\candt})$ in program \ref{eq:dualup}. We argue this in the first step below. The argument for $(\pricevdown_{\candt},\priceavdown_{\candt})$ is similar. 

The sufficiency of the conditions follows from showing that
\begin{align}\label{cond:aud}
(\forall\candt\in\Types)\priceavup_{\candt}\marginalv\geq\pricevup_{\candt}\priorv, \end{align}  implies $V_D^\uparrow\pair\geq0$. The proof that $\priceavdown_{\candt}\marginalv\geq\pricevdown_{\candt}\priorv$ for all $\candti\in\Types$  implies $V_D^\downarrow\pair\geq0$ is analogous, and we omit it for brevity. \autoref{proposition:mcv} implies this is enough to show $V_D\pair\geq0$.

The proof proceeds in six steps. First, we show that for all $\candt\in\Types$, $\priceavup_{\candt}$ is the solution to program \ref{eq:dualup} given $\pricev=\pricevup_{\candt}$. This result follows by observing that
\begin{align*}
    \pricea_{\candt}^\uparrow(\aaction)=\min_{\mult\in\reals_{\geq0}^{\carda}}\max_{\type\in\Types}\price_{\candt}^\uparrow(\type)-\multp(\aaction)[\diff(\type)+\ctt(\aaction)/\slope(\aaction)].
\end{align*}
Second, note that if a triple $(\pricev,\priceav,\multvp)$ is feasible for \ref{eq:dualup} and is such that $\pricev=\pricevup_{\candt}$ for some $\candti\in\Types$, the first step and \autoref{cond:aud} implies that $\priceav\marginalv\geq\pricev\priorv$.

Third, note that for constant \pricev, the optimal value of \priceav\ is $\max_\type\price(\type)$ (and  $\multvp\equiv0$), which implies $\priceav\marginalv \ge \pricev\priorv$.\footnote{Recall we are assuming that no action in $\Actions_+$ is strictly dominated so that for all actions in $\Actions_+$ $\gamma(\aaction)\diff(\type_1)+\slope(\aaction)\leq0$.} 
%
%

Fourth, we note that if a triple $(\pricev,\priceav,\multvp)$, feasible for \ref{eq:dualup}, can be written as $(\pricev,\priceav,\multvp)=\sum_{k=1}^K \alpha_k (\pricev_k,\priceav_k,\multv_k)$, with (i) $\alpha_k\geq0$, (ii) $(\pricev_k,\priceav_k,\multv_k)$ feasible for \ref{eq:dualup}, and (iii) $\pricev_k = \pricevup_{\candti}$ for some $\candti$ or $\pricev_k$ constant, then the second and third steps implies that $\priceav\marginalv\geq\pricev\priorv$. 

Fifth, note that we only need to show $\priceav\marginalv\geq\pricev\priorv$ for triples $(\pricev,\priceav,\multvp)$ that are feasible for \ref{eq:dualup}, where $\pricev$ takes the form of \autoref{eq:p-binds}. Note that under affine utility differences, up to a renormalization, \autoref{eq:p-binds} becomes:
\begin{align}\label{eq:pricev-aud}
\price(\type)=\min_{\aaction\in\Actions}\multvp(\aaction)\diff(\type)+\pricea(\aaction)+\multvp(\aaction)\ctt(\aaction)/\slope(\aaction).
\end{align}

Finally, consider any triple $(\pricev,\priceav,\multvp)$ that is feasible for \ref{eq:dualup} and such that \pricev\ is as in \autoref{eq:pricev-aud}. We argue it can always be written as in the fourth step. This completes the proof. 

We show this by induction on the number of actions where the value of $\multvp$ is strictly positive, which we denote by $n(\multvp):=|\{a:\multp(a)>0\}|$. If $n(\multvp)=0$, then \autoref{eq:pricev-aud} becomes $p(\omega) = \min_{a\in A} q(a)$, which means $\pricev$ is constant. Thus, \pricev\ is as in the fourth step.

%

Now suppose $(\pricev,\priceav,\multvp)$ can be written as in the fourth step whenever $n(\multvp)\le N-1$ for some $N\ge 1$, we show it is also true when $n(\multvp)=N$. Toward this, denote the maximum of $\pricev$ by $\price^\star$ and the smallest state index at which this maximum is attained by $\candi$. If \pricev\ is constant, then it is already in the form of the fourth step. Hence, assume it is not constant. Because \pricev\ is increasing, we have that $\candi>1$. 

Notice that we can write $\pricev$ as 
\begin{align}\label{eq:pricev-aud-2}
\price(\type)=\min\left\{\min_{\aaction:\multp(\aaction)>0}\{\multp(\aaction)\diff(\type)+\pricea(\aaction)+\multp(\aaction)\ctt(\aaction)/\slope(\aaction)\},\min_{\aaction:\multp(\aaction)=0}\pricea(\aaction)\right\},
\end{align}
where recall that $\multp(\aaction_\anum)=0$ and hence the set $\{\aaction:\multp(\aaction)=0\}$ is nonempty.
Below, we denote by $\price_+(\type)$ the first argument of the minimum on the right-hand side of \autoref{eq:pricev-aud-2}. Moreover, as in the proof of \autoref{proposition:mcv}, it is without loss of generality to assume that $(\pricev,\priceav,\multvp)$ is such that  $\pricea(\aaction)=\price^\star$ whenever $\multp(\aaction)=0$. Thus, the second argument of the minimum in \autoref{eq:pricev-aud-2} equals $\price^\star$. 

We now argue that we can without loss focus on $(\pricev,\priceav,\multvp)$ such that 
\begin{align}\label{eq:max-exhaust}
\price^\star=\min_{\aaction:\multp(\aaction)>0}\{\multp(\aaction)\diff(\candti)+\pricea(\aaction)+\multp(\aaction)\ctt(\aaction)/\slope(\aaction)\}=\price_+(\type_{\candi}).
\end{align}
In words, the two arguments of the minimum in \autoref{eq:pricev-aud-2} coincide at $\type=\candti$.

Suppose that $\price^\star<\price_+(\type_{\candi})$. Letting $c_1=\price_+(\type_{\candi-1})$ and $c_2=\price_+(\candti)$, we have $c_1<\price^\star<c_2$. For $k\in\{1,2\}$, define
\begin{align*}
    \pricea_k(\aaction)&=\pricea(\aaction)\mathbbm{1}\left[\multp(\aaction)>0\right]+c_k\mathbbm{1}\left[\multp(\aaction)=0\right],\\
    \mult_k(\aaction)&=\multp(\aaction).
    \end{align*}
  Furthermore, define
  \begin{align*}
    \price_1(\type_\tindex)&=\min_{\aaction\in\Actions}\mult_1(\aaction)\diff(\type_\tindex)+\pricea_1(\aaction)+\mult_1(\aaction)\ctt(\aaction)/\slope(\aaction)=\mathbbm{1}[\tindex<\candi-1]\price(\type_\tindex)+\mathbbm{1}[\tindex\geq\candi-1]c_1,\\
        \price_2(\type_\tindex)&=\min_{\aaction\in\Actions}\mult_2(\aaction)\diff(\type_\tindex)+\pricea_2(\aaction)+\mult_2(\aaction)\ctt(\aaction)/\slope(\aaction)=\mathbbm{1}[\tindex<\candi]\price(\type_\tindex)+\mathbbm{1}[\tindex\geq\candi]c_2.
\end{align*}
Letting $\alpha=(\price^\star-c_1)/(c_2-c_1)$, it is straightforward to check that $(\pricev,\priceav,\multvp)=(1-\alpha)\tripleo+\alpha\triplet$. Note that $n(\multvp)=n(\multv_1)=n(\multv_2)$. Moreover, $(\pricev_k,\priceav_k,\multvp_k)$ is such that $\max_{\type\in\Types}\price_k(\type)$ satisfies \autoref{eq:max-exhaust} for $k=1,2$. We conclude that arguing that the triple $(\pricev,\priceav,\multvp)$ can be written as in the fourth step for triples that satisfy \autoref{eq:max-exhaust} suffices to prove the result.

Take any such $(\pricev,\priceav,\multvp)$, define \tripleo\ as follows: 
    \begin{align*}
    \pricea_1(\aaction)&=-\frac{\ctt(\aaction)}{\slope(\aaction)}\mathbbm{1}\left[\multp(\aaction)>0\right]+\diff(\candti)\mathbbm{1}\left[\multp(\aaction)=0\right],\\
    \mult_1(\aaction)&=\mathbbm{1}\left[\multp(\aaction)>0\right],\\
    \price_1(\type)&=\min_{\aaction\in\Actions}\mult_1(\aaction)\diff(\type)+\pricea_1(\aaction)+\mult_1(\aaction)\ctt(\aaction)/\slope(\aaction).
    \end{align*}
By construction, \tripleo\ is feasible for \ref{eq:dualup}, and, for all $\type\in\Types$, 
\begin{align*}
    \price_1(\type)=\min_{\aaction\in\Actions}\mathbbm{1}\left[\multp(\aaction)>0\right]\diff(\type)+\diff(\candti)\mathbbm{1}\left[\multp(\aaction)=0\right]=\min\{\diff(\type),\diff(\candti)\}=\pricevup_{\candti}(\type).
\end{align*}
    Let $\underline{\mult}:=\min_{\aaction:\multp(\aaction)>0}\multp(\aaction)$ and define \triplet\ as follows:
    \begin{align*}
\pricea_2(\aaction)&=\pricea(\aaction)-\underline{\mult}\pricea_1(\aaction), \\
    \mult_2(\aaction)&=\multp(\aaction)-\underline{\mult}\mult_1(\aaction),\\
    \price_2(\type)&=\min_{\aaction\in\Actions}\mult_2(\aaction)\diff(\type)+\pricea_2(\aaction)+\mult_2(\aaction)\ctt(\aaction)/\slope(\aaction).
    \end{align*} Note that
    \begin{align*}
        \price_2(\type)=&\min\left\{\min_{\aaction:\multp(\aaction)>0}\{[\multp(\aaction)-\underline{\mult}]\diff(\type)+\pricea(\aaction)-\underline{\mult}\pricea_1(\aaction)+[\multp(\aaction)-\underline{\mult}]\ctt(\aaction)/\slope(\aaction)\},
        \min_{\aaction:\multp(\aaction)=0}\{\pricea(\aaction)-\underline{\mult}\pricea_1(\aaction)\}\right\}\\
        =&\min\left\{\price_+(\type)-\underline{\mult}\diff(\type),\price^\star-\underline{\mult}\diff(\candti)\right\}.
    \end{align*}
Because $\price^\star$ satisfies \autoref{eq:max-exhaust}, we have that
\begin{align*}
&\left(\price^\star-\underline{\mult}\diff(\candti)\right)-\left(\price_+(\type)-\underline{\mult}\diff(\type)\right)=
\\
&=\left(\min_{\aaction:\multp(\aaction)>0}(\multp(\aaction)-\underline{\mult})\diff(\candti)+\frac{\ctt(\aaction)}{\slope(\aaction)}\multp(\aaction)+\pricea(\aaction)\right)-\left(\min_{\aaction:\multp(\aaction)>0}(\multp(\aaction)-\underline{\mult})\diff(\type)+\frac{\ctt(\aaction)}{\slope(\aaction)}\multp(\aaction)+\pricea(\aaction)\right)
\end{align*}
It follows that $\price_2(\type)=\price(\type)-\underline{\mult}\price_1(\type)$. Together with the definition of $\priceav_2,\multv_2$, we conclude that
\begin{align}\label{eq:ih}(\pricev,\priceav,\multvp)=\underline{\mult}\tripleo+\triplet.\end{align}
Moreover, note that $n(\multv_2)\le n(\multvp)-1=N-1$. By the inductive hypothesis, we conclude that \triplet\ can be written as in the fourth step. Therefore, \autoref{eq:ih} is exactly a decomposition of $(\pricev,\priceav,\multvp)$ in the form of the fourth step.

The six steps together deliver that \autoref{cond:aud} implies that $V_D^\uparrow\pair\geq0$.
%
\end{proof}

\begin{proof}[Proof of \autoref{proposition:2-step}]
The proof is similar to that of \autoref{proposition:aud}. That the conditions are necessary follows from the feasibility of $(\pricevup_\aindex,\priceavup_\aindex)$ in program \ref{eq:dualup}. We argue this in the first step below. The argument for $(\pricevdown_\aindex,\priceavdown_\aindex)$ is similar.

Sufficiency follows from showing that
\begin{align}\label{cond:2-step}
\forall\aindex\in\{1,\dots,\anum-1\}, \priceavup_\aindex\marginalv\geq\pricevup_\aindex\priorv,
\end{align}
the value of program \ref{eq:dualup} is nonnegative. The proof that if for all $j\in\{1,\dots,\anum-1\}$, $\priceavdown_\aindex\marginalv\geq\pricevdown_\aindex\priorv$, the value of program \ref{eq:dualdown} is nonnegative is analogous and omitted for brevity. 

In this proof, we use the following shorthand notation: we write $\diff(\aaction_j,\type)$ instead of $\diff(\aaction_{\aindex+1},\aaction_j,\type)$ and we let
\begin{align*}
h(\aaction,\type)=\pricea(\aaction)+\mult(\aaction)\diff(\aaction,\type).
\end{align*}
The proof proceeds in six steps. First, we show that for all $\aindex\in\{1,\dots,\anum-1\}$, $\priceavup_{\aindex}$ is the solution to program \ref{eq:dualup} given $\pricev=\pricevup_{\aindex}$. Consider the multiplier 
    \begin{align*}
\mult(\aaction_k)=\frac{\diff(\aaction_\aindex,\type_{\tnum})-\diff(\aaction_\aindex,\type_1)}{\diff(\aaction_k,\type_{\tnum})- \diff(\aaction_k,\type_1)}\mathbbm{1}_{k\leq j},
\end{align*}
    and note that it solves
\begin{align*}
    \min_{\mult\in\realsa_{\geq0}}\max_{\type\in\Types}\priceup_\aindex(\type)-\mult(\aaction_k)\diff(\aaction_k,\type),
\end{align*}
    the optimal value of which is exactly $\priceaup_\aindex(\aaction_k)$. 

Second, note that if a triple $(\pricev,\priceav,\multv)$ is feasible for \ref{eq:dualup} and is such that $\pricev=\pricevup_{\aindex}$ for some $\aindex\in\{1,\dots,\anum-1\}$, the first step and \autoref{cond:2-step} implies that $\priceav\marginalv\geq\pricev\priorv$. 

Third, note that if a triple $(\pricev,\priceav,\multv)$ is feasible for \ref{eq:dualup}, $\pricev$ is constant, then by feasibility we have $\pricea(\aaction)\geq \price(\type)-\mult(\aaction)d(a,\omega)$. Because $\mult(a)\ge 0$ and $d(a,\omega)<0$ for some $\omega$, we know that $\pricea(\aaction)\geq \max_\type \price(\type)$, which implies $\priceav\marginalv \ge \pricev\priorv$. 

Fourth, we note that if a triple $(\pricev,\priceav,\multv)$, feasible for \ref{eq:dualup}, can be written as $(\pricev,\priceav,\multv)=\sum_{k=1}^K \alpha_k (\pricev_k,\priceav_k,\multv_k)$, with (i) $\alpha_k\geq0$, (ii) $(\pricev_k,\priceav_k,\multv_k)$ feasible for \ref{eq:dualup}, and (iii) either $\pricev_k = \pricevup_{\aindex}$ for some $\aindex$, or $\pricev_k$ is a constant vector, then the second and third steps imply that $\priceav\marginalv\geq\pricev\priorv$. 

Fifth, note that we only need to show $\priceav\marginalv\geq\pricev\priorv$ for triples $(\pricev,\priceav,\multv)$ that are feasible for \ref{eq:dualup}, where $\pricev$ takes the form of \autoref{eq:p-binds}, i.e., 
\begin{align}\label{eq:2stepnec}
    \price(\type)=\min_{\aaction\in\Actions} h(\aaction,\type).
\end{align}
Finally, consider now a triple $(\pricev,\priceav,\multv)$, such that $\pricev$ is as in \autoref{eq:2stepnec}. We show that it can be written as in the fourth step. We show this by induction on  the number of jump points of $\pricev$, denoted by $n(\pricev):=|\{\type_\tindex:\price(\type_{\tindex+1})>\price(\type_{\tindex})\}|$. If $n(\pricev)=0$, $\pricev$ is constant, and the result follows from the fourth step.

Now suppose $(\pricev,\priceav,\multv)$ can be written as in the fourth step whenever $n(\pricev)\le N-1$ for some $N\ge 1$, we show that is also true when $n(\pricev)=N$. Because $\pricev$ is not constant, we can denote by $\candi$ the smallest state index \tindex\ such that $\price(\type_\tindex)<\price(\type_{i+1})$. Because $\pricev$ satisfies \autoref{eq:2stepnec}, an action $\aaction^\star$ exists such that $\mult(\aaction^\star)>0$, $\diff(\aaction^\star,\type_{\candi})<\diff(\aaction^\star,\type_{\candi+1})$ (i.e., $\candi(\aaction^\star)=\candi$), and $\price(\type_{\candi})=h(\aaction^\star,\type_{\candi})$.

Define 
\begin{align*}
    \underline{\mult}&:=\frac{\price(\type_{i^\star+1})-\price(\type_{i^\star})}{\diff(\aaction^\star,\type_{\tnum})-\diff(\aaction^\star,\type_{1})},
\intertext{and}
   \mult_1(\aaction_\aindex)&:=\frac{[\price(\type_{\candi+1})-h(\aaction_\aindex,{\color{black}\type_{1}})]^+}{\diff(\aaction_\aindex,\type_\tnum)-\diff(\aaction_\aindex,\type_{1})}\mathbbm{1}_{\{\candi(\aindex)\leq \candi\}}.
\end{align*}
    Note $\mult_1(\aaction)\geq0$ and $\mult_1(\aaction^\star)=\underline{\mult}>0$. Moreover, $\mult_1(\aaction)\le \mult(\aaction)$ because $\price(\type_{\candi+1})\le h(\aaction_\aindex,\type_{\candi+1})$.
    
    Consider $\tripleo$, where 
    \begin{align*}
    \pricea_1(\aaction_\aindex)=\underline{\mult}\diff(\aaction^\star,\type_\tnum)-\mult_1(\aaction_\aindex)\diff(\aaction_\aindex,\type_\tnum),
\end{align*}
and $\price_1(\type)=\min_{\aaction\in\Actions}\mult_1(\aaction)\diff(\aaction,\type)+\pricea_1(\aaction)$.    

We show that $\price_1(\type)=\underline{\mult}\diff(\aaction^\star,\type)$. We can write $\price_1$ as follows:
\begin{align}\label{eq:price-2step}
\price_1(\type)&=\min\left\{\underline{\mult}\diff(\aaction^\star,\type_\tnum),\min_{\aindex:\candi(\aindex)\leq\candi}\pricea_1(\aaction_\aindex)+\mult_1(\aaction_\aindex)\diff(\aaction_\aindex,\type)\right\}\nonumber\\
&=\min\left\{\underline{\mult}\diff(\aaction^\star,\type_\tnum),\min_{j:\candi(j)\leq\candi}\underline{\mult}\diff(\aaction^\star,\type_\tnum)-\mult_1(\aaction_j)\diff(\aaction_j,\type_\tnum)+\mult_1(\aaction_j)\diff(\aaction_j,\type)\right\}.
\end{align}
For $\type>\candti$, $\diff(\aaction_j,\type)=\diff(\aaction_j,\type_\tnum)$ for $j$ such that $\candi(\aindex)\leq\candi$. Thus, $p_1(\type)=\underline{\mult}\diff(\aaction^\star,\type_\tnum)$. 

Consider now $\type\leq\candti$ and note that at $\aaction^\star$, $\pricea_1(\aaction^\star)+\mult_1(\aaction^\star)\diff(\aaction^\star,\type)=\underline{\mult}\diff(\aaction^\star,\type_1)$. Furthermore, the second term in the minimizer in \autoref{eq:price-2step} can be written as follows:
\begin{align*}
&\underline{\mult}\diff(\aaction^\star,\type_\tnum)-\mult(\aaction_j)\left(\diff(\aaction_j,\type_\tnum)-\diff(\aaction_j,\type_1)\right)=\underline{\mult}\diff(\aaction^\star,\type_\tnum)-[\price(\type_{\candi+1})-h(\aaction_\aindex,\type_1)]^+\geq\underline{\mult}\diff(\aaction^\star,\type_1),
\end{align*}
because $\underline{\mult}(\diff(\aaction^\star,\type_\tnum)-\diff(\aaction^\star,\type_1))=\price(\type_{\candi+1})-\price(\type_{\candi})\geq[\price(\type_{\candi+1})-h(\aaction_\aindex,\type_1)]^+$, where the second inequality follows because \pricev\ is constant for $\type\leq \candti$. Hence, for $\type\leq\candti$, $\price_1(\type)=\underline{\mult}\diff(\aaction^\star,\type_1)$.


Consider now the feasible triple \triplet, where $\pricea_2(\aaction)=\pricea(\aaction)-\pricea_1(\aaction)$, $\mult_2(\aaction)=\mult(\aaction)-\mult_1(\aaction)$, and $\price_2(\type)=\min_{\aaction\in\Actions}\mult_2(\aaction)\diff(\aaction,\type)+\pricea_2(\aaction)$. Note
\begin{align*}
    \price_2(\type)&=\min\left\{\min_{\aindex:\candi(\aindex)>\candi}h(\aaction_\aindex,\type),
    \min_{\aindex:\candi(\aindex)\leq \candi}\{h(\aaction_\aindex,\type)+\mult_1(\aaction_\aindex)[\diff(\aaction_\aindex,\type_\tnum)-\diff(\aaction_\aindex,\type)]\}\right\}-\underline{\mult}\diff(\aaction^\star,\type_\tnum).
\end{align*}
When $\type>\type_{\candi}$, $\diff(\aaction_\aindex,\type_\tnum)-\diff(\aaction_\aindex,\type)=0$ for all $\aindex$ such that $\candi(\aindex)\leq \candi$; hence, we have $\price_2(\type)=\price(\type)-\underline{\mult}\diff(\aaction^\star,\type_\tnum)=\price(\type)-\price_1(\type)$. 
    
Instead, when $\type\leq\type_{\candi}$, we have that for all $\aindex$ such that $\candi(\aindex)\leq\candi$, 
{\color{black}
\begin{align*}
    h(\aaction_\aindex,\type)+\mult_1(\aaction_\aindex)[\diff(\aaction_\aindex,\type_\tnum)-\diff(\aaction_\aindex,\type)]
    \geq\price(\type_{\candi+1})
    =h(\aaction^\star,\type_{\candi})+\mult_1(\aaction^\star)[\diff(\aaction^\star,\type_\tnum)-\diff(\aaction^\star,\type)].
\end{align*}
The inequality follows from $h(\aaction_\aindex,\type)=h(\aaction_\aindex,\type_{\candi+1})\geq\price(\type_{\candi+1})$ when $\type>\type_{\candi(\aindex)}$, and from $h(\aaction_\aindex,\type)+\mult_1(\aaction_\aindex)[\diff(\aaction_\aindex,\type_\tnum)-\diff(\aaction_\aindex,\type)]=h(\aaction_\aindex,\type_1)+[\price(\type_{\candi+1})-h(\aaction_\aindex,\type_1)]^+\geq\price(\type_{\candi+1})$ when $\type\leq\type_{\candi(\aindex)}$, and $p(\type_{\candi+1})>h(\aaction_j,\type_1)$ implies $h(\aaction_j,\type_I)>h(\aaction_j,\type_1)$.} Hence, 
\begin{align*}
    \min_{\aindex:\candi(\aindex)\leq \candi}\{h(\aaction_\aindex,\type)+\mult_1(\aaction_\aindex)[\diff(\aaction_\aindex,\type_\tnum)-\diff(\aaction_\aindex,\type)]\}
    &=\price(\type_{\candi+1})\le \min_{\aindex:\candi(\aindex)>\candi}h(\aaction_\aindex,\type).
\end{align*}
Consequently, by definition of $\underline{\mult}$, for all $\type\leq\type_{\candi}$, we have that
\begin{align*}
    \price_1(\type)+\price_2(\type)=\underline{\mult}\diff(\aaction^\star,\type_{1})+\price(\type_{\candi+1})-\underline{\mult}\diff(\aaction^\star,\type_\tnum)=\price(\type_{\candi})=\price(\type).
\end{align*}
In conclusion, $\triple=\tripleo+\triplet$. Moreover, notice that by construction,  $\price_2(\type_{\candi+1})=\price(\type_{\candi+1})-\price_1(\type_{\candi+1})=\price(\type)-\price_1(\type)=\price_2(\type)$ for all $\type\leq\type_{\candi}$. Hence, $n(\pricev_2)= n(\pricev)-1=N-1$. By the inductive hypothesis, we conclude that \triplet\ can be written as in the fourth step. Therefore, $(\pricev,\priceav,\multv)$ can be written as in the fourth step.
\end{proof}

\begin{proof}[Proof of \autoref{corollary:bin}]
By \autoref{proposition:aud}, \pair\ is \bce-consistent given \payoff\ if and only if for all $\candt\in\Types$,
\begin{align}
\marginal(\aaction_1)\min\{0,\diff(\candt)\}+\marginal(\aaction_2)\diff(\candt)\geq\sum_{\type\in\Types} \min\{\diff(\type),\diff(\candt)\}\prior(\type),\label{eq:plus}
\intertext{and}
-\marginal(\aaction_1)\diff(\candt)+\marginal(\aaction_2)\min\{0,-\diff(\candt)\}\geq\sum_{\type\in\Types} \min\{-\diff(\type),-\diff(\candt)\}\prior(\type).\label{eq:minus}
\end{align}
The conditions in \autoref{eq:plus} are trivial whenever $\diff(\candt)\leq0$. Instead, when $\diff(\candt)>0$, we have
\begin{align*}
\marginal(\aaction_2)\diff(\candt)\geq\diff(\candt)+\sum_{\type\leq\candt}\prior(\type)\left[\diff(\type)-\diff(\candt)\right],
\end{align*}
or, equivalently,
\begin{align}\label{eq:upper-1}
\marginal(\aaction_1)\leq\sum_{\type\leq\candt}\prior(\type)\left[1-\frac{\diff(\type)}{\diff(\candt)}\right].
\end{align}
Similarly, the conditions in \autoref{eq:minus} are trivial whenever $\diff(\candt)\geq0$. Instead, when $\diff(\candt)<0$, we obtain
\begin{align}\label{eq:upper-2}
\marginal(\aaction_2)\leq\sum_{\type\geq\candt}\prior(\type)\left[1-\frac{\diff(\type)}{\diff(\candt)}\right].
\end{align}
We now establish that the right-hand side of \autoref{eq:upper-1} is U-shaped in (the index of) \candt. Let $\type_i,\type_{i+1}\in\Types_+$ denote two adjacent states and consider the difference:
\begin{align*}
\sum_{l=1}^{i+1}\prior(\type_l)\left(1-\frac{\diff(\type_l)}{\diff(\type_{i+1})}\right)-\sum_{l=1}^{i}\prior(\type_l)\left(1-\frac{\diff(\type_l)}{\diff(\type_{i})}\right)=\left(\frac{\diff(\type_{i+1})-\diff(\type_i)}{\diff(\type_{i+1})\diff(\type_i)}\right)\sum_{l=1}^{i}\prior(\type_l)\diff(\type_l).
\end{align*}
The above difference is first negative when $i$ is small. As long as $\aaction_1$ is not optimal at the prior, a state $\type_{\aaction_1}$ exists such that the difference turns positive, which we defined prior to the statement of \autoref{corollary:bin}. 

Similarly, let $\type_i,\type_{i+1}\in\Types_-$ and consider the difference
\begin{align*}
\sum_{l=i+1}^{\tnum}\prior(\type_l)\left(1-\frac{\diff(\type_l)}{\diff(\type_{i+1})}\right)-\sum_{l=i}^{\tnum}\prior(\type_l)\left(1-\frac{\diff(\type_l)}{\diff(\type_{i})}\right)=\left(\frac{\diff(\type_{i+1})-\diff(\type_i)}{\diff(\type_i)\diff(\type_{i+1})}\right)\sum_{l=i+1}^{\tnum}\prior(\type_l)\diff(\type_l),
\end{align*}
which is initially positive for \tindex\ close to \tnum. If $\aaction_2$ is not optimal at the prior, a state $\type_{\aaction_2}$ exists such that the difference turns negative and corresponds to the one we defined prior to \autoref{corollary:bin}. This completes the proof.
\end{proof}
\subsection{Omitted proofs from \autoref{sec:applications}}\label{appendix:applications}
\begin{proof}[Proof of \autoref{proposition:aud-cs-mps}] Suppose \pair\ is \bce-consistent given \payoff, and let \priorb\ denote a \diff-mean-preserving spread of \prior. Fix $\candt\in\Types$ and consider
\begin{align*}
\sum_{\type\in\Types}\prior(\type)\min\{\diff(\type),\diff(\candt)\}&=\sum_{\diff\leq\diff(\candt)}\prior(\{\type:\diff(\type)=\diff\})\min\{\diff,\diff(\candt)\}\\&\geq\sum_{\diff\leq\diff(\candt)}\priorb(\{\type:\diff(\type)=\diff\})\min\{\diff,\diff(\candt)\},
\end{align*}
where the inequality follows from $\min\{\cdot,\diff(\candt)\}$ being concave and \priorb\ being a \diff-mean-preserving spread of \prior. The same logic implies 
\begin{align*}
\sum_{\type\in\Types}\prior(\type)\min\{-\diff(\type),-\diff(\candt)\}&=\sum_{-\diff\leq-\diff(\candt)}\prior(\{\type:\diff(\type)=\diff\})\min\{-\diff,-\diff(\candt)\}\\&\geq\sum_{\diff\leq\diff(\candt)}\priorb(\{\type:\diff(\type)=\diff\})\min\{-\diff,-\diff(\candt)\}.
\end{align*}
We conclude that $(\priorb,\marginal)$ is also \bce-consistent given \payoff. 
\end{proof}

\begin{proof}[Proof of \autoref{proposition:aud-cs-shift}]
Parametrizing the threshold states $\type_{\aaction_1}(\prior)$ and $\type_{\aaction_2}(\prior)$ by \param, note both are decreasing in \param. Thus, the lower and upper bounds on $\marginal(\aaction_2)$ increase.
\end{proof}

\begin{proof}[Proof of \autoref{proposition:aud-cs-ratio}] This result follows immediately from the expressions in \autoref{corollary:bin}.
\end{proof}
\subsection{Omitted proofs from \autoref{sec:core}}\label{appendix:ext} 

The problem in \cite{gale1957theorem} can be described as follows. We are given a graph $\graph=(\nodes,\edges)$ and a function $\demandbase:\nodes\rightarrow\reals$. If $\demandbase(\node)>0$, then $\demandbase(\node)$ is the demand of node \node\ for a given good; instead, if $\demandbase(\node)<0$, then $|\demandbase(\node)|$ is the supply of node \node\ for that good. We are also given a non-negative function \capacity\ defined on the edges $(\node,\node^\prime)\in\edges$ of the graph. A flow on $(\graph,\capacity)$ is a mapping $\flow:\nodes\times\nodes\rightarrow\reals$ such that for all $(\node,\node^\prime)\in\edges$, the following hold: (i) $\flow(\node,\node^\prime)+\flow(\node^\prime,\node)=0$, and (ii) $\flow(\node,\node^\prime)\leq\capacity(\node,\node^\prime)$. Demand \demandbase\ is feasible on $(\graph,\capacity)$ if a flow \flow\ exists such that for all $\node\in\nodes$, $\sum_{\node^\prime\in\nodes}\flow(\node^\prime,\node)\geq\demandbase(\node)$. That is, the demand \demandbase\ is feasible if a flow exists such that for every node the flow into that node is at least that node's demand. \cite{gale1957theorem} characterizes when a given demand \demandbase\ is feasible on $(\graph,\capacity)$.

Given a Bayes plausible distribution over posteriors \bsplit, we can define the graph $\graphp(\bsplit)$ as we did in the main text. Namely, the nodes are $\nodes=\Actions\cup\mathrm{supp}\,\bsplit$ and an edge $(\belief,\aaction)$ exists if and only if $\aaction\in\aaction^*(\belief)$. Given \bsplit\ and \marginal, define the demand \demand\ as follows: for each $\belief\in\mathrm{supp}\,\bsplit$, $\demand(\belief)=-\bsplit(\belief)$, and for each $\aaction\in\Actions$, $\demand(\aaction)=\marginal(\aaction)$. Finally, for any edge $\left( \belief ,\aaction\right) \in E$, the edge's flow capacity is given by $c\left( \belief ,\aaction\right) =\infty$ and $c\left( a,\belief \right) =0$.

\autoref{proposition:equivalence of representation} motivates the connection between our problem and that in \cite{gale1957theorem}. 
\begin{proposition}[Feasibility and \bce-consistency]\label{proposition:equivalence of representation}
The Bayes plausible distribution over posteriors \bsplit\ implements \marginal\ if and only if \demand\ is feasible on $\graphp(\bsplit)$.   
\end{proposition}

The proof of \autoref{proposition:equivalence of representation} relies on the following lemma:
\begin{lemma}[Market clearing]\label{lemma:exact satiation}
If $\demand$ is feasible on $\graphp(\bsplit)$, 
the flow out of any supply node $\belief \in \mathrm{supp}\,\bsplit$ is exactly $\bsplit \left(\belief\right) $ (and not less), and the flow into any demand node $\aaction\in\Actions$ is
exactly $\marginal\left(\aaction\right)$ (and not more).
\end{lemma}

\begin{proof}[Proof of \autoref{lemma:exact satiation}]
Suppose $\demand$ is feasible, and let \flow\ denote the corresponding flow on $(\graphp,\capacity)$. We show the
flow into any demand node $\aaction\in\Actions$ is exactly $\marginal\left( \aaction\right) $.
Towards a contradiction, suppose $\sum_{\belief \in \mathrm{supp}\,\bsplit}f\left( \belief
,\aaction\right) \geq \marginal\left( \aaction\right) $ for all $\aaction\in\Actions$, with  strict
inequality for some $a$. Summing over all actions on both sides of the
inequality yields \[\sum_{\aaction\in\Actions}\sum_{\belief \in \mathrm{supp}\,\bsplit}f\left( \belief ,\aaction\right)
>\sum_{\aaction\in\Actions}\marginal\left( \aaction\right) =1.\] Because $
\demand$ is feasible, the flow out of each $
\belief \in \mathrm{supp}\,\bsplit$ is at most $\bsplit \left( \belief \right) $; therefore for all $\belief \in \mathrm{supp}\,\bsplit$,
 \[\sum_{\aaction\in
A}f\left( \belief ,\aaction\right) \leq \bsplit \left( \belief \right) .\]
Summing again over all actions on both sides yields \[\sum_{\belief \in
\mathrm{supp}\,\bsplit}\sum_{\aaction\in\Actions}f\left( \belief ,\aaction\right) \leq \sum_{\belief \in \mathrm{supp}\,\bsplit}\bsplit \left( \belief
\right) =1,\]  a contradiction. The proof that the flow out of any
supply node $\belief $ is exactly $\bsplit \left( \belief \right) $ is analogous and
hence omitted.
\end{proof}

\begin{proof}[Proof of \autoref{proposition:equivalence of representation}]
Suppose first the Bayes plausible distribution over posteriors \bsplit\ is such that \demand\ is feasible on $\graphp(\bsplit)$, and let \flow\ denote the feasible flow. Consider a decision rule $\strat:\Posteriors\to\Delta(\Actions)$ such that the agent takes action $\aaction\in\Actions$ when the belief is $\belief \in \mathrm{supp}\,\bsplit$ with probability $\strat\left(\aaction|\belief \right) =\flow( \belief ,\aaction) /\bsplit( \belief) $. This defines a decision rule because
\begin{align*}
\sum_{\aaction\in\Actions}\strat \left( \aaction |\belief \right) =\frac{\sum_{\aaction\in\Actions}f\left( \belief
,\aaction\right) }{\bsplit \left( \belief \right) }=1,
\end{align*}
where the second equality is implied by \autoref{lemma:exact satiation}. Furthermore, \strat\ is optimal for the agent because $(\belief,\aaction)\in \edges$ only if $\aaction\in\aaction^{\ast }\left( \belief\right) $.

To verify that $(\bsplit,\strat)$ induce \marginal, note that for all $\aaction\in\Actions$,
\begin{align*}
\sum_{\belief \in \mathrm{supp}\,\bsplit}\bsplit \left( \belief \right) \strat( \aaction
|\belief) =\sum_{\belief \in \mathrm{supp}\,\bsplit}\flow( \belief ,\aaction) =\marginal(\aaction),
\end{align*}
where the second equality follows again from \autoref{lemma:exact satiation}. Thus, \bsplit\ implements \marginal.

Conversely, suppose $\pair$ are \bce-consistent given \payoff. Then, a Bayes plausible distribution over posteriors \bsplit\ and a decision rule \strat\ exist that implement \marginal.\footnote{Namely, \bce-consistency implies the existence of an obedient experiment from which we can infer the following distribution over posteriors. First, let
\begin{align*}
    \belief_a(\type)=\frac{\prior(\type)\joint(\aaction|\type)}{\sum_{\typeb\in\Types}\prior(\typeb)\joint(\aaction|\typeb)},
\end{align*}
and let $\bsplit(\{\belief_a\})=\sum_{\type\in\Types}\prior(\type)\joint(\aaction|\type)$. The decision rule $\strat(\cdot|\belief_a)=\mathbbm{1}[\aactionb=\aaction]$ completes the construction.} Define the graph $\graphp(\bsplit)$ and the demand \demand. Note the demand \demand\ is feasible on $\graphp(\bsplit)$ by defining the flow $\flow( \belief ,\aaction) =\strat( a |\belief ) \bsplit( \belief) 
$ for all $\left( \belief ,\aaction\right) \in \mathrm{supp}\,\bsplit\times\Actions$.
\end{proof}

\autoref{proposition:equivalence of representation} implies that verifying that \bsplit\ implements \marginal\ is equivalent to verifying the feasibility of the demand $\demand$ on graph $\graphp$. The main theorem in \cite{gale1957theorem} provides necessary and sufficient conditions under which $\demand$ is feasible. Adapted to our setting, the
conditions in \cite{gale1957theorem} can be stated as follows:
\begin{proposition}[\citealp{gale1957theorem}]\label{proposition: Gale 1957}The demand $\demand$ is feasible on graph $\graphp(\bsplit)$ if
and only if for every set $\Actionsb\subseteq\Actions$ a flow $f_{\Actionsb}$ exists such that the following hold:
\begin{enumerate}
\item\label{itm:gale-1} For all $\belief\in\mathrm{supp}\;\bsplit$, $\sum_{\aaction\in\Actions}\flow_{\Actionsb}( \belief ,\aaction) \leq \bsplit( \belief )$ and
\item\label{itm:gale-2} $\sum_{\aaction\in\Actionsb}\sum_{\belief \in \mathrm{supp}\,\bsplit}\flow_{\Actionsb}\left( \belief ,\aaction\right) \geq
\sum_{\aaction\in\Actionsb}\marginal( \aaction) $.
\end{enumerate}
\end{proposition}

Because of \autoref{lemma:exact satiation}, given a set $\Actionsb\subseteq\Actions$, items \ref{itm:gale-1} and \ref{itm:gale-2} in \autoref{proposition: Gale 1957} are satisfied for some flow $f_{\Actionsb}$ if and only if they are satisfied when the out flow from every supply node that is connected to nodes in $\Actionsb$ is maximal. In a slight abuse of notation, let $\Delta^*(\Actionsb)$ denote the set of beliefs for which \emph{some} action in \Actionsb\ is optimal; that is, $\Delta^*(\Actionsb)=\bigcup\{\opta:\aaction\in\Actionsb\}$. Note that in $\graphp(\bsplit)$, all beliefs in $\Delta^*(\Actionsb)$ are connected to actions in \Actionsb, and these beliefs are the only ones connected to actions in \Actionsb. \autoref{proposition: 
Gale 1957} implies \bsplit\ implements marginal if and only if 
\begin{align*}
(\forall\Actionsb\subseteq\Actions)\sum_{\belief\in\Delta^*(\Actionsb)}\bsplit(\belief)\geq\sum_{\aaction\in\Actionsb}\marginal(\aaction).
\end{align*}
Note $\left(\bigcup_{\aaction\in\Actionsb}\opta\right)^\complement=\bigcap_{\aaction\in\Actionsb}\left(\opta\right)^\complement$. Moreover, the latter set equals $\{\belief\in\Posteriors:(\exists\Actionsc\subseteq\Actionsb^\complement)\Actionsc=\aaction^*(\belief)\}$. We then conclude \bsplit\ implements \marginal\ if and only if for every subset $\Actionsb\subseteq\Actions$,
\begin{align*}
(\forall\Actionsb^\complement\subseteq\Actions)\sum_{\aaction\in\Actionsb^\complement}\marginal(\aaction)\geq\sum_{\Actionsc\subseteq\Actionsb^\complement}\bsplit_\Actions(\Actionsc),
\end{align*}
which is the statement in \autoref{proposition:core-bp}.

\section{\bce-consistency under the first-order approach}\label{sec:foa}
In this appendix, we show how our results extend under more general state and action spaces, provided that the \emph{first-order} approach applies. First, we introduce some technical assumptions we rely on for our results. Second, we set up the primal problem that verifies the \bce-consistency of \pair\ given \payoff\ and the corresponding dual. \autoref{lemma:dual-variables} delivers properties of the dual variables, which we use to extend the representation of the set \ppmpair\ to more general state and action spaces.  Finally, we show that \autoref{proposition:mcv} extends under the first-order approach.

\paragraph{Assumptions and notation} Assume \Actions\ and \Types\ are compact, Polish spaces, endowed with the Borel $\sigma$-algebra. In particular, we assume $\Types=[\underline{\type},\overline{\type}],\Actions=[\underline{\aaction},\overline{\aaction}]\subseteq\reals$. We endow product spaces with the product topology. It follows then that \Posteriors, $\Delta(\Actions)$, and $\Delta(\Actions\times\Types)$ are also compact, Polish spaces; endow them with the topology of weak convergence. Assume the \dm's utility function \payoff\ is such that the first-order approach applies. \cite{kolotilin2023persuasion} provide conditions on \payoff\ under which the first-order approach holds. To state them, we denote by $\diff(\aaction,\type)$ the partial derivative of \payoff\ with respect to \aaction. The three assumptions are as follows: (i) $\diff$ is three-times differentiable; (ii) \diff\ satisfies strict aggregate single-crossing; that is, $\int\diff(\aaction,\type)\belief(d\type)=0\Rightarrow\int\diff_1(\aaction,\type)\belief(d\type)<0$, where $\diff_1$ is the partial derivative of \diff\ with respect to \aaction; and (iii) $\min_{\type\in\Types}\diff(\underline{\aaction},\type)=\max_{\type\in\Types}\diff(\overline{\aaction},\type)=0$. These assumptions imply, among other things, that $\diff(\aaction,\cdot)$ is continuous; hence, \opta\ is compact as it is a closed subset of \Posteriors.

Under the assumptions above, we then have that \aaction\ is optimal at belief $\belief\in\Posteriors$ if and only if
\begin{align}\tag{FOC}\label{eq:foc}
    \int_{\Types}\diff(\aaction,\type)\belief(d\type)=0.
\end{align}
It follows that \opta\ is the set of all beliefs for which \ref{eq:foc} holds.

\paragraph{\bce-consistency} Let $\pair\in\Posteriors\times\Delta(\Actions)$. \pair\ is \bce-consistent given \payoff\ if and only if the value of the following program is $0$:
\begin{align}\label{eq:primal}\tag{P-C}
\sup_{\joint\in\Delta(\Actions\times\Types)}&0\\
&\text{s.t.} \left\{\begin{array}{ll}\int_{\Actions\times\tilde\Types}\mathrm{d}\joint(\aaction,\type)=\prior(\tilde\Types)&\text{for all measurable }\tilde\Types\subset\Types\\
\int_{\tilde\Actions\times\Types}\mathrm{d}\joint(\aaction,\type)=\marginal(\tilde\Actions)&\text{for all measurable }\tilde\Actions\subset\Actions\\
\int_{\tilde\Actions\times\Types}\diff(\aaction,\type)\mathrm{d}\joint(\aaction,\type)=0&\text{for all measurable }\tilde\Actions\subset\Actions\end{array}\right..\nonumber
\end{align}
The corresponding dual is as follows:\footnote{Because \prior\ and \marginal\ are probability distributions, we can state program \ref{eq:primal} as optimizing over the set of finitely additive signed measures of bounded variation on $\Actions\times\Types$, $ba(\Sigma_{\Actions\times\Types})$, without changing the value of the problem, where $\Sigma_{\Actions\times\Types}$ is the Borel $\sigma$-algebra on $\Actions\times\Types$. The corresponding dual is the set of measurable, bounded functions \cite[Theorem 14.4]{aliprantis2013infinite}. }
\begin{align}\label{eq:dual-c}\tag{D-C}
\inf_{\price\in B(\Types),\pricea\in B(\Actions),\mult\in B(\Actions)}&\int_{\Actions}\pricea(\aaction)\marginal(d\aaction)-\int_{\Types}\price(\type)\prior(d\type)\\\
\text{s.t. }&(\forall\aaction\in\Actions)(\forall\type\in\Types)\pricea(\aaction)\geq\price(\type)-\mult(\aaction)d(\aaction,\type),\nonumber
\end{align}
where $B(\Types)$ and $B(\Actions)$ denote the set of bounded measurable functions on \Types\ and \Actions, respectively. As in the main text, \pair\ is \bce-consistent if and only if the value of the dual is non-negative. We note that programs \ref{eq:primal} and \ref{eq:dual-c} are similar to the outcome-based primal and duals in \cite{kolotilin2023persuasion}, except ours have a marginal constraint on the set of actions. 

\autoref{lemma:dual-variables} shows that without loss of generality, we can take \price\ to be a continuous function on \Types, and can take \pricea\ to be an upper-semicontinuous function on \Actions, respectively. The continuity of \price\ follows from  \cite{kolotilin2018optimal} and \cite{kolotilin2023persuasion}. Instead, that \pricea\ is upper-semicontinuous requires proof and this result delivers \citet[Theorem 3]{strassen1965existence}.

\begin{lemma}\label{lemma:dual-variables}
It is without loss of generality to optimize over $\price\in C(\Types)$ and $\pricea\in USC(\Actions)$ in the dual program \ref{eq:dual-c}, where $USC(\Actions)$ is the set of upper-semicontinuous functions on \Actions.
\end{lemma}
\begin{proof}[Proof of \autoref{lemma:dual-variables}]
Let $(\price,\pricea,\mult)$ be feasible for \ref{eq:dual-c} and note $(\tilde{\price},\pricea,\mult)$, where
\begin{align}\label{eq:basic-p-id-c}
\tilde{\price}(\type)=\inf_{\aaction\in\Actions}\left[\pricea(\aaction)+\mult(\aaction)\diff(\aaction,\type)\right],
\end{align}
is feasible for \ref{eq:dual-c} and has weakly lower value. Moreover, continuity of $\diff(\aaction,\cdot)$ and compactness of \Actions\ implies $\tilde{\price}(\type)$ is continuous.

Similarly, let $(\price,\pricea,\mult)$ be feasible for \ref{eq:dual-c}, and note the triple $(\price,\tilde{\pricea},\mult)$, where 
\begin{align}\label{eq:pricea-c}
\tilde{\pricea}(\aaction)=\inf_{\multb(\aaction)}\sup_{\type\in\Types}\left[\price(\type)+\multb(\aaction)\diff(\aaction,\type)\right],
\end{align}
is feasible for \ref{eq:dual-c} and has weakly lower value. We now argue $\tilde\pricea$ in \autoref{eq:pricea-c} is upper-semicontinuous.  Note we can rewrite the right-hand side of \autoref{eq:pricea-c} as follows:
\begin{align}
\inf_{\multb(\aaction)}\sup_{\type\in\Types}\left[\price(\type)+\multb(\aaction)\diff(\aaction,\type)\right]=\inf_{\multb(\aaction)}\sup_{\belief\in\Posteriors}\int_{\Types}\left[\price(\type)+\multb(\aaction)\diff(\aaction,\type)\right]\mathrm{d}\belief(\type)=\sup_{\belief\in\opta}\int_{\Types}\price(\type)\mathrm{d}\belief(\type).
\end{align}
Note first we can replace $\sup$ with $\max$ on the right-hand side of the above expression, as we are maximizing a continuous function over a compact set. We already argued that \opta\ is compact. Furthermore, the correspondence $\aaction\rightrightarrows\opta$ is upper-hemicontinuous: let $\aaction_n\rightarrow\aaction^*$ and $\belief_n\in\opt(\aaction_n)$ such that $\belief_n\rightarrow\belief^*$. Note continuity of $\diff(\aaction^*,\cdot)$ implies
\begin{align*}
\int_\Types\diff(\aaction^*,\type)\dx\belief_n(\type)\rightarrow\int_{\Types}\diff(\aaction^*,\type)\dx\belief^*(\type),
\end{align*}
which is finite. \citet[Theorem 3.3]{serfozo1982convergence} implies
\begin{align*}
0=\int_\Types\diff(\aaction_n,\type)\dx\belief_n(\type)\rightarrow\int_{\Types}\diff(\aaction^*,\type)\dx\belief^*(\type).
\end{align*}
Hence, $\belief^*\in\opt(\aaction^*)$. \citet[Lemma 17.30]{aliprantis2013infinite} implies $\tilde\pricea$ is upper-semicontinuous and hence measurable.
\end{proof}
A consequence of \autoref{lemma:dual-variables} is the following. Define
\begin{align}
\Prices_{\foc}=\left\{\price\in C(\Types):(\exists \pricea\in USC(\Actions))(\exists \mult\in B(\Actions))\price(\type)=\inf_{\aaction\in\Actions}\pricea(\aaction)+\mult(\aaction)\diff(\aaction,\type)\right\}.
\end{align}
Then, \pair\ is \bce-consistent if and only if for all $\price\in\Prices_{\foc}$,
\begin{align}\label{eq:strassen}
\int_\Actions \left[\max_{\belief\in\opta}\int_\Types\price(\type)\belief(d\type)\right]\dx\marginal(\aaction)\geq\int_\Types\price(\type)\prior(d\type).
\end{align}
\autoref{eq:strassen} provides the analogue of \autoref{eq:support} under the first-order approach. \autoref{eq:strassen} also provides the analogue of \citet[Theorem 3]{strassen1965existence}. Although the result in \cite{strassen1965existence} requires checking \autoref{eq:strassen} holds for \emph{all} continuous functions, verifying \autoref{eq:strassen} holds on a smaller set of functions\textemdash namely, those in $\Prices_{\foc}$\textemdash suffices. 

\paragraph{Monotone and concave decision problems} Suppose now that, for all $\aaction\in\Actions$, $\diff(\aaction,\cdot)$ is increasing. Similar to \autoref{proposition:mcv}, define $\Prices_{\foc}^\uparrow$ and $\Prices_{\foc}^\downarrow$ to be the subsets of $\Prices_{\foc}$ in which $\mult(\aaction)\geq0$ and $\mult(\aaction)\leq0$, respectively.

\begin{lemma}\label{lemma:mcv-c}
Suppose that, for all $\aaction\in\Actions$, $\diff(\aaction,\cdot)$ is increasing. Then, \pair\ is \bce-consistent given \payoff\ if and only if the value of the dual is non-negative for all continuous functions $\price\in\Prices_{\foc}^\uparrow\cup\Prices_{\foc}^\downarrow$.
\end{lemma}
\begin{proof}
The construction is the same as in the proof of \autoref{proposition:mcv} in \autoref{appendix:main}, so we simply verify the objects in that proof are well defined under the assumptions in this section. 

Consider then a feasible triple $(\price,\pricea,\mult)$ such that \price\ satisfies \autoref{eq:basic-p-id-c}, but $\price\notin\Prices_{\foc}^\uparrow\cup\Prices_{\foc}^\downarrow$. Note \price\ is single-peaked and continuous, and \Types\ is compact, so \price\ attains a maximum $\price^\star$; let \candt\ denote the smallest such maximizer, which is well-defined by Berge's maximum theorem. 

By analogy to the proof of \autoref{proposition:mcv}, define $\priceup,\priceaup,\mult^\uparrow$ and $\pricedown,\priceadown,\mult^\downarrow$ as follows. First, $\mult^\uparrow=\max\{\mult,0\}$ and $\mult^\downarrow=\min\{\mult,0\}$, which are bounded and measurable because \mult\ is. Second, let
\begin{align*}
\priceup(\type)&=\price(\type)\mathbbm{1}[\type\leq\candt]+\price^\star\mathbbm{1}[\type>\candt],\\
\pricedown(\type)&=\price(\type)\mathbbm{1}[\type>\candt]+\price^\star\mathbbm{1}[\type\leq\candt],
\end{align*}
and note they are continuous because \price\ is and $\price(\candt)=\price^\star$. Finally, let
\begin{align*}
\priceaup(\aaction)&=\pricea(\aaction)\mathbbm{1}[\aaction\notin\Actions_\downarrow]+\price^\star\mathbbm{1}[\aaction\in\Actions_\downarrow]&&
\priceadown(\aaction)=\pricea(\aaction)\mathbbm{1}[\aaction\notin\Actions_\uparrow]+\price^\star\mathbbm{1}[\aaction\in\Actions_\uparrow],
\end{align*}
where $\Actions_\uparrow=\{\aaction\in\Actions:\mult(\aaction)>0\}$ and  $\Actions_\downarrow=\{\aaction\in\Actions:\mult(\aaction)<0\}$. Like in the proof of \autoref{proposition:mcv}, we have that $\priceup+\pricedown=\price+\price^\star$, $\priceaup+\priceadown=\pricea+\pricea^\star$. Thus, if we show 
\begin{align}\label{eq:show}
\priceaup(\aaction)=\max_{\belief\in\opta}\int\priceup(\type)\dx\belief(\type)\text{ and }\priceadown(\aaction)=\max_{\belief\in\opta}\int\pricedown(\type)\dx\belief(\type), 
\end{align}
we conclude that \pair\ is \bce-consistent given \payoff\ if and only if \autoref{eq:strassen} holds for $\price\in\Prices_{\foc}^\uparrow\cup\Prices_{\foc}^\downarrow$. We show \autoref{eq:show} holds for \priceaup; the proof for \priceadown\ is analogous. For $\aaction\in\Actions\setminus\Actions_\downarrow$, because $\priceup\geq\price$,
\[
\inf_{\multb(\aaction)}\max_{\type\in\Types}\priceup(\type)-\multb(\aaction)\diff(\aaction,\type)\geq\inf_{\multb(\aaction)}\max_{\type\in\Types}\price(\type)-\multb(\aaction)\diff(\aaction,\type)=\pricea(\aaction)=\priceaup(\aaction),
\]
but $\priceaup(\aaction)$ is attainable by $\multb(\aaction)=\multp(\aaction)$. Hence, $\inf_{\multb(\aaction)}\max_{\type\in\Types}\priceup(\type)-\multb(\aaction)\diff(\aaction,\type)=\priceaup(\aaction)$. 

Then, for $\aaction\in\Actions_\downarrow$, 
\begin{align*}
\price^*\geq\inf_{\multb(\aaction)}\max_{\type\in\Types}\priceup(\type)-\multb(\aaction)\diff(\aaction,\type)&=\inf_{\multb(\aaction)\geq0}\max_{\type\in\Types}\priceup(\type)-\multb(\aaction)\diff(\aaction,\type)=\inf_{\multb(\aaction)\geq0}\max_{\type\leq\type^*}\priceup(\type)-\multb(\aaction)\diff(\aaction,\type)\\
&=\inf_{\multb(\aaction)\geq0}\max_{\type\leq\type^*}\price(\type)-\multb(\aaction)\diff(\aaction,\type)\geq\max_{\type\leq\type^*}\price(\type)=\price^*.
\end{align*}
where the first inequality follows from $\price^*$ being attainable by $\multb(\aaction)=0$, the first equality follows from $\diff(\aaction,\overline{\type})=\max_{\type\in\Types}\diff(\aaction,\type)\geq\max_{\type\in\Types}\diff(\overline{\aaction},\type)=0$; thus, $\priceup(\overline{\type})-\multb(\aaction)\diff(\aaction,\overline{\type})\geq\price^*$ for any $\multb(\aaction)<0$ but $\price^\star$ can be attained at $\multb(\aaction)=0$. The second and third equalities are due to $\priceup$ being constant for $\type\geq\candt$ and $\priceup(\type)=\price(\type)$ for $\type\leq\candt$. Finally, the last inequality follows from $\inf_{\multb(\aaction)}\max_{\type\in\Types}\price(\type)-\multb(\aaction)\diff(\aaction,\type)$ being achieved at $\mult(\aaction)<0$ and  this infimum being less than $\price^\star$ implying that $\price(\candt)-\mult(\aaction)\diff(\aaction,\candt)\leq\price^\star$; thus, $\diff(\aaction,\candt)\leq0$. It follows that $\diff(\aaction,\type)\leq0$ for all $\type\leq\type^*$, which implies the last inequality above. In sum,
\[\priceaup(\aaction)=\inf_{\multb(\aaction)}\max_{\type\in\Types}\priceup(\type)-\multb(\aaction)\diff(\aaction,\type)=\max_{\belief\in\opta}\int_\Types \priceup(\type)\dx\belief(\type).\]
\end{proof}

Assuming that $\diff(\aaction,\type)$ has affine utility differences, \autoref{lemma:mcv-c} allows us to extend \autoref{proposition:aud} to the case in which \Actions\ and \Types\ are compact Polish spaces. Intuitively, note that when $\diff(\aaction,\type)=\slope(\aaction) \tilde\diff(\type)+\ctt(\aaction)$ for $\slope(\cdot)>0$ and $\tilde\diff(\cdot)$ increasing, the functions $p$ in $\Prices_{\foc}$ obtain as the minimum over functions that are linear in $\tilde{\diff}$, and hence are concave in $\tilde{\diff}$. This is analogous to the result that in the case of quadratic loss checking that for all concave functions their expected value is larger under \marginal\ than under \prior\ suffices to check that \pair\ is \bce-consistent. \autoref{lemma:mcv-c} refines this result by showing that checking \autoref{eq:strassen} holds for concave increasing and concave decreasing functions of $\tilde{\diff}(\cdot)$ is sufficient to check that \pair\ is \bce-consistent.


\section{\bce-consistency in games}\label{sec:games}
We illustrate here how to apply our results to study \bce-consistency in simple multi-agent settings. To do so, we extend our notation to games and then we introduce the definition of \bce-consistency in this setting.

\paragraph{Base game} The base game is the tuple $\base=\langle\Types,(\Actionsi,\payoffi)_{\playerindex\in\setplayers}\rangle$, defined as follows. Let \setplayers=$\{1,\dots,\nplayers\}$ denote the set of players. For each player \playerindex, let \Actionsi\ denote their finite set of actions. In a slight abuse of notation, let $\Actions=\times_{\playerindex\in\setplayers}\Actionsi$ denote the set of action profiles. Player \playeri's payoffs,  $\payoffi(\aaction,\type)$, depend on the action profile, $\aaction\in\Actions$, and the state of the world, $\type\in\Types$. 
 
\paragraph{Bayes correlated equilibrium} An outcome distribution $\joint\in\Delta(\Actions\times\Types)$ is a \emph{Bayes correlated equilibrium} of base game, $\base$, if for all agents $\playerindex\in\setplayers$, actions $\actioni,\actionbi\in\Actionsi$, the analogue of \autoref{eq:obedience} holds; namely,
\begin{align*}
\sum_{(\actionmi,\type)\in\Actionsmi\times\Types}\joint(\actioni,\actionmi,\type)\left[\payoffi(\actioni,\actionmi,\type)-\payoffi(\actionbi,\actionmi,\type)\right]\geq0.
\end{align*}
Let \bcebase\ denote the set of Bayes correlated equilibria of base game \base. 

\paragraph{\bce-consistency in games} In multi-agent settings, the definition of \bce-consistency depends on the data the analyst has about play. We consider two kinds of information the analyst may have about the players' actions, which are equivalent in the single-agent setting. In the first case, the analyst is endowed with a distribution over action profiles, $\marginal\in\Delta(\Actions)$. In the second case, the analyst is endowed with a profile of action distributions, one for each player; that is, $\marginalp=(\marginalbase_{0,1},\dots,\marginalbase_{0,N})\in\times_{\playerindex\in\setplayers}\Delta(\Actionsi)$. 

\begin{definition}[\bce- and M-\bce-consistent marginals]\label{definition:bce-c-game} The pair $\pair$ is \emph{\bce-consistent} in base game \base\ if a Bayes correlated equilibrium  $\joint\in\bcebase$ exists such that $(\joint_\Types,\joint_\Actions)=\pair$. Similarly, the profile $(\prior,\marginalp)$ is \emph{M-\bce-consistent} in base game \base\ if a Bayes correlated equilibrium $\joint\in\bcebase$ exists such that for all players $\playerindex\in\setplayers$, $\joint_{\Actionsi}=\marginali$, and $\joint_\Types=\prior$.
\end{definition}

\autoref{sec:public} studies under what conditions a pair of marginal distributions $\pair$ can be rationalized by a \emph{public} information structure. \autoref{sec:ring} applies our results to characterize the set of M-\bce-consistent marginals in ring-network games. 

\subsection{Public \bce-consistency in private Bayesian persuasion}\label{sec:public}
In this section, we consider the \emph{private} Bayesian persuasion setting of \cite{arieli2019private} and \cite{arieli2021feasible}. We assume each player's utility function depends only on her own action and the state of the world. That is, for all players $\playerindex\in\{1,\dots,N\}$, all action profiles $(\actioni,\actionmi)\in\Actions$, and states of the world $\type\in\Types$,
\begin{align*}
\payoffi\left( \actioni,\actionmi, \type \right) =\tilde{\payoff}_\playerindex\left( \actioni, \type \right). 
\end{align*} 
The analyst, who knows the base game \base\ and the marginal distribution of play $\marginal\in\Delta(\Actions)$,  wants to ascertain whether the distribution of play \marginal\ can be rationalized by a \emph{public} information structure (i.e., the players publicly observe the realization of a common signal structure before play). 

As we show next, the results in Sections \ref{sec:main} and \ref{sec:utility} can be applied to address this question. In what follows, we rely on the following definition:

\begin{definition}[Public \bce-consistency]
The pair $\pair$ is \emph{public} \bce-consistent in game \base\ if a Bayes correlated equilibrium \joint\ exists with marginals \pair, whose information structure uses public signals alone.
\end{definition}

Consider now an auxiliary decision problem $\bar\decision=\langle\Types, \Actions,\bar u\rangle$, in which a \dm\ with payoff $\bar u(\aaction,\type)=\sum_{\playerindex=1}^\nplayers\tilde{\payoff}_\playerindex(\actioni,\type)$ chooses an action $\aaction\in\Actions=\times_{\playerindex\in\nplayers}\Actionsi$ under incomplete information about \type. \autoref{proposition:first-order} links public \bce-consistency in first-order Bayesian persuasion to \bce-consistency in the auxiliary decision problem $\bar\decision$.
\begin{proposition}\label{proposition:first-order}
The	pair $\pair$ is public \bce-consistent in game \base\ if and only if $\pair$ is \bce-consistent in decision problem $\bar\decision$.
\end{proposition}
Consequently, the results in Sections \ref{sec:main} and \ref{sec:utility} can be used to study public \bce-consistent marginals in private persuasion games. To see why \autoref{proposition:first-order} holds,  let $\bar\aaction^*(\belief)$ denote the set of actions the \dm\ with payoff $\bar\payoff$ finds optimal when their belief is \belief; similarly, let $\aaction_\playerindex^*(\belief)$ denote the set of actions player \playerindex\ finds optimal when their belief is \belief. 
 It is immediate that $\bar\aaction^*(\belief)=\times_{\playerindex\in\nplayers}\aaction_\playerindex^*(\belief)$. That is, the profile $\aaction=\left( \aaction_{1},\dots ,\aaction_\nplayers\right) \in \Actions$ is optimal for the agent with payoff $\bar\payoff$ if and only if for all players $\playerindex\in\setplayers$ action \actioni\ is optimal in the game \base. Moreover,  given a posterior belief $\belief$, any distribution of (optimal) action profiles  the agent with payoff $\bar\payoff$ can generate, can also be generated by the players using a public correlation device or by duplicating signal realizations, and vice versa.\footnote{For example, suppose that the signal realization $s$ induces the posterior belief $\belief$. Suppose also that under $\belief$, the agent with payoff $\bar\payoff$ selects the two optimal action profiles $\aaction,\aactionb\in \overline{\aaction}^*(\belief)$ with equal probability. The same distribution of actions can be generated by the players: Indeed, one can \textquotedblleft split\textquotedblright\ the signal $s$ into two new signals, $s^{\prime }$ and $s^{\prime \prime }$, such that both new signals induce the same posterior belief $\belief$, and each of them is sent with half the probability of the original signal $s$. If each agent acts according to her corresponding optimal action in the profiles $\aaction$ and $\aactionb$ whenever $s^{\prime }$ and $s^{\prime \prime }$ are realized, respectively, the distribution over action profiles will coincide with that of the agent with payoff $\bar\payoff$.} Clearly, the connection between rationalizable distributions over action profiles in game \base\ and decision problem $\bar\base$ holds because of the focus on public \bce.

\subsection{Ring-network games}\label{sec:ring}
We consider here ring-network games as in \cite{kneeland2015identifying}, extended to account for incomplete information. A ring-network game is a base game \base\  in which players' payoffs satisfy the following:
\begin{align}\label{eq:rn-game}\tag{RN-P}
    \payoff_1(\aaction,\type)&=\rnpayoff_1(\aaction_1,\type)\\
    (\forall \playerindex\geq2)\payoffi(\aaction,\type)&=\rnpayoff_\playerindex(\aaction_{\playerindex-1},\actioni).\nonumber
\end{align}
In words, player $1$ cares about their action and the state of the world, whereas for $\playerindex\geq2$ player \playerindex\ cares about their action and that of player $\playerindex-1$. Ring-network games are used in the experimental literature that measures players' higher-order beliefs to identify departures from Nash equilibrium. 

The analyst knows the ring-network base game and each player \playerindex's action distribution, $\marginalbase_{0,\playerindex}\in\Delta(\Actionsi)$. The analyst wants to ascertain whether $(\prior,\marginalp)$ is M-\bce-consistent. Relying on the ring-network structure, \autoref{proposition:ring-network} ties the set of  M-\bce-consistent marginals in \base\ to the set of \bce-consistent marginals in single-agent games:
\begin{proposition}[M-\bce-consistency in ring-network games]\label{proposition:ring-network}
    The profile $(\prior,\marginalp)$ is M-\bce-consistent for the ring-network game $\langle\Types,(\Actionsi,\rnpayoff_\playerindex)_{\playerindex=1}^\nplayers\rangle$ if and only if the following holds:
    \begin{enumerate}
        \item $(\prior,\marginalbase_{0,1})$ is \bce-consistent in base game $\langle\Types,(\Actions_1,\rnpayoff_1)\rangle$, 
        \item For all $\playerindex\geq2$, $(\marginalbase_{0,\playerindex-1},\marginalbase_{0,\playerindex})$ are \bce-consistent in base game $\langle\Actions_{\playerindex-1},(\Actionsi,\rnpayoff_\playerindex)\rangle$.
    \end{enumerate}
\end{proposition}
Similar to \autoref{proposition:first-order}, \autoref{proposition:ring-network} exploits the structure of the ring-network game to reduce it to a series of single-agent problems in which, except for player $1$, the states are given by the actions of the preceding player and the prior distribution over this state space by the marginal over actions of the preceding player.  Indeed, for $\playerindex\geq2$, \bce-consistency of $(\marginalbase_{0,\playerindex-1},\marginalbase_{0,\playerindex})$ implies an information structure exists that rationalizes player $\playerindex$'s choices as the outcome of some information structure under ``prior'' $\marginalbase_{0,\playerindex-1}$, whereas \bce-consistency of $(\marginalbase_{0,\playerindex-2},\marginalbase_{0,\playerindex-1})$\footnote{With the understanding that $\marginalbase_{0,0}=\prior$.} guarantees the ``prior'' $\marginalbase_{0,\playerindex-1}$ is consistent with player $\playerindex-1$ observing the outcome of some information structure given their belief $\marginalbase_{0,\playerindex-2}$. 

\begin{proof}[Proof of \autoref{proposition:ring-network}]
In the ring-network base game, for a joint distribution $\joint\in\Delta(\Actions\times\Types)$, the obedience constraints can be written as follows:
\begin{align*}
    (\forall\aaction_1,\aactionb_1\in\Actions_1)\sum_{\type\in\Types}\joint_{\Types\times\Actions_1}(\aaction_1,\type)\left(\rnpayoff(\aaction_1,\type)-\rnpayoff(\aactionb_1,\type)\right)&\geq0\\
    (\forall \playerindex\in\{2,\dots,\nplayers\})(\forall\actioni,\actionbi\in\Actionsi)\sum_{\aaction_{\playerindex-1}\in\Actions_{\playerindex-1}}\joint_{\Actions_{\playerindex-1},\Actions_\playerindex}(\aaction,\type)\left(\rnpayoff(\aaction_{\playerindex-1},\actioni)-\rnpayoff(\aaction_{\playerindex-1},\actionbi)\right)&\geq0,
\end{align*}
where $\joint_{\Types\times\Actions_1}$ is the marginal of \joint\ over $\Types\times\Actions_1$ and similarly for $\playerindex\geq2$, $\joint_{\Actions_{\playerindex-1}\times\Actionsi}$ is the marginal of \joint\ over $\Actions_{\playerindex-1}\times\Actionsi$. Thus, it is immediate that the conditions in \autoref{proposition:ring-network} are necessary for $(\prior,\marginalp)$ to be M-\bce-consistent. 

For sufficiency, note \autoref{theorem:h-representation} implies that under the conditions of \autoref{proposition:ring-network}, $\left(\joint_{\Types\times\Actions_1},\dots,\joint_{\Actions_{\nplayers-1}\times\Actions_{\nplayers}}\right)$ exist, each of which satisfies the respective marginal conditions and obedience constraints. 

Given these distributions, define $\hat{\joint}\in\Delta(\Actions\times\Types)$ as follows: for each $(\aaction,\type)\in\Actions\times\Types$,
\begin{align}
    \hat{\joint}(\aaction,\type)=\joint_{\Actions_1\times\Types}(\aaction_1,\type)\joint_{\Actions_1\times\Actions_2}(\aaction_2|\aaction_1)\times\dots\joint_{\Actions_{N-1}\times\Actions_N}(\aaction_N|\aaction_{N-1}),
\end{align}
where, abusing notation, we let for $\playerindex\geq2$, $\joint_{\Actions_{\playerindex-1}\times\Actionsi}(\cdot|\aaction_{\playerindex-1})$ denote the distribution $\joint_{\Actions_{\playerindex-1}\times\Actionsi}$ conditional on $\aactionb_{\playerindex-1}=\aaction_{\playerindex-1}$.

Note $\hat{\joint}(\aaction,\type)$ satisfies the obedience constraints of player $1$ if and only if $\joint_{\Actions_1\times\Types}(\cdot)$ does. Indeed, for all $\aaction_1,\aactionb_1$, we have
\begin{align}    &\sum_{\aaction_{-1},\type}\hat{\joint}(\aaction_1,\aaction_{-1},\type)\left(\rnpayoff_1(\aaction_1,\type)-\rnpayoff_1(\aactionb_1,\type)\right)=\nonumber\\
    &\sum_{\type}\joint_{\Actions_1\times\Types}(\aaction_1,\type)\left(\rnpayoff_1(\aaction_1,\type)-\rnpayoff_1(\aactionb_1,\type)\right)\sum_{(\aaction_2,\dots,\aaction_\nplayers)}\prod_{\playerindex=2}^\nplayers\joint_{\Actions_{\playerindex-1}\times\Actionsi}(\actioni|\aaction_{\playerindex-1})=\nonumber\\
    &\sum_{\type}\joint_{\Actions_1\times\Types}(\aaction_1,\type)\left(\rnpayoff_1(\aaction_1,\type)-\rnpayoff_1(\aactionb_1,\type)\right).
\end{align}
Now consider player $\playerindex\geq 2$. For simplicity, fix $\playerindex=2$\textemdash the proof for the remaining players follows immediately. Then, let $\aaction_2,\aactionb_2\in\Actions_2$. We check that $\hat\joint$ satisfies the obedience constraint of player $2$ if and only if $\joint_{\Actions_1\times\Actions_2}$ does:
\begin{align*}
    &\sum_{\aaction_{-2},\type}\hat{\joint}(\aaction_2,\aaction_{-2},\type)\left(\rnpayoff_2(\aaction_1,\aaction_2)-\rnpayoff_2(\aaction_1,\aactionb_2)\right)=\\
    &\sum_{\aaction_1,\type}\joint_{\Actions_1\times\Types}(\aaction_1,\type)\joint_{\Actions_1\times\Actions_2}(\aaction_2|\aaction_1)\left(\rnpayoff_2(\aaction_1,\aaction_2)-\rnpayoff_2(\aaction_1,\aactionb_2)\right)\sum_{(\aaction_3,\dots,\aaction_\nplayers)}\prod_{\playerindex=3}^\nplayers\joint_{\Actions_{\playerindex-1}\times\Actionsi}(\actioni|\aaction_{\playerindex-1})=\\
&\sum_{\aaction_1\in\Actions_1}\left(\sum_{\type}\joint_{\Actions_1\times\Types}(\aaction_1,\type)\right)\joint_{\Actions_1\times\Actions_2}(\aaction_2|\aaction_1)\left(\rnpayoff_2(\aaction_1,\aaction_2)-\rnpayoff_2(\aaction_1,\aactionb_2)\right)=\\
    &\sum_{\aaction_1\in\Actions_1}\marginalbase_{01}(\aaction_1)\joint_{\Actions_1\times\Actions_2}(\aaction_2|\aaction_1)\left(\rnpayoff_2(\aaction_1,\aaction_2)-\rnpayoff_2(\aaction_1,\aactionb_2)\right)=\\
    &\sum_{\aaction_1\in\Actions_1}\joint_{\Actions_1\times\Actions_2}(\aaction_1,\aaction_2)\left(\rnpayoff_2(\aaction_1,\aaction_2)-\rnpayoff_2(\aaction_1,\aactionb_2)\right),
\end{align*}
    where the third equality follows from the assumption that $\joint_{\Actions_1\times\Types}$ satisfies the marginal constraints for player $1$. It follows that $\hat\joint$ is obedient for player $2$ if and only if $\joint_{\Actions_1\times\Actions_2}$ is.
\end{proof}

\end{document}